\documentclass[11pt]{article}

\def\colorful{1}

\usepackage{amsthm,amsfonts,amsmath,amssymb,epsfig,color,float,graphicx,verbatim, enumitem}
\usepackage{multirow}

\newif\ifhyper\IfFileExists{hyperref.sty}{\hypertrue}{\hyperfalse}
\hypertrue
\ifhyper\usepackage{hyperref}\fi

\usepackage{enumitem}
\usepackage[capitalize]{cleveref} 

\makeatletter
\renewcommand{\section}{\@startsection{section}{1}{0pt}{-12pt}{5pt}{\large\bf}}
\renewcommand{\subsection}{\@startsection{subsection}{2}{0pt}{-12pt}{-5pt}{\normalsize\bf}}
\renewcommand{\subsubsection}{\@startsection{subsubsection}{3}{0pt}{-12pt}{-5pt}{\normalsize\bf}}
\makeatother

\usepackage{framed}
\usepackage{nicefrac}

\def\nnewcolor{1}
\ifnum\nnewcolor=1

\fi
\ifnum\nnewcolor=0

\fi

\ifnum\colorful=1

\else

\fi

\newtheorem{theorem}{Theorem}[section]

\newtheorem{lemma}[theorem]{Lemma}
\newtheorem{informal theorem}[theorem]{Theorem (informal statement)}
\newtheorem{condition}[theorem]{Condition}
\newtheorem{proposition}[theorem]{Proposition}
\newtheorem{corollary}[theorem]{Corollary}
\newtheorem{claim}[theorem]{Claim}

\newtheorem{remark}[theorem]{Remark}

\theoremstyle{definition}
\newtheorem{definition}[theorem]{Definition}

\newcommand{\R}{\mathbb{R}}
\newcommand{\C}{\mathbb{C}}
\newcommand{\Z}{\mathbb{Z}}
\newcommand{\E}{\mathbb{E}}

\newcommand{\poly}{\mathrm{poly}}
\newcommand{\polylog}{\mathrm{polylog}}
\newcommand{\cov}{\text{Cov}}

\newcommand{\p}{\mathbf{P}}
\newcommand{\q}{\mathbf{Q}}

\newcommand{\dtv}{d_{\mathrm TV}}

\newcommand{\wh}[1]{{\widehat{#1}}}

\newcommand{\var}{\text{Var}}

\newcommand{\ignore}[1]{}

\newcommand{\eps}{\epsilon}
\newcommand{\h}{\mathbf{H}}

\newcommand{\bone}{\mathbf{1}}

\newcommand{\Var}{\mathop{\textnormal{Var}}\nolimits}

\newcommand{\re}{\mathrm{Re}}
\newcommand{\im}{\mathrm{Im}}

\renewcommand{\eqref}[1]{Eq.~(\ref{#1})}

\newcommand{\eqdef}{\stackrel{{\mathrm {\footnotesize def}}}{=}}

\newcommand{\littlesum}{\mathop{\textstyle \sum}}

\newenvironment{algorithm}[1][\  ] %
{ \rm
\begin{tabbing}
....\=.....\=.....\=.....\=.....\=  \+ \kill
} %
{\end{tabbing} }

\floatstyle{ruled}
\newfloat{Algorithm}{H}{alg}
\newfloat{Subroutine}{H}{sub}

{
\begin{minipage}{1.0\linewidth} \begin{algorithm} %
} { \end{algorithm} \end{minipage} }

\title{The Fourier Transform of Poisson Multinomial Distributions\\
and its Algorithmic Applications 
}

\author{
Ilias Diakonikolas\thanks{Part of this work was performed while the author was at the University of Edinburgh. 
Supported in part by a Marie Curie Career Integration Grant and EPSRC grant EP/L021749/1.}\\
University of Southern California\\
{\tt  diakonik@usc.edu}\\
\and
Daniel M. Kane\thanks{Supported by NSF Award CCF-1553288. 
This work was initiated while visiting the University of Edinburgh.}\\
University of California, San Diego\\
{\tt dakane@cs.ucsd.edu}\\
\and
Alistair Stewart$^{\ast}$\\
University of Southern California\\
{\tt alistais@usc.edu}
}

\begin{document}

\maketitle

\thispagestyle{empty}

\vspace{-0.2cm}

\begin{abstract}
We study Poisson Multinomial Distributions -- a fundamental family of discrete distributions 
that generalize the binomial and multinomial distributions, and are commonly encountered
in computer science. 
Formally, an $(n, k)$-Poisson Multinomial Distribution (PMD) is a random variable
of the form $X = \sum_{i=1}^n X_i$, where the $X_i$'s are independent random vectors supported 
on the set  $\{e_1, e_2, \ldots, e_k \}$ of standard basis vectors in $\R^k$.
In this paper, we obtain a refined structural understanding of PMDs 
by analyzing their Fourier transform. 
As our core structural result, we prove that the Fourier transform of PMDs is {\em approximately sparse}, 
i.e., roughly speaking, its $L_1$-norm is small outside a small set. By building on this result, we obtain the following 
applications:

\smallskip

 {\bf Learning Theory.}
We design the first computationally efficient learning algorithm for PMDs
with respect to the total variation distance. Our algorithm learns an arbitrary $(n, k)$-PMD 
within variation distance $\eps$ using a near-optimal sample size of $\widetilde{O}_k(1/\eps^2),$ 
and runs in time $\widetilde{O}_k(1/\eps^2) \cdot \log n.$ Previously, no algorithm with a $\poly(1/\eps)$
runtime was known, even for $k=3.$

\smallskip

{\bf Game Theory.} We give the first efficient polynomial-time approximation scheme (EPTAS) for computing Nash equilibria
in anonymous games. For normalized anonymous games
with $n$ players and $k$ strategies, our algorithm computes a well-supported $\eps$-Nash equilibrium in time 
$n^{O(k^3)} \cdot (k/\eps)^{O(k^3\log(k/\eps)/\log\log(k/\eps))^{k-1}}.$ 
The best previous algorithm for this problem~\cite{DaskalakisP08, DaskalakisP2014} 
had running time $n^{(f(k)/\eps)^k},$ where $f(k) = \Omega(k^{k^2})$, for any $k>2.$

\smallskip

{\bf Statistics.} We prove a multivariate central limit theorem (CLT) that relates 
an arbitrary PMD to a discretized multivariate Gaussian with the same mean and covariance, in total variation distance. 
Our new CLT strengthens the CLT of Valiant and Valiant~\cite{VV10b, ValiantValiant:11} by completely removing the dependence on $n$ in the error bound.

\smallskip

Along the way we prove several new structural results of independent interest about PMDs. These include: (i) a robust moment-matching lemma, 
roughly stating that two PMDs that approximately agree on their low-degree parameter moments are close in variation distance; 
(ii) near-optimal size proper $\eps$-covers for PMDs in total variation distance (constructive upper bound and nearly-matching lower bound).
In addition to Fourier analysis, we employ a number of analytic tools, including the saddlepoint method from complex analysis,
that may find other applications.
\end{abstract}

\thispagestyle{empty}
\setcounter{page}{0}

\newpage

\section{Introduction}  \label{sec:intro}


\subsection{Background and Motivation} \label{ssec:motiv}

The Poisson Multinomial Distribution (PMD) is the discrete probability distribution
of a sum of mutually independent categorical random variables over the same sample space. A categorical
random variable ($k$-CRV) describes the result of a random event that takes on one of $k \ge 2$ possible outcomes.
Formally,
an $(n, k)$-PMD is any random variable of the form $X = \sum_{i=1}^n X_i,$
where the $X_i$'s are independent random vectors supported on the set
$\{e_1, e_2, \ldots, e_k \}$ of standard basis vectors in $\R^k$.

PMDs comprise a broad class of discrete distributions of fundamental importance in computer science, probability, and statistics.
A large body of work in the probability and statistics literature has been devoted to the study of the behavior
of PMDs under various structural conditions~\cite{Barbour88, Loh92, BHJ:92, Bentkus:03, Roos99, Roos10}.
PMDs generalize the familiar binomial and multinomial distributions, and describe many distributions
commonly encountered in computer science (see, e.g.,~\cite{DaskalakisP07, DaskalakisP08, Valiant08stoc, ValiantValiant:11}).
The $k=2$ case corresponds to the Poisson binomial distribution (PBD), introduced by Poisson~\cite{Poisson:37}
as a non-trivial generalization of the binomial distribution.

Recent years have witnessed a flurry of research activity on PMDs and related distributions,
from several perspectives of theoretical computer science,
including learning~\cite{DDS12stoc, DDOST13focs, DKS15, DKT15, DKS15b},
property testing~\cite{Valiant08stoc, VV10b, ValiantValiant:11},
 computational game theory~\cite{DaskalakisP07, DaskalakisP08, BorgsCIKMP08, DaskalakisP09, DaskalakisP2014, GT14},
 and derandomization~\cite{GMRZ11, BDS12, De15, GKM15}. More specifically,
 the following questions have been of interest to the TCS community:
 \begin{enumerate}
 \item Is there a statistically and computationally efficient algorithm for learning PMDs
 from independent samples in total variation distance?

 \item How fast can we compute approximate Nash equilibria in anonymous games with many players and
 a small number of strategies per player?

 \item How well can a PMD be approximated, in total variation distance, by a discretized Gaussian with the same mean and covariance matrix?
 \end{enumerate}
The first question is a fundamental problem in unsupervised learning that has received considerable recent
attention in TCS~\cite{DDS12stoc, DDOST13focs, DKS15, DKT15, DKS15b}.
The aforementioned works have studied the learnability of PMDs, and related distribution families,
in particular PBDs (i.e., $(n, 2)$-PMDs) and sums of independent integer random variables.
Prior to this work, no computationally efficient learning algorithm for PMDs was known, even for the case of $k=3.$

The second question concerns an important class of succinct games previously studied
in the economics literature~\cite{Milchtaich1996, Blonski1999, Blonski2005}, whose (exact) Nash equilibrium
computation was recently shown to be intractable~\cite{CDO15}.
The formal connection between computing Nash equilibria in these games and PMDs was established in a sequence of papers
by Daskalakis and Papadimitriou~\cite{DaskalakisP07, DaskalakisP08, DaskalakisP09, DaskalakisP2014}, who leveraged it to gave the first PTAS for the problem.
Prior to this work, no efficient PTAS was known, even for anonymous games with $3$ strategies per player.

The third question refers to the design of Central Limit Theorems (CLTs) for PMDs with respect to the total variation distance.
Despite substantial amount of work in probability theory, the first strong CLT of this form appears to have been shown by Valiant and Valiant
~\cite{VV10b, ValiantValiant:11}, motivated by applications in distribution property testing. 
\cite{VV10b, ValiantValiant:11} leveraged their CLT to obtain tight  lower bounds for several fundamental problems
in property testing. We remark that the error bound of the~\cite{VV10b} CLT has a logarithmic dependence 
on the size $n$ of the PMD (number of summands), and
it was conjectured in~\cite{VV10b} that this dependence is unnecessary.

\subsection{Our Results} \label{sec:results}
The main technical contribution of this work is the use of
Fourier analytic techniques to obtain a refined understanding
of the structure of PMDs. As our core structural result, we prove that the Fourier transform
of PMDs is {\em approximately sparse}, i.e., roughly speaking, its $L_1$-norm is small outside a small set.
By building on this property, we are able to obtain various new structural results about PMDs,
and make progress on the three questions stated in the previous subsection.
In this subsection, we describe our algorithmic and structural contributions in detail.

We start by stating our algorithmic results in learning and computational game theory, followed
by an informal description of our structural results and the connections between them.

\paragraph{Distribution Learning.} As our main learning result, we obtain the first statistically and computationally
efficient learning algorithm for PMDs with respect to the total variation distance. In particular, we show:

\begin{theorem}[Efficiently Learning PMDs] \label{thm:learn-pmd}
For all $n, k \in \Z_+$ and $\eps>0$, there is an algorithm for learning $(n, k)$-PMDs with the following
performance guarantee: Let $\p$ be an unknown $(n, k)$-PMD.
The algorithm uses $m = O\left(k^{4k}\log^{2k}(k/\eps)/\eps^2\right)$
samples from $\p$,
runs in time\footnote{We work
in the standard ``word RAM'' model in which basic arithmetic
operations on $O(\log n)$-bit integers are assumed to take constant time.}$ O\left(k^{6k}\log^{3k}(k/\eps)/\eps^2\right) \cdot \log n,$
and with probability at least $9/10$ outputs an $\eps$-sampler for $\p.$
\end{theorem}


We remark that our learning algorithm outputs a succinct description of its hypothesis $\h,$ via its Discrete Fourier Transform (DFT), $\wh{\h},$
which is supported on a small size set. We show that the DFT gives both an efficient $\eps$-sampler and an efficient $\eps$-evaluation oracle
for $\p.$

Our algorithm learns an unknown $(n, k)$-PMD within variation distance $\eps,$
with sample complexity $\widetilde{O}_k (1/\eps^2),$ and computational complexity $\widetilde{O}_k (1/\eps^2)  \cdot \log n.$
The sample complexity of our algorithm is near-optimal for any fixed $k$, as $\Omega(k/\eps^2)$ samples are necessary,
even for $n=1.$ We note that recent work by Daskalakis {\em et al.}~\cite{DKT15} established a similar sample upper bound,
however their algorithm is not computationally efficient. More specifically, it runs in time $(1/\eps)^{\Omega(k^{5k} \log^{k+1}(1/\eps))},$ 
which is quasi-polynomial in $1/\eps,$ even for $k=2.$ For the $k=2$ case, in recent work~\cite{DKS15} the authors of this paper gave an algorithm with sample complexity
and runtime $\widetilde{O}(1/\eps^2).$ Prior to this work, no algorithm with a $\poly(1/\eps)$ sample size
and runtime was known, even for $k=3.$

Our learning algorithm and its analysis are described in Section~\ref{sec:algo}.

\paragraph{Computational Game Theory.}
As our second algorithmic contribution, we give the first
efficient polynomial-time approximation scheme (EPTAS) for computing Nash equilibria 
in anonymous games with many players and a small number of strategies. 
In anonymous games, all players have the same set of strategies, and 
the payoff of a player depends on the strategy played by the player and 
the number of other players who play each of the strategies.
In particular, we show:

\begin{theorem}[EPTAS for Nash in Anonymous Games] \label{thm:nash-anonymous}
There is an EPTAS for the mixed Nash equilibrium problem for normalized anonymous
games with a constant number of strategies. More precisely, there exists an algorithm with the following performance guarantee:
for all $\eps>0$, and any normalized anonymous game $\cal G$ of $n$ players and $k$ strategies,
the algorithm runs in time $(kn)^{O(k^3)}(1/\eps)^{O(k^3\log(k/\eps)/\log\log(k/\eps))^{k-1}},$ 
and outputs a (well-supported) $\eps$-Nash equilibrium of $\cal G.$
\end{theorem}

The previous PTAS for this problem~\cite{DaskalakisP08, DaskalakisP2014} has running time
$n^{O(2^{k^2}  (f(k)/\eps)^{6k})},$ where $f(k) \leq 2^{3k-1} k^{k^2+1} k!.$
Our algorithm decouples the dependence on $n$ and $1/\eps,$
and, importantly, its running time dependence on $1/\eps$ is quasi-polynomial.
For $k=2,$ an algorithm with runtime $\poly(n) (1/\eps)^{O(\log^2(1/\eps))}$ was given in
~\cite{DaskalakisP09}, which was sharpened to  $\poly(n) (1/\eps)^{O(\log(1/\eps))}$
in the recent work of the authors~\cite{DKS15}. Hence,
we obtain, for any value of $k,$ the same qualitative runtime dependence  on $1/\eps$ as in the case $k=2.$

Similarly to~\cite{DaskalakisP08, DaskalakisP2014}, our algorithm proceeds by constructing a {\em proper} $\eps$-cover,
in total variation distance, for the space of PMDs. A proper $\eps$-cover for $\mathcal{M}_{n, k},$ the set of all
$(n, k)$-PMDs, is a subset $C$ of $\mathcal{M}_{n, k}$ such that any distribution in $\mathcal{M}_{n, k}$ is within total variation distance 
$\eps$ from some distribution in $C.$
Our main technical contribution is the efficient construction of a proper $\eps$-cover of near-minimum size (see Theorem~\ref{thm:cover-pmd}).
We note that, as follows from Theorem~\ref{thm:cover-lb}, the quasi-polynomial dependence on $1/\eps$ and
the doubly exponential dependence on $k$ in the runtime are unavoidable for {\em any} cover-based algorithm.
Our cover upper and lower bounds and our Nash approximation algorithm are given in Section~\ref{sec:cover-nash}.




\paragraph{Statistics.} Using our Fourier-based machinery, we prove a strong ``size-free'' CLT
relating the total variation distance between a PMD and an appropriately discretized Gaussian with the same mean
and covariance matrix. In particular, we show: 

\begin{theorem} \label{thm:clt}
Let $X$ be an $(n,k)$-PMD with covariance matrix $\Sigma.$
Suppose that $\Sigma$ has no eigenvectors other than $\bone =(1,1,\ldots,1)$ with eigenvalue less than $\sigma.$
Then, there exists a discrete Gaussian $G$ so that
$$
\dtv(X,G) \leq \poly(k)/\poly(\sigma).
$$
\end{theorem}

As mentioned above, Valiant and Valiant~\cite{VV10b, ValiantValiant:11} proved a CLT of this form and used it as their main
technical tool to obtain tight information-theoretic lower bounds for fundamental statistical estimation tasks.
This and related CLTs have since been used in proving lower bounds for other problems (see, e.g.,~\cite{ChenST14}).
The error bound in the CLT of~\cite{VV10b, ValiantValiant:11} is of the form
$\poly(k)/\poly(\sigma) \cdot (1+ \log n)^{2/3},$ i.e., it has a dependence on the size $n$ of the underlying PMD.
Our Theorem~\ref{thm:clt} provides a {\em qualitative} improvement over the aforementioned bound,
by establishing that {\em no} dependence on $n$ is necessary. We note that \cite{VV10b} conjectured that such a qualitative
improvement may be possible.

We remark that our techniques for proving Theorem~\ref{thm:clt} are orthogonal to those of~\cite{VV10b, ValiantValiant:11}.
While Valiant and Valiant use Stein's method, we prove our strengthened CLT using the Fourier techniques that underly this paper.
We view Fourier analysis as the right technical tool to analyze sums of independent random variables.
An additional ingredient that we require is the saddlepoint method from complex analysis. 
We hope that our new CLT will be of broader use as an analytic tool to the TCS community. 
Our CLT is proved in Section~\ref{sec:clt}.

\paragraph{Structure of PMDs.} We now provide a brief intuitive overview of our new structural results for PMDs,
the relation between them, and their connection to our algorithmic results mentioned above.
The unifying theme of our work is a refined analysis of the structure of PMDs, based on their Fourier transform.
The Fourier transform is one of the most natural technical tools to consider for analyzing sums of independent random variables,
and indeed one of the classical proofs of the (asymptotic) central limit theorem is based on Fourier methods.
The basis of our results, both algorithmic and structural, is the following statement:

\medskip

\noindent {\bf Informal Lemma} (Sparsity of the Fourier Transform of PMDs.)
{\em For any $(n, k)$-PMD $\p$, and any $\eps>0$ there exists a ``small'' set $T = T( \p, \eps),$
such that the $L_1$-norm of its Fourier transform, $\wh{\p},$ outside the set $T$ is at most $\eps.$}

\medskip

We will need two different versions of the above statement for our applications, and therefore we do not provide
a formal statement at this stage. The precise meaning of the term ``small'' depends on the setting:
For the continuous Fourier transform, we essentially prove that the product of the volume
of the effective support of the Fourier transform times the number of points
in the effective support of our distribution is small.
In particular, the set $T$ is a scaled version of the dual ellipsoid 
to the ellipsoid defined by the covariance matrix of $\p.$
Hence, roughly speaking, $\wh{\p}$ has an effective support that is the dual of the effective support of $\p.$
(See Lemma~\ref{lem:ft-es} in Section~\ref{sec:cover-nash} for the precise statement.)

In the case of the Discrete Fourier Transform (DFT),
we show that there exists a discrete set with small cardinality,
such that $L_1$-norm of the DFT outside this set is small.
At a high-level, to prove this statement, we need the appropriate definition of the (multidimensional) DFT,
which turns out to be non-trivial, and is crucial for the computational efficiency of our learning algorithm.
More specifically, we chose the period of the DFT
to reflect the shape of the effective support of our PMD.
(See Proposition~\ref{prop:ft-effective-support} in Section~\ref{sec:algo} for the statement.)



With Fourier sparsity as our starting point, we obtain new structural results of independent interest for PMDs.
The first is a ``robust'' moment-matching lemma, which we now informally state:

\medskip

\noindent {\bf Informal Lemma} (Parameter Moment Closeness Implies Closeness in Distribution.)
{\em For any pair of $(n, k)$-PMDs $\p, \q$, if the ``low-degree''  parameter moment profiles of $\p$ and $\q$ are close,
then $\p, \q$ are close in total variation distance.}

\medskip
See Definition~\ref{def:param-moments} for the definition of parameter moments of a PMD.
The formal statement of the aforementioned lemma appears as Lemma~\ref{lem:moments-imply-dtv}
in Section~\ref{ssec:moments-struct}. Our robust moment-matching lemma is the basis for our proper cover algorithm
and our EPTAS for Nash equilibria in anonymous games.
Our constructive cover upper bound is the following:

\begin{theorem}[Optimal Covers for PMDs] \label{thm:cover-pmd}
For all $n, k \in \Z_+,$ $k>2,$ and $\eps>0$, there exists an $\eps$-cover
${\cal M}_{n, k, \eps} \subseteq {\cal M}_{n, k}$, under the total variation distance, of the set
${\cal M}_{n, k}$ of $(n, k)$-PMDs of size
$|{\cal M}_{n, k, \eps}| \le n^{O(k^2)} \cdot (1/\eps)^{O(k\log(k/\eps)/\log\log(k/\eps))^{k-1}}.$
In addition, there exists an algorithm to construct the set ${\cal M}_{n, k, \eps}$ that runs in time
$n^{O(k^3)} \cdot (1/\eps)^{O(k^3\log(k/\eps)/\log\log(k/\eps))^{k-1}}.$
\end{theorem}

A sparse proper cover quantifies the ``size'' of the space of PMDs and
provides useful structural information that can be exploited in a variety of applications. 
In addition to Nash equilibria in anonymous games,
our efficient proper cover construction provides a smaller search space
for approximately solving essentially any optimization problem over PMDs. 
As another corollary of our cover construction, we obtain the first EPTAS
for computing threat points in anonymous games.

\medskip

Perhaps surprisingly, we also prove that our above upper bound is essentially tight:

\begin{theorem}[Cover Lower Bound for PMDs] \label{thm:cover-lb}
For any $k>2$, $\eps>0$ sufficiently small as a function of $k,$ and $n =\Omega_k(\log(1/\eps)/ \log\log(1/\eps))^{k-1}$,
any $\eps$-cover for ${\cal M}_{n, k}$ has size at least
$n^{\Omega(k)} \cdot (1/\eps)^{\Omega_k(\log(1/\eps)/\log\log(1/\eps))^{k-1}}.$
\end{theorem}

We remark that, in previous work~\cite{DKS15}, the authors proved a tight cover size bound
of $n \cdot (1/\eps)^{\Theta(k \log(1/\eps))}$ for $(n, k)$-SIIRVs, 
i.e., sums of $n$ independent scalar random variables each supported on $[k].$
While a cover size lower bound for  $(n, k)$-SIIRVs directly implies the same lower bound
for $(n, k)$-PMDs, the opposite is not true. Indeed, Theorems~\ref{thm:cover-pmd} and~\ref{thm:cover-lb} 
show that covers for $(n, k)$-PMDs are inherently larger, requiring a doubly exponential dependence
on $k.$

\subsection{Our Approach and Techniques} \label{ssec:techniques}
At a high-level, the Fourier techniques of this paper can be viewed as
a  highly non-trivial generalization of the techniques in our recent paper~\cite{DKS15} on sums of independent scalar random variables.
We would like to emphasize that a number of new conceptual and technical ideas are required to overcome the various obstacles arising
in the multi-dimensional setting.

\medskip

We start with an intuitive explanation of two key ideas that form the basis of our approach.

\paragraph{Sparsity of the Fourier Transform of PMDs.}
Since the Fourier Transform (FT) of a PMD is the product of the FTs of its component CRVs,
its magnitude is the product of terms each bounded from above by $1.$
Note that each term in the product is strictly less than $1$ except in a small region, unless
the component CRV is trivial (i.e., essentially deterministic). Roughly speaking,
to establish the sparsity of the FT of PMDs, we proceed as follows: We bound from above the magnitude of the FT
by the FT of a Gaussian with the same covariance matrix as our PMD.
(See, for example, Lemma~\ref{lem:gaussian-bound-from-interval}.)
This gives us tail bounds for the FT of the PMD in terms of the FT of this Gaussian, and
when combined with the concentration of the PMD itself, yields the desired property.

\paragraph{Approximation of the logarithm of the Fourier Transform.}
A key ingredient in our proofs is the approximation of the logarithm of the Fourier Transform (log FT) of PMDs
by low-degree polynomials. Observe that the log FT is a sum of terms, which is convenient for the analysis.
We focus on approximating the log FT by a low-degree Taylor polynomial within the effective support of the FT.
(Note that outside the effective support the log FT can be infinity.)
Morally speaking, the log FT is smooth, i.e., it is approximated
by the first several terms of its Taylor series. Formally however, this statement is in general not true
and requires various technical conditions, depending on the setting.
One important point to note is that the sparsity of the FT controls the domain in which this approximation
will need to hold, and thus help us bound the Taylor error.
We will need to ensure that the sizes of the Taylor coefficients are not too large
given the location of the effective support, which turns out to be a non-trivial technical hurdle.
To ensure this, we need to be very careful about how we perform this Taylor expansion.
In particular, the correct choice of the point that we Taylor expand around will be critical for our applications.
We elaborate on these difficulties in the relevant technical sections.
Finally, we remark that the degree of polynomial approximation we will require depends on the setting: In our cover upper bounds,
we will require (nearly) logarithmic degree, while for our CLT degree-$2$ approximation suffices.

\medskip

We are now ready to give an overview of the ideas in the proofs of each of our results.

\paragraph{Efficient Learning Algorithm.}
The high-level structure of our learning algorithm relies on the sparsity of the Fourier transform, and is
similar to the algorithm in our previous work~\cite{DKS15}
for learning sums of independent integer random variables.
More specifically, our learning algorithm estimates the effective support of the DFT,
and then computes the empirical DFT in this effective support.
This high-level description would perhaps suffice, if we were
only interested in bounding the sample complexity. In order to obtain
a computationally efficient algorithm, it is crucial to use the appropriate definition of the DFT
and its inverse.

In more detail, our algorithm works as follows:
It starts by drawing $\poly(k)$ samples to estimate
the mean vector and covariance matrix of our PMD to good accuracy.
Using these estimates, we can bound the effective support of our distribution in an appropriate ellipsoid.
In particular, we show that our PMD lies (whp) in a fundamental domain
of an appropriate integer lattice $L = M\Z^k,$
where $M \in \Z^{k \times k}$ is an integer matrix whose columns are appropriate functions
of the eigenvalues and eigenvectors of the (sample) covariance matrix.
This property allows us to learn our unknown PMD $X$ by learning the random variable $X\pmod{L}.$
To do this, we learn its Discrete Fourier transform.
Let $L^{\ast}$ be the dual lattice to $L$
(i.e., the set of points $\xi$ so that $\xi\cdot x\in \Z$ for all $x\in L$).
Importantly, we define the DFT, $\wh{\p},$ of our PMD $X \sim \p$ on the dual lattice $L^{\ast},$
that is, $\wh{\p}:L^{\ast}/\Z^k \rightarrow \C$ with
$\wh{\p}(\xi) = \E[e(\xi \cdot X)].$
A useful property of this definition is the following: the probability that $X \pmod{L}$ attains a given value $x$
is given by the inverse DFT, defined on the lattice $L,$
namely $\Pr\left[X \pmod{L} = x\right] = \frac{1}{|\det(M)|}\sum_{\xi \in L^{\ast}/\Z^k} \wh{\p}(\xi)e(-\xi\cdot x).$

The main structural property needed for the analysis of our algorithm is that there exists an explicit set $T$ with integer coordinates
and cardinality $(k\log(1/\eps))^{O(k)}$ that contains all but $O(\eps)$ of the $L_1$ mass of $\wh{\p}.$
Given this property, our algorithm draws an additional set of samples of size $(k \log(1/\eps))^{O(k)}/\eps^2$ from the PMD,
and computes the empirical DFT (modulo $L$) on its effective support $T.$
Using these ingredients, we are able to show that the inverse
of the empirical DFT defines a pseudo-distribution that is $\eps$-close to our unknown PMD in total variation distance.

Observe that the support of the inverse DFT can be large, namely $\Omega(n^{k-1}).$
Our  algorithm {\em does not} explicitly evaluate the inverse DFT at all these points, but
outputs a succinct description of its hypothesis $\h$, via its DFT $\wh{\h}.$ We emphasize that this succinct description
suffices to efficiently obtain both an approximate evaluation oracle and an approximate sampler for our target PMD $\p.$
Indeed, it is clear that computing the inverse DFT at a single point can be done in time  $O(|T|) = (k\log(1/\eps))^{O(k)},$
and gives an approximate oracle for the probability mass function of $\p.$
By using additional algorithmic ingredients, we show how to use
an oracle for the DFT, $\wh{\h}$, as a black-box to obtain a computationally efficient approximate sampler for $\p.$

Our learning algorithm and its analysis are given in Section~\ref{sec:algo}.

\paragraph{Constructive Proper Cover and Anonymous Games.}
The correctness of our learning algorithm easily implies (see Section~\ref{sec:cover-from-alg})
an algorithm to construct a {\em non-proper} $\eps$-cover for PMDs of size $n^{O(k^2)} \cdot (1/\eps)^{\log(1/\eps))^{O(k)}}.$
While this upper bound is close to being best possible (see Section~\ref{sec:cover-lb}),
it does not suffice for our algorithmic applications in anonymous games.
For these applications, it is crucial to obtain an efficient algorithm that constructs a {\em proper} $\eps$-cover,
and in fact one that works in a certain stylized way.

To construct a proper cover, we rely on the sparsity of the continuous Fourier Transform of PMDs.
Namely, we show that for any PMD $\p,$ with effective support $S \subseteq [n]^k,$
there exists an appropriately defined set $T \subseteq [0, 1]^k$ such that the contribution of $\overline{T}$
to the $L_1$-norm of  $|\wh{\p}|$ is at most $\eps / |S|.$ By using this property, we show
that any two PMDs, with approximately the same variance in each direction, that have continuous Fourier transforms
close to each other in the set $T,$ are close in total variation distance. We build on this lemma to prove
our robust moment-matching result. Roughly speaking, we show that two PMDs,
with approximately the same variance in each direction, that are ``close'' to each other in their low-degree  parameter moments
are also close in total variation distance. We emphasize that the meaning of the term ``close'' here is quite subtle:
we need to appropriately partition the component CRVs into groups, and approximate the parameter moments of the PMDs formed by 
each group within a different degree and different accuracy for each degree. 
(See Lemma~ \ref{lem:moments-imply-dtv} in Section~\ref{ssec:moments-struct}.)

Our algorithm to construct a proper cover, and our EPTAS for Nash equilibria
in anonymous games proceed by a careful dynamic programming approach,
that is based on our aforementioned robust moment-matching result.

Finally, we note that combining our moment-matching lemma with a recent result in algebraic geometry
gives us the following structural result of independent interest: Every PMD is $\eps$-close to another PMD
that is a sum of at most $O(k + \log(1/\eps))^{k}$ distinct $k$-CRVs.

The aforementioned algorithmic and structural results are given in Section~\ref{sec:cover-nash}.

\paragraph{Cover Size Lower Bound.}
As mentioned above, a crucial ingredient of our cover upper bound is
a robust moment-matching lemma, which
translates closeness between the low-degree  parameter  moments
of two PMDs to closeness between their Fourier Transforms,
and in turn to closeness in total variation distance.
To prove our cover lower bound, we follow the opposite direction.
We construct an explicit set of PMDs with the property that
{\em any} pair of distinct PMDs in our set
have a non-trivial difference in (at least) one of their low-degree  parameter moments.
We then show that difference in one of the  parameter moments implies
that there exists a point where the probability generating functions
have a non-trivial difference. Notably, our proof for this step
is non-constructive making essential use of Cauchy's integral formula.
Finally, we can easily translate a pointwise difference between the probability
generating functions to a non-trivial total variation distance error.
We present our cover lower bound construction in Section~\ref{sec:cover-lb}.


\paragraph{Central Limit Theorem for PMDs.}
The basic idea of the proof of our CLT will be to compare the Fourier transform of our PMD
$X$ to that of the discrete Gaussian $G$ with the same mean and covariance.
By taking the inverse Fourier transform, we will be able to conclude that these distributions are pointwise close.
A careful analysis using a Taylor approximation and the fact that both $\wh{X}$ and $\wh{G}$
have small effective support, gives us a total variation distance error independent of the size $n.$
Alas, this approach results in an error dependence that is exponential in $k.$
To obtain an error bound that scales polynomially with $k,$ 
we require stronger bounds between $X$ and $G$ at points away from the mean.
Intuitively, we need to take advantage of cancellation in the inverse Fourier transform integrals.
To achieve this, we will use the saddlepoint method from complex analysis.
The full proof of our CLT is given in Section~\ref{sec:clt}.

\subsection{Related and Prior Work} \label{ssec:related}

There is extensive literature on distribution learning and computation of approximate Nash equilibria in various
classes of games. We have already mentioned the most relevant references in the introduction.

Daskalakis {\em et al.}~\cite{DKT15} studied the structure and learnability of PMDs.
They obtained a non-proper $\eps$-cover of size $n^{k^2} \cdot 2^{O(k^{5k} \log(1/\eps)^{k+2})},$
and an information-theoretic upper bound on the learning sample complexity of  $O(k^{5k} \log(1/\eps)^{k+2}/\eps^2).$
The dependence on $1/\eps$ in their cover size is also quasi-polynomial, but is suboptimal as follows
from our upper and lower bounds. Importantly, the \cite{DKT15} construction yields a {\em non-proper} cover.
As previously mentioned, a {\em proper} cover construction is necessary for our algorithmic applications.
We note that the learning algorithm of~\cite{DKT15} relies on enumeration over a cover, 
hence runs in time quasi-polynomial in $1/\eps,$ even for $k=2.$
The techniques of~\cite{DKT15} are orthogonal to ours. 
Their cover upper bound is obtained by a clever black-box application of the CLT of~\cite{VV10b},
combined with a non-robust moment-matching lemma that they deduce from a result of Roos~\cite{Roos02}.
We remind the reader that our Fourier techniques strengthen both these technical tools:
Theorem~\ref{thm:clt} strengthens the CLT of~\cite{VV10b}, and we prove a {\em robust} and quantitatively 
essentially optimal moment-matching lemma.

In recent work~\cite{DKS15}, the authors used Fourier analytic techniques to study the structure and learnability
of sums of independent integer random variables (SIIRVs). The techniques of this paper can be viewed as a (highly nontrivial) generalization
of those in~\cite{DKS15}. We also note that the upper bounds we obtain in this paper for learning and covering PMDs do not subsume
the ones in~\cite{DKS15}. In fact, our cover upper and lower bounds in this work show that optimal covers for PMDs
are inherently larger than optimal covers for SIIRVs. Moreover, the sample complexity of our SIIRV learning algorithm~\cite{DKS15}
is significantly better than that of our PMD learning algorithm in this paper.


\subsection{Concurrent and Independent Work} \label{ssec:conc}

Concurrently and independently to our work, \cite{DDKT15} obtained qualitatively similar results
using different techniques. We now provide a statement of the \cite{DDKT15}
results in tandem with a comparison to our work.

\cite{DDKT15} give a learning algorithm for PMDs with sample complexity
$(k \log(1/\eps)^{O(k)}/\eps^2)$ and runtime $(k/\eps)^{O(k^2)}.$ 
The \cite{DDKT15} algorithm uses the continuous Fourier transform, exploiting 
its sparsity property, plus additional structural and algorithmic ingredients. 
The aforementioned runtime is not polynomial in the sample size, 
unless $k$ is fixed.  In contrast, our learning algorithm runs in sample--polynomial time, and, 
for fixed $k$, in nearly-linear time.
The \cite{DDKT15} learning algorithm outputs an explicit hypothesis, 
which can be easily sampled. On the other hand, our algorithm outputs 
a succinct description of its hypothesis (via its DFT), 
and we show how to efficiently sample from it.

\cite{DDKT15} also prove a size-free CLT, analogous to our Theorem~\ref{thm:clt},
with error polynomial in $k$ and $1/\sigma.$  
Their CLT is obtained by bootstrapping the CLT of~\cite{VV10b, ValiantValiant:11}
using techniques from~\cite{DKT15}. As previously mentioned, 
our proof is technically orthogonal to~\cite{VV10b, ValiantValiant:11, DDKT15}, 
making use of the sparsity of the Fourier transform 
combined with tools from complex analysis. 
It is worth noting that 
our CLT also achieves a near-optimal dependence in the error as a function of $1/\sigma$
(up to log factors).

Finally, \cite{DDKT15} prove analogues of 
Theorems~\ref{thm:nash-anonymous},~\ref{thm:cover-pmd}, and~\ref{thm:cover-lb} with qualitatively similar bounds to ours.
We note that \cite{DDKT15} improve the dependence on $n$ in the cover size 
to an optimal $n^{O(k)},$ while the dependence on $\eps$ in their cover upper 
bound is the same as in~\cite{DKT15}. The cover size lower bound of 
\cite{DDKT15} is qualitatively of the right form, though slightly suboptimal as a function 
of $\eps.$  The algorithms to construct proper covers and the corresponding EPTAS for anonymous games 
in both works have running time roughly comparable to the PMD cover size.

\subsection{Organization} \label{sec:org}
In Section~\ref{sec:algo}, we describe and analyze our learning algorithm for PMDs.
Section~\ref{sec:cover-nash} contains our proper cover upper bound construction, our cover size lower bound, 
and the related approximation algorithm for Nash equilibria in anonymous games. 
Finally, Section~\ref{sec:clt} contains the proof of our CLT.

\section{Preliminaries}

In this section, we record the necessary definitions and terminology that will be used throughout the technical sections of this paper.

\paragraph{Notation.} For $n \in \Z_+$, we will denote $[n] \eqdef \{1, \ldots, n\}.$ For a vector $v \in \R^n$,
and $p \ge 1$, we will denote $\|v\|_p \eqdef \left(\sum_{i=1}^n |v_i|^{p}\right)^{1/p}.$ We will use the boldface notation
$\mathbf{0}$ to denote the zero vector or matrix in the appropriate dimension.

\paragraph{Poisson Multinomial Distributions.} We start by defining our basic object of study:
\begin{definition}[$(n, k)$-PMD]  \label{def:pmd}
For $k \in \Z_{+}$, let $e_j$, $j \in [k]$, be the standard unit vector along dimension $j$ in $\R^k$.
A {\em $k$-Categorical Random Variable} ($k$-CRV) is a vector random variable supported on the set
$\{e_1, e_2, \ldots, e_k \}$.
A {\em$k$-Poisson Multinomial Distribution of order $n$}, or $(n, k)$-PMD,
is any vector random variable of the form $X = \sum_{i=1}^n X_i$ where the $X_i$'s are independent $k$-CRVs.
We will denote by ${\cal M}_{n,k}$ the set of all $(n, k)$-PMDs.
\end{definition}

\noindent We will require the following notion of a parameter moment for a PMD:

\begin{definition}[$m^{th}$-parameter moment of a PMD] \label{def:param-moments}
Let $X = \sum_{i=1}^n X_i$ be an $(n, k)$-PMD such that for $1 \le i \le n$ and $1 \le j \le k$ we denote $p_{i, j} = \Pr[X_i = e_j]$.
For $m=(m_1,\ldots,m_{k}) \in \Z^{k}_+$, we define the {\em $m^{th}$-parameter moment of $X$} to be
$M_m(X) \eqdef \sum_{i=1}^n \prod_{j=1}^{k} p_{i,j}^{m_j}.$ We will refer to $|m|_1 = \sum_{j=1}^{k} m_j$ as the {\em degree}
of the  parameter moment $M_m(X).$
\end{definition}

\paragraph{(Pseudo-)Distributions and Total Variation Distance.}
A function $\p : A \to \R$, over a finite set $A$, is called a {\em distribution}
if $\p(a) \ge 0$ for all $a \in A$, and $\sum_{a \in A} \p (a)=1.$
The function $\p$ is called a {\em pseudo-distribution} if $\sum_{a \in A} \p (a)=1.$
For $S \subseteq A$, we sometimes write $\p(S)$ to denote $\sum_{a \in S}\p(a)$.
A distribution $\p$ supported on a finite domain $A$ can be viewed as the probability mass
function of a random variable $X$, i.e., $\p(a) = \Pr_{X \sim \p}[X= a].$

The {\em total variation distance} between two pseudo-distributions
$\p$ and $\q$ supported on a finite domain $A$ is
$\dtv\left(\p, \q \right) \eqdef \max_{S \subseteq A} \left |\p(S)-\q(S) \right|= (1/2) \cdot \| \p -\q  \|_1 = (1/2) \cdot \littlesum_{a \in A} |\p(a)-\q(a)|.
$
If $X$ and $Y$ are two random variables ranging over a finite set, their total
variation distance $\dtv(X,Y)$ is defined as the total variation
distance between their distributions.
For convenience, we will often blur the distinction between a random variable and its distribution.

\paragraph{Covers.}
Let $(\mathcal{X}, d)$ be a metric space. Given $\eps > 0$, a subset ${\cal Y} \subseteq {\cal X}$
 is said to be a proper \emph{$\eps$-cover of ${\cal X}$} with respect to the metric $d: \mathcal{X}^2 \to \R_+,$
if for every $\mathbf{x} \in {\cal X}$ there exists some $\mathbf{y} \in {\cal Y}$ such that $d(\mathbf{x}, \mathbf{y}) \leq \eps.$
(If ${\cal Y}$ is not necessarily a subset of ${\cal X},$ then we obtain a non-proper $\eps$-cover.)
There may exist many $\eps$-covers of ${\cal X}$, but one is typically interested in one with the minimum cardinality.
The {\em $\eps$-covering number} of $({\cal X}, d)$ is the minimum cardinality of any $\eps$-cover of ${\cal X}$.
Intuitively, the covering number of a metric space captures the ``size'' of the space.
In this work, we will be interested on efficiently constructing sparse covers
for PMDs under the total variation distance metric.

\paragraph{Distribution Learning.} We now define the notion of distribution learning we use in this paper.
Note that an explicit description of a discrete distribution via its probability mass function
scales linearly with the support size. Since we are interested in the computational complexity
of distribution learning, our algorithms will need to use a {\em succinct description} of their output hypothesis.
A simple succinct representation of a discrete distribution is via an evaluation
oracle for the probability mass function:

\begin{definition}[Evaluation Oracle]  \label{def:evaluator}
Let $\p$ be a distribution over $[n]^k.$
An \emph{evaluation oracle} for $\p$ is a polynomial size circuit $C$
with $m= O(k \log n)$ input bits $z \in [n]^k$
such that for  each $z \in [n]^k,$ the output of the circuit $C(z)$
equals the binary representation of the probability $\p(z)$.
For $\eps > 0$, an \emph{$\eps$-evaluation oracle} for $\p$ is an evaluation oracle
for some pseudo-distribution $\p'$ which has $\dtv(\p', \p) \leq \eps.$
\end{definition}

One of the most general ways to succinctly specify a distribution is to give the code of an
efficient algorithm that takes ``pure'' randomness and transforms it into a sample from the distribution.
This is the standard notion of a sampler:

\begin{definition}[Sampler] \label{def:sampler}
Let $\p$ be a distribution over $[n]^k.$
An \emph{$\eps$-sampler} for $\p$ is a circuit $C$
with $m= O(k \log n+\log(1/\eps))$ input bits $z$ and $m'= O(k\log n)$ output bits $y$
which is such that when $z \sim U_m$ 
then  $y \sim \p',$ for some distribution $\p'$ which has $\dtv(\p',\p) \leq \eps.$
\end{definition}

We can now give a formal definition of distribution learning:

\newpage

\begin{definition}[Distribution Learning]  \label{def:learning-sampler}
Let ${\cal D}$ be a family of distributions.
A randomized algorithm $A^{\cal D}$ is a {\em distribution learning algorithm for class $\cal D,$}
if for any $\eps>0,$ and any $\p \in \cal D,$ on input $\eps$ and sample access to $\p,$
with probability $9/10,$ algorithm $A^{\cal D}$
outputs an $\eps$-sampler (or an $\eps$-evaluation oracle) for $\p.$
\end{definition}

\begin{remark}
We emphasize that our learning algorithm in Section~\ref{sec:algo} outputs {\em both an $\eps$-sampler
and an $\eps$-evaluation oracle} for the target distribution.
\end{remark}

\paragraph{Anonymous Games and Nash Equilibria.}

An anonymous game is a triple $(n, k, \{u^i_{\ell}\}_{i \in [n], \ell \in [k]})$ where $[n]$, $n \ge 2$, is the set of players,
$[k]$, $k \ge 2$, a common set of strategies available to all players, and $u^i_{\ell}$ the payoff function of player
$i$ when she plays strategy $\ell$. This function maps the set of partitions
$\Pi_{n-1}^k = \{(x_1, \ldots, x_k) \mid x_{\ell} \in \Z_+ \textrm{ for all } \ell \in [k] \wedge \sum_{\ell=1}^k x_{\ell} = n-1\}$
to the interval $[0, 1]$. That is, it is assumed that the payoff of each player depends on her own strategy and only the
number of other players choosing each of the $k$ strategies.

We denote by $\Delta^k_{n-1}$ the convex hull of the set $\Pi_{n-1}^k$, i.e.,
$\Delta^k_{n-1} = \{ (x_1, \ldots, x_k) \mid x_{\ell} \ge 0 \textrm{ for all } \ell \in [k] \wedge \sum_{\ell=1}^k x_{\ell} = n-1 \}.$
A {\em mixed strategy} is an element of $\Delta^k \eqdef \Delta^k_{1}.$
A {\em mixed strategy profile} is a mapping $\delta$ from $[n]$ to $\Delta^k$.
We denote by $\delta_i$ the mixed strategy of player $i$ in the profile $\delta$ and $\delta_{-i}$
the collection of all mixed strategies but $i$'s in $\delta$.
For $\eps \ge 0$, a mixed strategy profile $\delta$ is a (well-supported) {\em $\eps$-Nash equilibrium}  iff
for all $i \in [n]$ and  $\ell, \ell' \in [k]$ we have:
$\E_{x \sim \delta_{-i}} [u^i_{\ell}(x)] > \E_{x \sim \delta_{-i}} [u^i_{\ell'}(x)]+  \eps \implies \delta_i(\ell') = 0.$
Note that given a mixed strategy profile $\delta$, we can compute a player's expected payoff in time $\poly(n^k)$
by straightforward dynamic programming.

Note that the mixed strategy $\delta_i$ of player $i \in [n]$
defines the $k$-CRV $X_i$, i.e., a random vector supported in the set $\{e_1, \ldots, e_k\}$,
such that $\Pr[X_i = e_{\ell}] = \delta_i(\ell)$, for all $\ell$. Hence, if $(X_1, \ldots, X_n)$ is a mixed strategy
profile, the expected payoff of player $i \in [n]$ for using pure strategy $\ell \in [k]$ is
$ \E\left[ u^i_{\ell} \left(\sum_{j \ne i, j \in [n]} X_j \right)  \right].$

\paragraph{Multidimensional Fourier Transform.}
Throughout this paper, we will make essential use of the (continuous and the discrete) multidimensional Fourier transform.
For $x \in \R$, we will denote $e(x) \eqdef  \exp(-2 \pi i x)$.
The {\em (continuous) Fourier Transform (FT)} of a function
$F: \Z^k \rightarrow \C$ is  the function $\widehat{F}: [0, 1]^k \rightarrow \C$
defined as $\widehat{F}(\xi)=\sum_{x \in  \Z^k} e(\xi \cdot x) F(x).$
For the case that $F$ is a probability mass function, we can equivalently write
$\widehat{F}(\xi)= \E_{x \sim F} \left[ e(\xi \cdot x) \right].$

For computational purposes, we will also need the Discrete Fourier Transform (DFT)
and its inverse, whose definition is somewhat more subtle.
Let $M \in \Z^{k \times k}$ be an integer $k \times k$ matrix.
We consider the integer lattice
$L  = L(M) =  M \Z^k \eqdef \{ p \in \Z^k \mid p = M q,  q \in \Z^k \}$, and its dual lattice
 $L^{\ast} = L^{\ast}(M)   \eqdef \{ \xi \in \R^k \mid \xi \cdot x \in \Z \textrm{ for all } x \in L \}.$
 Note that  $L^{\ast} = (M^T)^{-1} \Z^k,$ and that $L^{\ast}$ is not necessarily integral.
The quotient $ \Z^k \slash L$ is the set of equivalence classes of points in $\Z^k$ such that two points $x, y \in \Z^k$
are in the same equivalence class iff $x - y \in L$.
Similarly, the quotient $L^{\ast} \slash \Z^k$ is the set of equivalence
classes of points in $L^{\ast}$ such that any two points $x, y \in L^{\ast}$ are in the same equivalence class iff $x -y \in \Z^k$.

The {\em Discrete Fourier Transform (DFT) modulo $M$}, $M \in \Z^{k \times k}$, of a function
$F: \Z^k \rightarrow \C$ is  the function $\widehat{F}_M: L^{\ast} \slash \Z^k  \rightarrow \C$
defined as $\widehat{F}_M(\xi)=\sum_{x \in  \Z^k} e(\xi \cdot x) F(x).$ 
(We will remove the subscript $M$ when it is clear from the context.)
Similarly, for the case that $F$ is a probability mass function, we can equivalently write
$\widehat{F}(\xi)= \E_{x \sim F} \left[ e(\xi \cdot x) \right].$ The {\em inverse DFT} of a function $\widehat{G}: L^{\ast} \slash \Z^k  \rightarrow \C$
is the function $G: A \rightarrow \C$ defined on a {\em fundamental domain} $A$ of $L(M)$ as follows:
$G(x) = \frac{1}{|\det(M)|} \sum_{\xi \in L^{\ast} \slash \Z^k} \widehat{G}(x) e(- \xi \cdot x).$
Note that these operations are inverse of each other,
namely for any function $F: A \rightarrow \C$, the inverse DFT of $\widehat{F}$ is identified with  $F.$

Let $X = \sum_{i=1}^n X_i$ be an $(n, k)$-PMD such that for $1 \le i \le n$ and $1 \le j \le k$ we denote $p_{i, j} = \Pr[X_i = e_j]$,
where $\sum_{j=1}^k p_{i, j} = 1.$ To avoid clutter in the notation, we will sometimes use the symbol $X$ to denote the corresponding probability mass function.
With this convention, we can write that $\wh{X}(\xi) = \prod_{i=1}^n \wh{X_i}(\xi) = \prod_{i=1}^n \sum_{j=1}^{k} e(\xi_j) p_{i,j}.$

\vspace{-0.3cm}

\paragraph{Basics from Linear Algebra.}
We remind the reader a few basic definitions from linear algebra
that we will repeatedly use throughout this paper.
The Frobenius norm of $A\in\R^{m \times n}$ is $\|A\|_F \eqdef \sqrt{\sum_{i,j} A_{i,j}^2}.$
The spectral norm (or induced $L_2$-norm) of $A\in\R^{m\times n}$ is defined as
$\|A\|_2 \eqdef \max_{x: \|x\|_2=1} \|Ax\|_2 = \sqrt{\lambda_{\max}(A^T A)}.$
We note that for any $A\in\R^{m \times n}$, it holds $\|A\|_2 \leq \|A\|_F.$
A symmetric matrix $A\in\R^{n \times n}$ is called positive semidefinite (PSD),
denoted by $A \succeq \mathbf{0},$
if $x^T A x \geq 0$ for all $x \in \R^n$, or equivalently all the eigenvalues of $A$ are nonnegative.
Similarly, a symmetric matrix $A\in\R^{n \times n}$ is called positive definite (PD),
denoted by $A \succ \mathbf{0},$
if $x^T A x > 0$ for all $x \in \R^n$, $x \neq \mathbf{0}$,
or equivalently all the eigenvalues of $A$ are strictly positive.
For two symmetric matrices $A, B \in \R^{n \times n}$ we write $A \succeq B$ to denote
that the difference $A-B$ is PSD, i.e., $A-B \succeq  \mathbf{0}.$
Similarly, we write $A \succ B$ to denote
that the difference $A-B$ is PD, i.e., $A-B \succ  \mathbf{0}.$


\section{Efficiently Learning PMDs} \label{sec:algo}

In this section, we describe and analyze our sample near-optimal and computationally efficient
learning algorithm for PMDs. This section is organized as follows:
In Section~\ref{ssec:algo-dft}, we give our main algorithm which, given samples from a PMD $\p$,
efficiently computes a succinct description of a hypothesis pseudo-distribution $\h$ such that $\dtv(\h, \p) \le \eps/3.$
As previously explained, the succinct description of $\h$ is via its DFT $\wh{\h}$, which is supported
on a discrete set $T$ of cardinality $|T| = (k \log (1/\eps))^{O(k)}$.
Note that $\wh{\h}$ provides an $\eps$-evaluation oracle for $\p$ with running time $O(|T|).$
In Section~\ref{ssec:sampler}, we show how to use $\wh{\h},$ in a black-box manner, to efficiently obtain an $\eps$-sampler for $\p$, i.e.,
sample from a distribution $\q$ such that $\dtv(\q, \p) \le \eps.$ Finally, in Section~\ref{sec:cover-from-alg} we show
how a nearly--tight cover upper bound can easily be deduced from our learning algorithm.

\subsection{Main Learning Algorithm} \label{ssec:algo-dft}
In this subsection, we give an algorithm  {\tt Efficient-Learn-PMD} establishing the following theorem:

\begin{theorem} \label{thm:alg-dft}
For all $n, k \in \Z_+$ and $\eps>0$,
the algorithm  {\tt Efficient-Learn-PMD} has the following
performance guarantee: Let $\p$ be an unknown $(n, k)$-PMD.
The algorithm uses $O\left(k^{4k}\log^{2k}(k/\eps)/\eps^2\right)$
samples from $\p$,  runs in time {$O\left(k^{6k} \log^{3k}(k/\eps)/\eps^2 + k^4 \log \log n\right),$}
and outputs the DFT $\wh{\h}$ of a pseudo-distribution $\h$ that,
with probability at least $9/10$, satisfies $\dtv(\h, \p) \le \eps/3.$
\end{theorem}


Our learning algorithm is described in the following pseudo-code:

\bigskip

\fbox{\parbox{6.3in}{
{\bf Algorithm} {\tt Efficient-Learn-PMD}\\
{\em Input:} sample access to an $(n, k)$-PMD $X \sim \p$ and $\eps>0.$\\
{\em Output:} A set $T\subseteq (\R/\Z)^k$ of cardinality $|T| \le O(k^2 \log(k/\eps))^k$, 
and the DFT $\wh{\h}: T \to \C$ of a pseudo-distribution $\h$ such that $\dtv(\h, \p) \le \eps/3.$

\vspace{0.2cm}

Let $C>0$ be a sufficiently large universal constant.

\begin{enumerate}
\item Draw {$m_0 = O(k^4)$} samples from $X$, and 
let $\wh{\mu}$ be the sample mean and $\wh{\Sigma}$ the sample covariance matrix.

\item Compute an approximate spectral decomposition of $\wh{\Sigma}$, i.e., 
an orthonormal eigenbasis $v_i$ with corresponding eigenvalues $\lambda_i$, $i \in [k].$

\item Let $M \in \Z^{k \times k}$ be the matrix whose $i^{th}$ column 
is the closest integer point to the vector $C \left(\sqrt{k \ln(k/\eps)\lambda_i+k^2\ln^2(k/\eps)}\right)v_i.$

\item Define $T\subseteq (\R/\Z)^k$ to be the set of points $\xi = (\xi_1, \ldots, \xi_k)$ 
of the form $\xi = (M^T)^{-1} \cdot v {+\Z^k},$ for some $v\in \Z^k$ with $\|v\|_2 \leq C^2 k^2 \ln(k/\eps).$

\item \label{step:compute-DFT} 
Draw $m = O\left(k^{4k}\log^{2k}(k/\eps)/\eps^2\right)$ samples $s_i$, $i \in [m]$, from $X$, 
and output the empirical DFT $\wh{\h}: T \rightarrow \C$, i.e.,
$\wh{\h}(\xi) = \frac{1}{m} \sum_{i=1}^m e(\xi\cdot s_i).$

{\em /* The DFT  $\wh{\h}$ is a succinct description of the pseudo-distribution $\h$, the inverse DFT of $\wh{\h}$, 
defined by:  $\h(x) = \frac{1}{|\det(M)|}\sum_{\xi\in T}\wh{\h}(\xi)e(-\xi\cdot x),$ for  $x\in \Z^k\cap (\wh{\mu} + M \cdot (-1/2,1/2]^k)$, 
and $\h(x)=0$ otherwise. Our algorithm {\bf does not} output $\h$ explicitly, but implicitly via its DFT.*/}

\end{enumerate}

}}

\bigskip

Let $X$ be the unknown target $(n, k)$-PMD.
We will denote by $\p$ the probability mass function of $X$, i.e., $X \sim \p$.
Throughout this analysis,  we will denote by $\mu$ and $\Sigma$ 
the mean vector and covariance matrix of $X.$

First, note that the algorithm  {\tt Efficient-Learn-PMD} is easily seen 
to have the desired sample and time complexity.
Indeed, the algorithm draws $m_0$ samples in Step~1 and $m$ 
samples in Step~5, for a total sample complexity of $O(k^{4k}\log^{2k}(k/\eps)/\eps^2).$
The runtime of the algorithm is dominated by computing the DFT in Step~\ref{step:compute-DFT}
which takes time $O(m |T|)=O(k^{6k}\log^{3k}(k/\eps)/\eps^2).$ 
Computing an approximate eigendecomposition can be done in time $O(k^4 \log \log n)$(see, e.g., \cite{PCZ98}).
The remaining part of this section is devoted
to proving the correctness of our algorithm.

\begin{remark}
{\em We remark that in Step~4 of our algorithm, the notation $\xi = (M^T)^{-1} \cdot v {+\Z^k}$
refers to an equivalence class of points. In particular, any pair of distinct vectors $v, v' \in \Z^k$ satisfying
$\|v\|_2, \|v'\|_2 \leq C^2 k^2 \ln(k/\eps),$ and $(M^T)^{-1} \cdot (v-v') \in \Z^k$ correspond to the same point $\xi,$
and therefore are not counted twice.}
\end{remark}

\paragraph{Overview of Analysis.}
We begin with a brief overview of the analysis.
First, we show (Lemma~\ref{lem:conc-pmd}) that at least $1-O(\eps)$ of the probability mass of $X$
lies in the ellipsoid with center $\mu$ 
and covariance matrix $\widetilde{\Sigma} = O(k\log(k/\eps))\Sigma+O(k\log(k/\eps))^2 I.$
Moreover, with high probability over the samples drawn in Step~1 of the algorithm, 
the estimates $\wh{\Sigma}$ and $\wh{\mu}$ will be good approximations of $\Sigma$ and $\mu$ (Lemma~\ref{lem:sample-mean-and-covariance}). 
By combining these two lemmas, we obtain (Corollary~\ref{cor:conc-pmd-sample}) that at least $1-O(\eps)$ of the probability mass of $X$
lies in the ellipsoid with center $\wh{\mu}$ 
and covariance matrix $\Sigma' = O(k\log(k/\eps)) \wh{\Sigma}+O(k\log(k/\eps))^2 I.$

By the above, and by our choice of the matrix $M \in \Z^{k \times k},$
we use linear-algebraic arguments to prove (Lemma~\ref{lem:fd})  
that almost all of the probability mass of $X$ lies in the set $\wh{\mu}+M(-1/2,1/2]^k$, 
a fundamental domain of the lattice $L=M\Z^k.$
This lemma is crucial because it implies that, to learn our PMD $X,$ 
it suffices to learn the random variable $X\pmod{L}.$
We do this by learning the Discrete Fourier transform of this distribution.
This step can be implemented efficiently due to the sparsity property of the DFT 
(Proposition~\ref{prop:ft-effective-support}):
except for points in $T$, the magnitude of the DFT will be very small.
Establishing the desired sparsity property for the DFT is the main technical contribution of this section.

Given the above, it it fairly easy to complete the analysis of correctness. 
For every point in $T$ we can learn the DFT up to absolute error $O(1/\sqrt{m}).$
Since the cardinality of $T$ is appropriately small, this implies that the total error
over $T$ is small.  The sparsity property of the DFT (Lemma~\ref{lem:dtv-final}) completes the proof.


\paragraph{Detailed Analysis.}
We now proceed with the detailed analysis of our algorithm.
We start by showing that PMDs are concentrated with high probability. 
More specifically, the following lemma shows that an unknown PMD 
$X$, with mean vector $\mu$ and covariance matrix $\Sigma$,
is effectively supported in an ellipsoid centered at $\mu$, 
whose principal axes are determined by the eigenvectors and eigenvalues of $\Sigma$ 
and the desired concentration probability:

\begin{lemma} \label{lem:conc-pmd}
Let $X$ be an $(n, k)$-PMD with mean vector $\mu$ and covariance matrix $\Sigma.$
For any $0< \eps <1,$ consider the positive-definite matrix 
$\widetilde \Sigma=k\ln(k/\eps)\Sigma+ k^2 \ln^2(k/\eps)I$.
Then, with probability at least $1-\eps/10$ over $X,$
we have that $(X-\mu)^T \cdot \widetilde \Sigma^{-1} \cdot (X-\mu) = O(1).$
\end{lemma}
\begin{proof}
Let $X = \sum_{i=1}^n X_i$, where the $X_i$'s are independent $k$-CRVs.
We can write $\mu = \E[X] = \sum_{i=1}^n \mu_i,$ where $\mu_i = \E [X_i].$
Note that for any unit vector $u \in \R^k$, $\|u\|_2 = 1$, 
we have that the scalar random variable $u \cdot (X-\mu)$ is a sum of independent, mean $0$, 
bounded random variables. Indeed, we have that $u \cdot (X-\mu) = \sum_{i=1}^n u \cdot (X_i - \mu_i),$
and that $\E[u \cdot (X_i - \mu_i)] = u \cdot (\E[X_i] - \mu_i) = 0.$ Moreover, we can write
$$|u \cdot (X_i-\mu_i)| \leq \|u\|_2 \cdot \|X_i-\mu_i\|_2 \leq\|X_i\|_2 + \|\mu_i\|_2 \leq 1 + \sqrt{k} \|\mu_i\|_1 \leq 2 \sqrt{k} \;,$$
where we used the Cauchy-Schwartz inequality twice, the triangle inequality, and 
the fact that a $k$-CRV $X_i$ with mean $\mu_i$ by definition satisfy $\|X_i\|_2 = 1,$ and $\|\mu_i\|_1=1.$

Let $\nu$ be the variance of $u \cdot (X-\mu).$
By Bernstein's inequality, we obtain that for $t=\sqrt{2\nu \ln (10k/\eps)} + 2 \sqrt{k} \ln (10k/\eps)$
it holds
\begin{equation} \label{eqn:bern2}
\Pr\left[|u \cdot (X-\mu)| > t \right] \leq  \exp\left(  - \frac{t^2/2}{\nu+ 2\sqrt{k}t/3}\right) \leq \frac{\eps}{10k} \;.
\end{equation}

Let $\Sigma$, the covariance matrix of $X,$
have an orthonormal eigenbasis $u_j \in \R^k$, $\|u_j\|_2 = 1,$ 
with corresponding eigenvalues $\nu_j$, $j \in [k].$ Since $\Sigma$ is positive-semidefinite,
we have that $\nu_j \ge 0$, for all $j \in [k].$
We consider the random variable $u_j\cdot (X-\mu)$.  
In addition to being a sum of independent, mean $0$, bounded random variables,
we claim that $\Var [u_j \cdot (X-\mu) ]  = \nu_j.$ First, it is clear that 
 $\Var [u_j \cdot (X-\mu)]  =  \E \left[(u_j \cdot (X-\mu))^2\right].$
Moreover, note that for any vector $v \in \R^k$,
we have that $\E[ \left(v \cdot (X-\mu)\right)^2] = v^T \cdot \Sigma \cdot v.$
For $v = u_j$,  we thus get  $\E \left[(u_j \cdot (X-\mu))^2\right] =u_j^T \cdot \Sigma \cdot u_j = \nu_j.$

Applying (\ref{eqn:bern2}) for $u_j\cdot (X-\mu),$ with  $t_j=\sqrt{2\nu_j \ln (10k/\eps)} + 2 \sqrt{k} \ln (10k/\eps),$
yields that for all $j \in [k]$ we have
$$\Pr\left[|u_j \cdot (X-\mu)| > t_j \right] \leq  \exp\left( - \frac{t_j^2/2}{\nu_j+ 2\sqrt{k}t_j/3}\right) \leq \frac{\eps}{10k}. $$
By a union bound, it follows that 
$$ \Pr\left[ \forall  j \in [k]:  \left|u_j \cdot (X-\mu)\right| \leq  t_j \right] \geq 1-\eps/10\;.$$
We condition on this event.

Since  $u_j$ and $\nu_j$ are the eigenvectors and eigenvalues of $\Sigma,$
we have that $\Sigma = U^T \cdot \mathrm{diag}\left(\nu_j\right) \cdot U,$ 
where $U$ has $j^{th}$ column $u_j.$
We can thus write $$\widetilde{\Sigma} =U^T \cdot \mathrm{diag}\left(k \nu_j \ln (k/\eps) + k^2\ln^2 (k/\eps)\right) \cdot U.$$
Therefore, we have:
\begin{align*}
(X-\mu)^T \cdot {\widetilde \Sigma}^{-1} \cdot  (X-\mu) & = \left\|\mathrm{diag}\left(\sqrt{k \nu_j \ln (k/\eps) + k^2\ln^2 (k/\eps)}\right)^{-1} \cdot U^T \cdot (X-\mu)\right\|_2^2 \\
								& \leq k \left\|\mathrm{diag}\left(\sqrt{k \nu_j \ln (k/\eps) + k^2\ln^2 (k/\eps)}\right)^{-1} \cdot U^T \cdot (X-\mu)\right\|_\infty^2 \\
								& = \left(\max_{j \in [k]} \frac{|u_j \cdot (X-\mu)|}{\sqrt{\nu_j \ln (k/\eps) + k \ln^2 (k/\eps)}}\right)^2 \\
								& = O(1) \;,
\end{align*}
where the last inequality follows from our conditioning, the definition of $t_j,$ and the elementary inequality
$\sqrt{a +b} \ge (\sqrt{a}+\sqrt{b})/\sqrt{2},$ $a, b \in \R_+.$
This completes the proof of Lemma~\ref{lem:conc-pmd}.
\end{proof}

Lemma~\ref{lem:conc-pmd} shows that an arbitrary $(n, k)$-PMD $X$ puts at least $1-\eps/10$ of its probability 
mass in the ellipsoid $\mathcal{E} = \{x \in \R^k : (x-\mu)^T \cdot (\widetilde{\Sigma})^{-1} \cdot (x-\mu) \le c \}, $
where $c>0$ is an appropriate universal constant. This is the ellipsoid centered at $\mu$, whose principal semiaxes
are parallel to the $u_j$'s, i.e., the eigenvectors of $\Sigma$, or equivalently of $(\widetilde{\Sigma})^{-1}.$ 
The length of the principal semiaxis corresponding to $u_j$ 
is determined by the corresponding eigenvalue of $(\widetilde{\Sigma})^{-1},$
and is equal to $c^{-1/2} \cdot \sqrt{k \nu_j \ln (k/\eps) + k^2\ln^2 (k/\eps)}.$

Note that this ellipsoid depends on the mean vector $\mu$ and covariance matrix $\Sigma,$
that are unknown to the algorithm. To obtain a bounding ellipsoid that is known to the algorithm, we will use
the following lemma (see Appendix~\ref{app:mean-cov-est} for the simple proof)
showing that $\wh{\mu}$ and $\wh{\Sigma}$ are good approximations to $\mu$ and $\Sigma$ respectively.

\begin{lemma} \label{lem:sample-mean-and-covariance}
With probability at least $19/20$ over the samples drawn in Step~1 of the algorithm, 
we have that $(\wh{\mu}-\mu)^T \cdot (\Sigma+I)^{-1} \cdot (\wh{\mu}-\mu) = O(1)$, and
$2(\Sigma+I) \succeq \wh{\Sigma}+I \succeq (\Sigma+I)/2.$
\end{lemma}

We also need to deal with the error introduced in the eigendecomposition of $\wh{\Sigma}$. 
Concretely, we factorize $\wh{\Sigma}$ as $V^T \Lambda V,$ for an orthogonal matrix $V$ and diagonal matrix $\Lambda.$ 
This factorization is necessarily inexact. 
By increasing the precision to which we learn $\wh{\Sigma}$ by a constant factor, 
we can still have $2(\Sigma+I) \succeq V^T \Lambda V+I \succeq (\Sigma+I)/2.$ 
We could redefine $\wh{\Sigma}$ in terms of our computed orthonormal eigenbasis, 
i.e., $\wh{\Sigma}: = V^T \Lambda V$. Thus, we may henceforth assume that the decomposition $\wh{\Sigma} = V^T \Lambda V$ is exact.

For the rest of this section, we will condition on the event that the statements of 
Lemma~\ref{lem:sample-mean-and-covariance} are satisfied.
By combining Lemmas~\ref{lem:conc-pmd} and~\ref{lem:sample-mean-and-covariance} ,
we show that we can get a known ellipsoid containing the effective support of $X,$
by replacing $\mu$ and $\Sigma$ in the definition of $\cal E$ by their sample versions.
More specifically, we have the following corollary:

\begin{corollary} \label{cor:conc-pmd-sample}
Let $\Sigma'=k\ln(k/\eps)\wh{\Sigma}+k^2\ln^2(k/\eps)I$.
Then, with probability at least $1-\eps/10$ over $X,$ we have that
$(X-\wh{\mu})^T \cdot (\Sigma')^{-1} \cdot (X-\wh{\mu}) = O(1).$
\end{corollary}
\begin{proof}
By Lemma \ref{lem:sample-mean-and-covariance}, it holds that $2(\Sigma+I) \succeq \wh{\Sigma}+I \succeq (\Sigma+I)/2$. Hence, we have that
$$  2\left(k\ln(k/\eps)\Sigma+ k^2 \ln^2(k/\eps)I\right)  \succeq k\ln(k/\eps)\wh{\Sigma}+k^2\ln^2(k/\eps)I  \succeq \frac{1}{2}\left(k\ln(k/\eps)\Sigma+ k^2 \ln^2(k/\eps)I\right) \;.$$
In terms of $\Sigma'$ and $\widetilde \Sigma$, this is $2 \widetilde \Sigma \succeq \Sigma' \succeq \frac{1}{2} \widetilde \Sigma$.
By standard results,
taking inverses reverses the positive semi-definite ordering (see e.g., Corollary 7.7.4 (a) in \cite{HornJohnson:85}). 
Hence,
$$2 {\widetilde \Sigma}^{-1} \succeq \Sigma'^{-1} \succeq \frac{1}{2} \widetilde \Sigma^{-1} \;.$$
Combining the above with Lemma \ref{lem:conc-pmd}, with probability  at least $1-\eps/10$ over $X$ we have that
\begin{equation} \label{eq:1-new} 
(X-\mu)^T \cdot \Sigma'^{-1} \cdot (X-\mu)  \leq 2 (X-\mu)^T \cdot \widetilde \Sigma^{-1} \cdot (X-\mu) = O(1) \;. 
\end{equation}
Since $\Sigma' \succeq \frac{1}{2} \widetilde \Sigma \succeq \Sigma + I$, and therefore 
 $(\Sigma')^{-1} \preceq (\Sigma + I)^{-1}$, Lemma \ref{lem:sample-mean-and-covariance} gives that
\begin{equation} \label{eq:2-new} 
(\wh{\mu}-\mu)^T \cdot \Sigma'^{-1} \cdot (\wh{\mu}-\mu) \leq (\wh{\mu}-\mu)^T \cdot (\Sigma+I)^{-1} \cdot (\wh{\mu}-\mu) = O(1) \;. 
\end{equation}
We then obtain:
\begin{align*}
(X-\wh{\mu})^T \cdot (\Sigma')^{-1} \cdot (X-\wh{\mu}) & = ((X-\mu) - (\wh{\mu}-\mu))^T \cdot (\Sigma')^{-1} \cdot ((X-\mu) - (\wh{\mu}-\mu)) \\
& = (X-\mu)^T \cdot \Sigma'^{-1} \cdot (X-\mu) + (\wh{\mu}-\mu)^T \cdot \Sigma'^{-1} \cdot (\wh{\mu}-\mu) \\
& - (X-\mu)^T \cdot \Sigma'^{-1} \cdot (\wh{\mu}-\mu) - (\wh{\mu}-\mu)^T \cdot \Sigma'^{-1} \cdot (X-\mu) \;.
\end{align*}
Equations (\ref{eq:1-new}) and (\ref{eq:2-new}) yield that the first two terms are $O(1)$. 
Since $ \Sigma'^{-1}$ is positive-definite, $x^T \Sigma'^{-1} y$ as a function of vectors $x,y \in \R^k$ is an inner product. 
So, we may apply the Cauchy-Schwartz inequality to bound each of the last two terms from above 
by $$\sqrt{((X-\mu)^T \cdot \Sigma'^{-1} \cdot (X-\mu)) \cdot ((\wh{\mu}-\mu)^T \cdot \Sigma'^{-1} \cdot (\wh{\mu}-\mu))}=O(1)\;,$$
where the last equality follows from (\ref{eq:1-new}) and (\ref{eq:2-new}).
This completes the proof of  Corollary~\ref{cor:conc-pmd-sample}.
\end{proof}

Corollary~\ref{cor:conc-pmd-sample} shows that our unknown PMD $X$ puts at least $1-\eps/10$ of its probability 
mass in the ellipsoid $\mathcal{E}' = \{x \in \R^k : (x-\wh{\mu})^T \cdot (\Sigma')^{-1} \cdot (x-\wh{\mu}) \le c \}, $
for an appropriate universal constant $c>0.$ This is the ellipsoid centered at $\wh{\mu}$, whose principal semiaxes
are parallel to the $v_j$'s, the eigenvectors of $\wh{\Sigma}$, and the length of the principal semiaxis parallel to $v_j$ is
equal to $c^{-1/2} \cdot \sqrt{k \lambda_j \ln (k/\eps) + k^2\ln^2 (k/\eps)}.$

Our next step is to to relate the ellipsoid $\mathcal{E}'$ to the integer matrix $M \in \Z^{k \times k}$ 
used in our algorithm. Let $M' \in \R^{k \times k}$ be the matrix with $j^{th}$ column
$C \left(\sqrt{k \ln(k/\eps)\lambda_j+k^2\ln^2(k/\eps)}\right)v_j,$ where $C>0$ 
is the constant in the algorithm statement.
The matrix $M \in \Z^{k \times k}$ is obtained by rounding each entry
of $M'$ to the closest integer point.
We note that the ellipsoid $\mathcal{E}'$ can be equivalently expressed as
$\mathcal{E}' = \{x \in \R^k : \|(M')^{-1} \cdot (x-\wh{\mu})\|_2 \leq 1/4 \}.$
Using the relation between $M$ and $M'$, we show that $\mathcal{E}'$ is enclosed
in the ellipsoid $\{x \in \R^k : \|(M)^{-1} \cdot (x-\wh{\mu})\|_2 < 1/2\},$
which is in turn enclosed in the parallelepiped with integer corner points
$\{x \in \R^k : \|(M)^{-1} \cdot (x-\wh{\mu})\|_{\infty} < 1/2\}.$
This parallelepiped is a fundamental domain of the lattice $L = M \Z^{k}.$ 
Formally, we have:

\begin{lemma} \label{lem:fd}
With probability at least $1-\epsilon/10$ over $X,$
we have that $X\in \wh{\mu}+M(-1/2,1/2]^k.$
\end{lemma}
\begin{proof}
Let $M' \in \R^{k \times k}$ be the matrix with columns $C \left(\sqrt{k \ln(k/\eps)\lambda_i+k^2\ln^2(k/\eps)}\right)v_i,$ where $C>0$
is the constant in the algorithm statement. 
Note that,
\begin{equation} \label{eq:MSigma}
M'(M')^T = C^2 k\ln(k/\eps)\wh{\Sigma}+C^2k^2\ln^2(k/\eps)I =  C^2 \Sigma' \;.
\end{equation}
For a large enough constant $C$, Corollary~\ref{cor:conc-pmd-sample}
implies that with probability at least $1-\eps/10,$
\begin{equation} \label{eqn:ellipsoid-M'}
\|(M')^{-1} \cdot (X-\wh{\mu})\|_2^2 = C^{-2} (X-\wh{\mu})^T \cdot (\Sigma')^{-1} \cdot (X-\wh{\mu}) \leq 1/16 \;.
\end{equation}
Note that the above is an equivalent description of the ellipsoid $\mathcal{E}'.$
Our lemma will follow from the following claim:
\begin{claim} \label{clm:rounding}
For any $x \in \R^k$, it holds
\begin{equation} \label{eq:ellipsoid-containment}
\|M^{-1} x\| _2 <  2 \|(M')^{-1} x\|_2 \textrm{ and } \|M^T x\|_2 <  2 \|(M')^T x\|_2  \;.
\end{equation}
\end{claim} 
\begin{proof}
By construction, $M$ and $M'$ differ by at most $1$ in each entry,
and therefore $M-M'$ has Frobenius norm (and, thus, induced $L_2$-norm) 
at most $k.$
For any $x \in \R^k$, we thus have that
$$\|(M')^T x\|_2 = \sqrt{x^T \cdot M' \cdot (M')^T \cdot x} \geq  \sqrt{x^T \cdot (C^2k^2I) \cdot x} >  2k \|x\|_2 \;,$$
and therefore
$$\|M^T x\|_2 \geq \|(M')^T x\|_2 - \|(M-M')^T\|_2 \|x\|_2 >  \frac{1}{2}\|(M')^T x\|_2 .$$
Similarly, we get $\|M^T x\|_2 \leq \|(M')^T x\|_2 + \|(M-M')^T\|_2 \|x\|_2 <  2 \|(M')^T x\|_2.$
In terms of the PSD ordering, we have:
\begin{equation} \label{eq:M-M'-PSD}
\frac{1}{4} M' (M')^T \prec M M^T \prec 4 M' (M')^T \;.
\end{equation}
Since $M' (M')^T \succeq I$,
both $M' (M')^T$ and $M M^T$ are positive-definite,
and so $M$ and $M'$ are invertible.
Taking inverses in Equation (\ref{eq:M-M'-PSD}) reverses the ordering, that is:
$$ \frac{1}{4} (M'^{-1})^T M'^{-1}  \prec (M^{-1})^T M^{-1}  \prec  4 (M'^{-1})^T M'^{-1} \;.$$
The claim now follows.
\end{proof}
Hence, Claim~\ref{clm:rounding} implies that with probability at least $1-\eps/10$, we have:
$$ \|M^{-1} (X-\wh{\mu}) \|_\infty \leq \|M^{-1} (X-\wh{\mu}) \|_2 < 2 \|(M')^{-1} (X-\wh{\mu}) \|_2 <  1/2 \;,$$
where the last inequality follows from (\ref{eqn:ellipsoid-M'}).
In other words,
with probability at least $1-\eps/10$, $X$ lies in $\wh{\mu}+M(-1/2,1/2]^k$, which was to be proved.
\end{proof}

Recall that $L$ denotes the lattice $M\Z^k$.
The above lemma implies that it is sufficient to learn the random variable $X\pmod{L}$.
To do this, we will learn its Discrete Fourier transform.
Let $L^{\ast}$ be the dual lattice to $L.$
Recall that the DFT of the PMD $\p$, with $X \sim \p$, 
is the function $\wh{\p}:L^{\ast}/\Z^k \rightarrow \C$ defined by
$\wh{\p}(\xi) = \E[e(\xi \cdot X)].$
Moreover, the probability that $X \pmod{L}$ attains a given value $x$ 
is given by the inverse DFT, namely
$$
\Pr\left[X \pmod{L} = x\right] = \frac{1}{|\det(M)|}\sum_{\xi \in L^{\ast}/\Z^k} \wh{\p}(\xi)e(-\xi\cdot x) \;.
$$
The main component of the analysis is the following proposition, establishing that 
the total contribution to the above sum coming from points $\xi \not\in T$ is small.
In particular, we prove the following:

\begin{proposition} \label{prop:ft-effective-support}
We have that
$\sum_{\xi \in (L^{\ast}/\Z^k)  \setminus T }|\wh{\p}(\xi)| < \eps/10.$
\end{proposition}

To prove this proposition, we will need a number of intermediate claims and lemmas.
We start with the following claim, showing that for every point $\xi \in \R^k,$
there exists an integer shift whose coordinates 
lie in an interval of length strictly less than $1$:

\begin{claim} \label{clm:pigeon-hole-interval}
For each $\xi\in \R^k$, there exists  $a \in \Z_+$ with $0\leq a\leq k$, and $b \in \Z^k$
such that $\xi-b \in \left[\frac{a}{k+1},\frac{a+k}{k+1}\right]^k$.
\end{claim}
\begin{proof}
Consider the $k$ fractional parts of the coordinates of $\xi$,
i.e., $\xi_i - \lfloor \xi_i \rfloor$, for $1 \leq i \leq k$.
Now consider the $k+1$ intervals
$I_{a'}=\left(\frac{a'-1}{k+1}, \frac{a'}{k+1}\right)$,
for $1 \leq a' \leq k+1$.
By the pigeonhole principle,
there is an $a'$ such that
$\xi_i - \lfloor \xi_i \rfloor \notin I_{a'}$, for all $i$, $1 \le i \le k.$
We define $a=a'$ when $a'<k+1$, and $a=0$ when $a'=k+1$.

For any $i$, with $1 \le i \le k$,
since $\xi_i - \lfloor \xi_i \rfloor \notin I_{a'}$, we have that
$\xi_i - \lfloor \xi_i \rfloor \in \left[0,\frac{a-1}{k+1}\right] \cup \left[\frac{a}{k+1},1\right]$
(taking the first interval to be empty if $a=0$).
Hence, by setting one of $b_i=\lfloor \xi_i \rfloor $,
or $b_i=\lfloor \xi_i \rfloor-1$, we get $\xi_i-b_i \in \left[\frac{a}{k+1},\frac{a+k}{k+1}\right].$
This completes the proof.
\end{proof}


The following lemma gives a ``Gaussian decay'' upper bound 
on the magnitude of the DFT, at points $\xi$ whose coordinates lie in an interval of length less than $1.$
Roughly speaking, the proof of Proposition~\ref{prop:ft-effective-support} proceeds by applying this lemma for all $\xi \notin T.$

\begin{lemma} \label{lem:gaussian-bound-from-interval}
Fix $\delta \in (0, 1)$. Suppose that $\xi \in \R^k$
has coordinates lying in an interval $I$ of length $1-\delta$.  Then,
$|\wh{\p}(\xi)| = \exp(-\Omega(\delta^2 \xi^T \cdot \Sigma \cdot \xi)).$
\end{lemma}
\begin{proof}
Since $\p$ is a PMD, we have $X = \sum_{i=1}^n X_i$, where $X_i \sim \p_i$ for independent $k$-CRV's $\p_i$, we have that
$ |\wh{\p}(\xi)| = \prod_{i=1}^n |\wh{\p_i}(\xi)|.$
Note also that $\xi^T \cdot \Sigma \cdot \xi = \var[\xi\cdot X] = \sum_{i=1}^n \var[\xi\cdot X_i].$
It therefore suffices to show that for each $i \in [n]$ it holds
$$|\wh{\p_i}(\xi)| = \exp(-\Omega(\delta^2 \var[\xi\cdot X_i]))\;.$$
Let $X'_i \sim \p_i$ be an independent copy of $X_i$,
and $Y_i = X_i-X'_i$.
We note that $\var[\xi\cdot X_i] = (1/2)\E[(\xi\cdot Y_i)^2]$,
and that $|\wh{\p_i}(\xi)|^2 = \E[e(\xi\cdot Y_i)]$.
Since $Y_i$ is a symmetric random variable, we have that
$$|\wh{\p_i}(\xi)|^2 = \E[e(\xi\cdot Y_i)] = \E[\cos(2\pi \xi\cdot Y_i)] \;.$$
We will need the following technical claim:
\begin{claim} \label{claim:cos-claim}
Fix $0< \delta < 1.$
For all $x \in \R$ with $|x| \le 1-\delta$, it holds
$1-\cos(2 \pi x) \geq \delta^2x^2.$
\end{claim}
\begin{proof}
When $0 \leq |x| \leq 1/4$, $\sin(2 \pi x)$ is concave,
since its second derivative is $-4 \pi^2\sin(2 \pi x) \leq 0$.
So, we have $\sin(2 \pi x) \geq (1 - 4x) \sin(0) + 4x \sin (\pi/2)=4x.$
Integrating the latter inequality,
we obtain that $(1-\cos(2 \pi x))/2 \pi \geq 2 x^2$, i.e.,
$(1-\cos(2 \pi x)) \geq 4 \pi x^2$.
Thus, for $0 \leq |x| \leq 1/4$, we have $1-\cos(2 \pi x)) \geq 4 \pi x^2 \geq \delta^2 x^2$.

When $1/4 \leq |x| \leq 3/4$, we have $(1-\cos(2 \pi x)) \geq 1 \geq \delta^2 x^2$.
Finally, when $3/4 \leq |x| \leq 1 - \delta$,
we have $0 \leq 1-|x| \leq \delta \leq 1/4$,
and therefore $1-\cos(2 \pi x)=1-\cos(2 \pi (1-|x|)) \geq 1- \cos(2 \pi \delta) \geq 4 \pi \delta^2 \geq \delta^2 x^2$.
This establishes the proof of the claim.
\end{proof}
\noindent Since $\xi \cdot Y_i$ is by assumption supported on the interval $[-1+\delta,1-\delta]$,
we have that $|\wh{\p_i}(\xi)|^2$ is
$$
\E[\cos(2\pi \xi\cdot Y_i)] = \E[1-\Omega(\delta^2 (\xi \cdot Y_i)^2)]
\leq \exp(-\Omega(\delta^2 \E[(\xi \cdot Y_i)^2])) = \exp(-\Omega(\delta^2 \var[\xi\cdot X_i])) \;.
$$
This completes the proof of Lemma~\ref{lem:gaussian-bound-from-interval}.
\end{proof}

We are now ready to prove the following crucial lemma, 
which shows that the DFT of $\p$ is effectively supported on the set $T.$

\begin{lemma} \label{lem:dft-ub}
For integers $0\leq a\leq k$, we have that
$$
\sum_{\xi\in L^{\ast}\cap \left[\frac{a}{k+1},\frac{a+k}{k+1}\right]^k  \setminus (T+\Z^k)} |\wh{\p}(\xi)| < \frac{\eps}{10(k+1)} \;.
$$
\end{lemma}

We start by providing a brief overview of the proof.
First note that Claim~\ref{clm:pigeon-hole-interval} and Lemma~\ref{lem:gaussian-bound-from-interval} 
together imply that for $\xi \in [a/(k+1),(a+k)/(k+1)]^k$,  if $\xi^T \cdot \Sigma \cdot \xi \ge k^2 \log(1/\eps'),$ 
then $|\wh{\p}(\xi)| \le \eps'$, for any $\eps'>0.$ Observe that the set $\{ \xi \in \R^k: \xi^T \cdot \Sigma \cdot  \xi \leq k^2 \log(1/\eps') \}$ 
is not an ellipsoid, because $\Sigma$ is singular. However, by using the fact that $M$ and $M'$ are close to each other, 
-- more specifically, using ingredients from the proof of Lemma~\ref{lem:fd} -- 
we are able to bound its intersection with $[a/(k+1),(a+k)/(k+1)]^k$ 
by an appropriate ellipsoid of the form $\{ \xi^T \cdot (M \cdot  M^T) \cdot \xi \le r \}.$

The gain here is that $M^T \xi \in \Z^k.$ 
This allows us to provide an upper bound on the cardinality of 
the set of lattice points in $L^{\ast}$ in one of these ellipsoids. 
Note that, in terms of $v=M^T\xi$, these are the integer points that lie in a sphere of some radius 
$r > r_0 = C^2 k^2 \ln(k/\epsilon).$ Now, if we consider the set $2^{t} r_0 \leq \xi^T \cdot (M M^T) \cdot \xi < 2^{t+1}r_0,$ 
we have both an upper bound on the magnitude of the DFT and on the number of integer points in the set. 
By summing over all values of $t \ge 0$, we get an upper bound on the error coming from points outside of $T.$

\begin{proof}[Proof of Lemma~\ref{lem:dft-ub}.]
Since $\xi \notin T+\Z^k$, we have that $\|M^T \xi\|_2 > C^2 k^2 \ln(k/\eps).$
Since the coordinates of $\xi$ lie in $\left[\frac{a}{k+1},\frac{a+k}{k+1}\right]$,
an interval of length $1- 1/(k+1)$, we may apply Lemma~\ref{lem:gaussian-bound-from-interval} to obtain:
\begin{align*}
|\wh{\p}(\xi)| & = \exp\left(-\Omega\left(k^{-2} \xi^T \cdot \Sigma \cdot \xi\right)\right) \\
& = \exp\left(-\Omega\left(k^{-2} \xi^T \cdot (\wh{\Sigma}-I) \cdot \xi\right)\right)  \tag*{(by Lemma~\ref{lem:sample-mean-and-covariance})}\\
& = \exp\left(-\Omega\left(k^{-2} \left( \frac{\xi^T M' (M')^T \xi}{C^2k\log(k/\eps)} -k \log(k/\eps) \|\xi\|_2^2\right)\right)\right) \tag*{(by Equation (\ref{eq:MSigma}))}\\
& = \exp\left(-\Omega\left(k^{-2} \left(\frac{\|(M')^T \xi\|_2^2}{C^2k\log(k/\eps)}-k^2 \log(k/\eps) \|\xi\|_\infty^2\right)\right)\right) \\
 & = \exp\left(-\Omega\left(k^{-2} \left(\frac{\|M^T \xi\|_2^2}{C^2 k\log(k/\eps)}-k^2 \log(k/\eps) \|\xi\|_\infty^2\right)\right)\right) \tag*{(by Equation (\ref{eq:ellipsoid-containment}))}\\
& = \exp\left(-\Omega\left(C^{-2} k^{-3} \log^{-1}(k/\eps) \|M^T \xi\|_2^2 - \log(k/\eps)\right)\right) \tag*{(since $\|\xi\|_\infty \leq 2$)} \\
& = \exp\left(-\Omega\left(C^{-2} k^{-3} \log^{-1}(k/\eps) \|M^T \xi\|_2^2\right)\right) \tag*{(since $\|M^T\xi\|_2 > C^2 k^2 \ln(k/\eps))$} \;.
\end{align*}
Next note that for $\xi\in L^{\ast}$ we have that $M^T\xi \in \Z^k$. Thus, letting $v=M^T\xi$, it suffices to show that
\begin{equation} \label{ineq:outside-sphere-sum}
\sum_{v\in \Z^k, \|v\|_2 \geq C^2 k^2 \ln(k/\eps)} \exp\left(-\Omega\left(C^{-2} k^{-3} \log^{-1}(k/\eps) \|v\|_2^2\right)\right) < \frac{\eps}{10(k+1)} \;.
\end{equation}
Although the integer points in the above sum are not in the sphere $ \|v\|_2 \leq C^2 k^2 \ln(k/\eps),$
they lie in some sphere $\|v\|_2 \leq 2^{t+1} C^2 k^2 \ln(k/\eps)$, for some integer $t >0$.
The number of integral points in one of these spheres is less than that of the appropriate enclosing cube.
Namely, we have that
\begin{eqnarray}
&& \#\left\{ v\in \Z^k, \|v\|_2 \leq 2^{t+1} C^2 k^2 \ln(k/\eps) \right\} \nonumber \\
&\leq & \# \left\{ v\in \Z^k, \|v\|_\infty \leq 2^{t+1} C^2 k^2 \ln(k/\eps) \right\} \nonumber \\
&= & \left(1+2\lfloor2^{t+1} C^2 k^2 \ln(k/\eps)\rfloor \right)^k \;. \label{eqn:cube-points}
\end{eqnarray}
Inequality (\ref{ineq:outside-sphere-sum}) is obtained by bounding the LHS from above as follows:
\begin{align*}
& \sum_{t=0}^\infty \exp\left(-\Omega(C^{-2} k^{-3} \log^{-1}(k/\eps) (2^tC^2 k^2 \ln(k/\eps))^2)\right) \cdot \# \left\{v\in\Z^k,\|v\|_2 \leq 2^{t+1} C^2 k^2 \ln(k/\eps) \right\} \\
 = & \sum_{t=0}^\infty \exp\left(-\Omega(C k \log(k/\eps) 4^t)\right) \cdot \# \left\{v\in\Z^k,\|v\|_2 \leq 2^{t+1} C^2 k^2 \ln(k/\eps) \right\} \\
 \leq & \sum_{t=0}^\infty \exp\left(-\Omega(C k \log(k/\eps) 4^t)\right) \cdot \left(2^{t+2} C^2 k^2 \log(k/\eps)\right)^k   \tag*{(by Equation (\ref{eqn:cube-points}))}\\
 = & \sum_{t=0}^\infty \exp\left(-\Omega(C k \log(k/\eps) 4^t)\right) \exp\left(O(k(t+\log k + \log \log(k/\eps)))\right) \\
 \leq & \exp\left(-\ln((k+1)/10 \eps)\right) \cdot  \sum_{t=0}^\infty  \exp(- k{\sqrt{C}}(4^t-t))) \\
\leq & \frac{\eps^2}{10(k+1)} \cdot \sum_{t=0}^\infty \exp(-{\sqrt{C}}(4^t-t))\\
< & \frac{\eps}{10(k+1)} \;.
\end{align*}
This completes the proof of Lemma~\ref{lem:dft-ub}.
\end{proof}

We are now prepared to prove Proposition~ \ref{prop:ft-effective-support}.

\begin{proof}[Proof of Proposition \ref{prop:ft-effective-support}]
Let $T_a$ be the set of points $\xi \in \R^k$ which have a lift with all coordinates
in the interval $\left[\frac{a}{k+1},\frac{a+k}{k+1}\right]^k,$
for some integer $0 \leq a \leq k$.
By Claim~\ref{clm:pigeon-hole-interval}, we have that $\bigcup_a T_a =(L^{\ast}/\Z^k)$.
By Lemma~\ref{lem:dft-ub}, for all $0 \leq a \leq k$,
$$\sum_{\xi \in T_a \setminus T }|\wh{\p}(\xi)| < \frac{\eps}{10(k+1)} \;,$$
and so we have:
$$\sum_{\xi \in (L^{\ast}/\Z^k) \setminus T }|\wh{\p}(\xi)|
\leq \sum_{a=0}^k \sum_{\xi \in T_a \setminus T }|\wh{\p}(\xi)| < \eps/10 \;.$$
\end{proof}

Our next simple lemma states that the empirical DFT is a good approximation to the true DFT
on the set $T.$

\begin{lemma} \label{lem:emp-dft}
Letting $m=(C^5k^4\ln^{2}(k/\eps))^{k}/\eps^2$,
with $19/20$ probability over the choice of $m$ samples in Step~5,
we have that
$
\sum_{\xi \in T} | \wh{\h}(\xi) - \wh{\p}(\xi)| < \eps/10.
$
\end{lemma}
\begin{proof}
For any given $\xi \in T,$
we note that $\wh{\h}(\xi)$ is the average of $m$ samples from $e(\xi\cdot X)$,
a random variable whose distribution has mean $\wh{\p}(\xi)$ and variance at most $O(1)$.
Therefore, we have that
$$
\E[| \wh{\h}(\xi) - \wh{\p}(\xi)|] \leq O(1)/\sqrt{m}.
$$
Summing over $\xi \in T,$ and noting that $|T|\leq O(C^2 k^2 \log(k/\eps))^k$,
we get that the expectation of the quantity in question is less than $\eps/400.$
Markov's inequality completes the argument.
\end{proof}

Finally, we bound from above the total variation distance between $\p$ and $\h.$
\begin{lemma} \label{lem:dtv-final}
Assuming that the conclusion of the previous lemma holds,
then for any $x\in \Z^k/L$ we have that
$$
\left| \Pr[X\equiv x\pmod{L}] - \frac{1}{|\det(M)|}\sum_{\xi\in T}\wh{\h}(\xi)e(-\xi\cdot x) \right| \leq \frac{\eps}{5|\det(M)|}.
$$
\end{lemma}
\begin{proof}
We note that
\begin{align*}
& \left|\Pr[X\equiv x\pmod{L}] - \frac{1}{|\det(M)|}\sum_{\xi\in T}\wh{\h}(\xi)e(-\xi\cdot x) \right|\\
= & \left|\frac{1}{|\det(M)|} \sum_{\xi\in L^{\ast}/\Z^k} \wh{\p}(\xi) e(-\xi\cdot x) - \frac{1}{\det(M)} \sum_{\xi\in T}\wh{\h}(\xi)e(-\xi\cdot x) \right|\\
\leq & \frac{1}{|\det(M)|} \sum_{\xi\in L^{\ast}/\Z^k, \xi\not \in T} |\wh{\p}(\xi)| - \frac{1}{|\det(M)|}\sum_{\xi \in T}|\wh{\p}(\xi)-\wh{\h}(\xi)|\\
\leq & \frac{\eps}{5|\det(M)|} \;,
\end{align*}
where the last line follows from Proposition~\ref{prop:ft-effective-support} and Lemma~\ref{lem:emp-dft}.
\end{proof}

It follows that, for each $x\in \wh{\mu}+M(-1/2,1/2]^k$, our hypothesis pseudo-distribution $\h(x)$
equals the probability that $X\equiv x\pmod{L}$ plus an error of at most $\frac{\eps}{5|\det(M)|}.$
In other words, the pseudo-distribution defined by $\h \pmod{L}$ differs 
from $X\pmod{L}$ by at most $\left(\frac{\eps}{5|\det(M)|} \right)|\Z^k/L| = \eps/5.$
On the other hand, letting $X' \sim \p'$ be obtained by moving 
a sample from $X$ to its unique representative modulo $L$ lying in $\wh{\mu}+M(-1/2,1/2]^k$,
we have that $X=X'$ with probability at least $1-\eps/10.$ Therefore, $\dtv(\p, \p')\leq \eps/10.$
Note that $X\pmod{L}=X'\pmod{L}$, and so $\dtv(\h\pmod{L}, \p' \pmod{L})<\eps/5.$ Moreover,
$\h$ and $\p'$ are both supported on the same fundamental domain of $L,$
and hence $\dtv(\h,\p')=\dtv(\h\pmod{L},\p' \pmod{L})<\eps/5.$
Therefore, assuming that the above high probability events hold, we have that
$\dtv(\h,\p)\leq \dtv(\h,\p')+\dtv(\p, \p') \leq 3\eps/10.$

This completes the analysis and the proof of Theorem~\ref{thm:alg-dft}.

\subsection{An Efficient Sampler for our Hypothesis} \label{ssec:sampler}

The learning algorithm of Section~\ref{ssec:algo-dft} outputs a succinct description
of the hypothesis pseudo-distribution $\h$, via its DFT. This immediately provides us with an efficient
evaluation oracle for $\h,$ i.e., an $\eps$-evaluation oracle for our target PMD $\p.$
The running time of this oracle is linear in the size of $T,$  the effective support of the DFT.

Note that we can explicitly output the hypothesis $\h$
by computing the inverse DFT at all the points
of the support of $\h.$ However, in contrast to the effective support of
$\wh{\h},$ the support of $\h$ can be large, and this explicit description would not lead to a computationally efficient algorithm.
In this subsection, we show how to efficiently obtain an $\eps$-sampler for our unknown PMD $\p,$ using the
DFT representation of $\h$ as a black-box.
In particular, starting with the DFT of an accurate hypothesis $\h,$ represented via its DFT,
we show how to efficiently obtain an $\eps$-sampler for the unknown target distribution.
We remark that the efficient procedure of this subsection is not restricted to PMDs,
but is more general, applying to all  discrete distributions with an approximately sparse DFT (over any dimension)
for which an efficient oracle for the DFT is available.

In particular, we prove the following theorem:

\begin{theorem} \label{thm:sampler}
Let $M \in \Z^{k \times k}$, $m \in \R^k$, and $S=m+{M(-1/2,1/2]}^k \cap \Z^k.$
Let $\h: S \to \R$ be a pseudo-distribution succinctly represented via its DFT (modulo $M$), $\wh{\h}$, which is supported on a set $T,$
i.e., $\h(x)= (1/|\det(M)|) \cdot \sum_{\xi \in T} e(- \xi \cdot x) \widehat{\h}(\xi),$
for $x \in S,$ with $\mathbf{0} \in T$ and $\widehat{\h}(\mathbf{0})=1.$
Suppose that there exists a distribution $\p$ 
with $\dtv(\h, \p) \leq \eps/3.$
Then, there exists an $\eps$-sampler for $\p,$ i.e., a sampler for a distribution $\q$ such that
$\dtv(\p,\q) \leq \eps,$ running in time $O(\log(|\det (M)|) \log(|\det(M)|/\epsilon) \cdot |T| \cdot \poly(k)).$
\end{theorem}

We remark that the $\eps$-sampler in the above theorem statement
can be described as a randomized algorithm that takes as input $M$, $T$, $\widehat{\h}(\xi),$ for $\xi \in T,$ 
and the Smith normal form decomposition of $M$(see Lemma \ref{lem:smith}).

We start by observing that our main learning result, Theorem~\ref{thm:learn-pmd}, follows by combining
Theorem~\ref{thm:alg-dft} with Theorem~\ref{thm:sampler}.
Indeed, note that the matrix $M$ in the definition of our PMD algorithm in Section~\ref{ssec:algo-dft} satisfies
$|\det(M)| \le {O(\sqrt{\det(\widetilde{\Sigma})})}  \leq \left(n k \log(1/\eps)\right)^{O(k)}$. 
Also recall that $|T|= O(k^{2k} \log^k(k/\eps))$. Since $M$ has largest entry $n,$ 
by Lemma~\ref{lem:smith}, we can compute its Smith normal form decomposition in time $\poly(k) \log n.$
Hence, for the case of PMDs, we obtain the following corollary, establishing Theorem~\ref{thm:learn-pmd}:

\begin{corollary} \label{cor:pmd-sampler}
For all $n, k \in \Z_+$ and $\eps>0,$
there is an algorithm with the following
performance guarantee: Let $\p$ be an unknown $(n, k)$-PMD.
The algorithm uses $O\left(k^{4k}\log^{2k}(k/\eps)/\eps^2\right)$
samples from $\p$,  runs in time $O\left(k^{6k} \log^{3k}(k/\eps)/\eps^2 \cdot \log n\right)$,
and with probability at least $9/10$ outputs an $\eps$-sampler for $\p.$
This $\eps$-sampler runs (i.e., produces a sample) in time $\poly(k) O( k^{2k} \log^{k+1}(k/\eps)) \cdot \log^2 n.$
\end{corollary}


\medskip

This section is devoted to the proof of Theorem~\ref{thm:sampler}.
We first handle the case of one-dimensional distributions, and then appropriately reduce the high-dimensional
case to the one-dimensional.

\begin{remark}
{\em We remark that the assumption that $\widehat{\h}(\mathbf{0})=1$ in our theorem statement,
ensures that $\sum_{x \in S} \h(x) = 1,$ and so,
for any distribution $\p$ over $S$ the  total variational distance $\dtv(\h,\p) \eqdef \frac{1}{2} \sum_{x \in S} |\h(x)-\p(x)|$
is well behaved in the sense that $\dtv(\h,\p) = \sum_{x:\p(x) > \h(x)} (\p(x)-\h(x))= \sum_{x:\p(x) < \h(x)} (\h(x)-\p(x)).$
This fact will be useful in the correctness of our sampler.}
\end{remark}


We start by providing some high-level intuition.
Roughly speaking, we obtain the desired sampler by considering an appropriate definition
of a Cumulative Distribution Function (CDF) corresponding to $\h.$
For the $1$-dimensional case (i.e., the case $k=1$ in our theorem statement),
the definition of the CDF is clear, and our sampler proceeds as follows:
We use the DFT to obtain a closed form expression for the CDF of $\h,$
and then we query the CDF using an appropriate binary search procedure to sample from the distribution.
One subtle point is that $\h(x)$ is a pseudo-distribution, i.e. it is not necessarily non-negative at all points.
Our analysis shows that this does not pose any problems with correctness, by using the aforementioned remark.

For the case of two or more dimensions ($k \geq 2$), we essentially provide
a computationally efficient reduction to the $1$-dimensional case.
In particular, we exploit the fact that the underlying domain is discrete,
to define an efficiently computable bijection from the domain to the integers,
and consider the corresponding $1$-dimensional CDF.
To achieve this, we efficiently decompose the integer matrix $M \in \Z^{k \times k}$ using a version
of the Smith Normal Form, effectively reducing to the case that $M$ is diagonal.
For the diagonal case, we can intuitively treat the dimensions independently, using the lexicographic ordering.

\medskip

Our first lemma handles the $1$-dimensional case, assuming the existence
of an efficient oracle for the CDF:

\begin{lemma} \label{lem:bs}
Given a pseudo-distribution $\h$ supported on $[a,b]  \cap \Z$, $a,b \in \Z,$
with CDF $c_{\h}(x)=\sum_{i: a \leq i \leq x}  \h(i)$ 
(which satisfies $c_{\h}(b)=1$), and oracle access to a function $c(x)$ so that $|c(x)-c_{\h}(x)| < \epsilon/(10(b-a+1))$ for all $x,$ 
we have the following:
If there is a distribution $\p$ with $\dtv(\h,\p) \leq \eps/3,$ there is a sampler for a distribution $\q$ with $\dtv(\p,\q) \leq \eps,$
using $O(\log (b+1-a) + \log (1/\eps))$ uniform random bits as input,
and running in time $O((D+1)(\log (b+1-a)) + \log (1/\eps)),$
where $D$ is the running time of evaluating the CDF $c(x).$
\end{lemma}

\begin{proof}
We begin our analysis by producing an algorithm that works when we are able to exactly sample $c_{\h}(x)$.

We can compute an inverse to the CDF $d_{\h}:[0,1] \rightarrow [a,b] \cap \Z$, at $y \in [0,1]$,
using binary search, as follows:
\begin{enumerate}
\item We have an interval $[a',b']$, initially $[a-1,b]$, with $c_{\h}(a') \leq y \leq c_{\h}(b')$ and $c_{\h}(a') < c_{\h}(b').$
\item If $b'-a'=1$, output $d_{\h}(y)=b'.$
\item Otherwise, find the midpoint $c'=\lfloor(a'+b')/2 \rfloor.$
\item \label{step:last} If $c_{\h}(a') < c_{\h}(c')$ and $y \leq c_{\h}(c')$,
repeat with $[a',c']$; else repeat with $[c',b]$.
\end{enumerate}
The function $d_{\h}$ can be thought of as some kind of inverse to the CDF $c_{\h}:[a-1,b] \cap \Z \rightarrow [0,1]$ in the following sense:
\begin{claim} \label{clm:invariant}
The function $d_{\h}$ satisfies: For any $y \in [0, 1]$,  it holds
$c_{\h}(d_{\h}(y)-1) \leq y \leq c_{\h}(d_{\h}(y))$ and $c_{\h}(d_{\h}(y)-1) < c_{\h}(d_{\h}(y)).$
\end{claim}
\begin{proof}
Note that if we don't have $c_{\h}(a') < c_{\h}(c')$ and $y \leq c_{\h}(c')$, then $c_{\h}(c') < y \leq c_{\h}(b')$.
So, Step \ref{step:last} gives an interval $[a',b']$ which satisfies $c_{\h}(a') \leq y \leq c_{\h}(b')$ and $c_{\h}(a') < c_{\h}(b')$.
The initial interval $[a-1,b]$ satisfies these conditions since $c_{\h}(a-1) =0$ and $c_{\h}(b)=1$.
By induction, all $[a',b']$ in the execution of the above algorithm have $c_{\h}(a') \leq y \leq c_{\h}(b')$ and $c_{\h}(a') < c_{\h}(b')$.
Since this is impossible if $a'=b'$, and Step \ref{step:last} always recurses on a shorter interval, we eventually have $b'-a'=1$.
Then, the conditions $c_{\h}(a') \leq y \leq c_{\h}(b')$ and $c_{\h}(a') < c_{\h}(b')$ give the claim.
\end{proof}

 Computing $d_{\h}(y)$ requires $O(\log (b-a+1))$ evaluations of $c_{\h}$,
 and $O(\log (b-a+1))$  comparisons of $y.$ For the rest of this proof, we will use $n = b-a+1$ to denote the support size.

Consider the random variable $d_{\h}(Y)$, for $Y$ uniformly distributed in $[0,1]$, whose distribution we will call $\q'$.
When $d_{\h}(Y)=x$, we have $c_{\h}(x-1) \leq Y \leq c_{\h}(x)$, and so when $\q'(x) > 0$,
we have $\q'(x) \leq \Pr\left[c_{\h}(x-1) \leq Y \leq c_{\h}(x)\right] = c_{\h}(x)-c_{\h}(x-1) = \h(x)$.
So, when $\h(x) > 0$, we have $\h(x) \geq \q'(x)$.
But when $\h(x) \leq 0$, we have $\q'(x)=0$, since then $c_{\h}(x) < c_{\h}(x-1)$ and no $y$ has $c_{\h}(x-1) \leq y \leq c_{\h}(x)$.
So, 
we have $\dtv(\q',\h) = \sum_{x:\h(x) < 0} -\h(x) \leq \dtv(\h,\p) \leq \eps/3$.

We now show how to effectively sample from $\q'$.
The issue is how to simulate a sample from the uniform distribution on $[0,1]$ with uniform random bits.
We do this by flipping coins for the bits of $Y$ lazily. 
We note that we will only need to know more than $m$
 bits of $Y$ if $Y$ is within $2^{-m}$ of one of the values of $c_{\h}(x)$ for some $x.$ 
 By a union bound, this happens with probability at most $n2^{-m}$ over the choice of $Y.$ 
 Therefore, for $m > \log_2(10n/\epsilon)$, the probability that this will happen is at most $\epsilon/10$ and can be ignored.

Therefore, the random variable $d_{\h}(Y')$, for $Y'$ uniformly distributed on the multiples of $2^{-r}$ in $[0,1)$ for $r = O(\log n + \log (1/\eps)),$
has distribution $\q'$ that satisfies $\dtv(\q,\q') \leq \eps/10.$ This means that $\dtv(\p,\q') \leq \dtv(\p,\h) + \dtv(\h,\q) + \dtv(\q,\q') \leq 9\eps/10.$
That is, we obtain an $\eps$-sampler that uses $O(\log n + \log (1/\eps))$ coin flips,
$O(\log n)$ calls to $c_{\h}(x)$, and has the desired running time.

We now need to show how this can be simulated without access to $c_{\h},$ 
and instead only having access to its approximation $c(x).$ 
The modification required is rather straightforward. 
Essentially, we can run the same algorithm using $c(x)$ in place of $c_{\h}(x).$ 
We note that all comparisons with $Y$ will produce the same result, 
unless the chosen $Y$ is between $c(x)$ and $c_{\h}(x)$ for some value of $x.$ 
Observe that because of our bounds on their difference, the probability 
of this occurring for any given value of $x$ is at most $\epsilon/(10 n).$
By a union bound, the probability of it occurring for any $x$ is at most $\epsilon/10.$ 
Thus, with probability at least $1-\epsilon/10$ our algorithm returns the same result 
that it would have had it had access to $c_{\h}(x)$ instead of $c(x).$
This implies that the variable sampled by this algorithm has variation distance at most $\epsilon/10$ 
from what would have been sampled by our other algorithm. Therefore, this algorithm samples a $\q$ with $\dtv(\p,\q)\leq \epsilon.$
\end{proof}

We next show that we can efficiently compute an appropriate CDF using the DFT.
For the $1$-dimensional case, this follows easily via a closed form expression.
For the high-dimensional case, we first obtain a closed form expression for the case
that the matrix $M$ is diagonal. We then reduce the general case to the diagonal case,
by using a Smith normal form decomposition.

\begin{proposition} \label{prop:CDF}
{For $\h$ as in Theorem \ref{thm:sampler},} we have the following:
\begin{enumerate}
\item[(i)] If $k=1$, there is an algorithm to compute the CDF
$c_{\h}:[a,b] \cap \Z \rightarrow [0,1]$
with $c_{\h}(x)=\sum_{i: a \leq i \leq x}  \h(i)$ to any precision $\delta>0$,
where $a = m-\lceil M/2 \rceil + 1$ and $b=m+ \lfloor M/2 \rfloor$, $M \in \Z_+$.
The algorithm runs in time $O(|T|\log(1/\delta)).$

\item[(ii)]  If $M \in \Z^{k \times k}$ is diagonal,
there is an algorithm which computes the CDF to any precision $\delta>0$
under the lexicographic ordering $\leq_{\mathrm{lex}}$, i.e.,
$c_{\h}(x)=\sum_{y \in T: y \leq_{\mathrm{lex}}  x} \h(y)$.
The algorithm runs in time $O(k^2|T|\log(1/\delta)).$

\item[(iii)]  For any $M \in \Z^{k \times k}$, there is an explicit ordering $\leq_{g}$
for which we can compute the CDF  $c_{\h}(x)=\sum_{y \in T: y \leq_{g}  x} \h(y)$ to any precision $\delta>0$.
This computation can be done in time $O(k^2|T|\log(1/\delta)+\poly(k)).$
\end{enumerate}
In cases (ii) and (iii), we can also compute the embedding of the corresponding ordering
onto the integers $[|\det M|] = \{1,2,...,|\det M|\}$,
i.e., we can give a monotone bijection $f: S \rightarrow [|\det M|]$
for which we can efficiently compute $f$ and $f^{-1}$
(i.e., with the same running time bound we give for computing $c_{\h}(x)$).
\end{proposition}
\begin{proof}
Recall that the PMF of $\h$ at $x \in S$ is given by the inverse DFT:
\begin{equation}  \label{eq:inverse-DFT-H}
\h(x) = \frac{1}{|\det M|} \sum_{\xi \in T} e(- \xi \cdot x) \widehat{\h}(\xi) \;.
\end{equation}

\paragraph{Proof of (i):}
For (i), the CDF is given by:
\begin{align*}
c_{\h}(x) & = \frac{1}{M} \sum_{i: a \leq i \leq x} \sum_{\xi \in T} e(-\xi x) \widehat{\h}(\xi)\\
		& = \frac{1}{M} \sum_{\xi \in T} \widehat{\h}(\xi) \sum_{i: a \leq i \leq x} e(-\xi x)
\end{align*}
When $\xi \not= 0$, the term $\sum_{i: a \leq i \leq x} e(-\xi \cdot x)$ is a geometric series.
By standard results on its sum, we have:
$$\sum_{i: a \leq i \leq x} e(-\xi x) = \frac{e(-\xi a) - e(-\xi (x+1))}{1-e(-\xi)} \;.$$
When $\xi = 0$, $e(-\xi)=1$, and we get $\sum_{a \leq i \leq x} e(-\xi x) = i + 1 - a.$
In this case, we also have $\widehat{\h}(\xi)=1.$
Putting this together we have:
\begin{equation} \label{eq:CDF-1d}
c_{\h}(x) = \frac{1}{M} \left( i + 1 - a + \sum_{\xi \in T \setminus \{0\}} \wh{\h}(\xi) \frac{e(-\xi a) - e(-\xi (x+1))}{1-e(-\xi)} \right) \;.
\end{equation}
Hence, we obtain a closed form expression for the CDF that can be approximated to desired precision in time $O(|T|\log(1/\delta)).$

\paragraph{Proof of (ii):}
For (ii), we can write $M=\mathrm{diag}(M_i)$, $1 \le i \le k$, and
$S= \prod_{i=1}^k ([a_i,b_i] \cap \Z),$
{where  $a_i = m_i-\lceil |M_i|/2 \rceil + 1$
and $b_i=m_i+ \lfloor |M_i|/2 \rfloor$. }
With our lexicographic ordering, we have:
\begin{align*}
c_{\h}(x)
                                               & = \sum_{y \in S: y \leq_{\mathrm{lex}} x} \h(y)\\
					& = \sum_{y_1=a_1}^{x_1-1} \sum_{y_2=a_2}^{b_2} \cdots \sum_{y_k=a_k}^{b_k}  \h(y) \\
					& + \sum_{y_2=a_2}^{x_2-1} \sum_{y_3=a_3}^{b_3} \cdots \sum_{y_k=a_k}^{b_k} \h(x_1,y_2,\ldots, y_k) \\
					& \ldots \\
					& + \sum_{y_k=a_k}^{x_k-1} \h(x_1, \ldots ,x_{k-1}, y_k) \\
					& + \h(x) \;.
\end{align*}
To avoid clutter in the notation, we define $c_{\h,i}(x)$ to be one of these sums, i.e.,
\begin{eqnarray*}
c_{\h,i}(x) & \eqdef  & \sum_{y_i=a_i}^{x_i-1} \sum_{y_{i+1}=a_{i+1}}^{b_{i+1}} \cdots \sum_{y_k=a_k}^{b_k} \h(x_1,\ldots, x_{i-1},y_i,\ldots,y_k) \\
& = & \frac{1}{|\det(M)|} \cdot
\sum_{y_i=a_i}^{x_i-1} \sum_{y_{i+1}=a_{i+1}}^{b_{i+1}} \cdots \sum_{y_k=a_k}^{b_k}
\sum_{\xi \in T} \widehat{\h}(\xi) e\left( - \sum_{j=1}^{i-1} \xi_j x_j -  \sum_{j=i}^{k} \xi_j y_j \right) \\
& = & \frac{1}{|\det(M)|} \cdot
\sum_{\xi \in T} \widehat{\h}(\xi)  e\left(-  \sum_{j=1}^{i-1} \xi_j x_j \right)
\left(\sum_{y_i=a_i}^{x_i-1} e(- \xi_i y_i)\right) \prod_{j=i+1}^k \sum_{y_i=a_i}^{b_i} e(-\xi_j y_j) \\
& = & \frac{1}{|\det(M)|} \cdot
\sum_{\xi \in T} \widehat{\h}(\xi) e\left(- \sum_{j=1}^{i-1} \xi_j x_j\right) s_i(a_i,x_i-1) \prod_{j=i+1}^k s_j(a_j,b_j) \;,
\end{eqnarray*}
where $s_i(a'_i,b'_i) := \sum_{y_i=a'_i}^{b'_i} e(-\xi_i y_i).$
As before, this is a geometric series, so either $\xi_i=0$,
when we have $s_i = b'_i+1-a'_i$, or $s_i=\frac{e(-\xi_i a'_i) - e(-\xi (b'_i+1))}{1-e(-\xi)}$.

We can thus evaluate $c_{\h,i}$ in $O(|T| k)$ arithmetic operations and so compute $c_{\h}(x)$ to desired accuracy in $O(|T| k^2\log(1/\delta))$ time.
We also note that $f: S \rightarrow \{0,1,...,(\prod_i |M_i|)-1 \}$,
defined by $f(x) := \sum_i (x_i-a_i) \prod_{j=1}^i M_j$
is a strictly monotone bijection,
and that $f$ and $f^{-1}$ can be computed in time $O(k).$
Now, $c_{\h}(f^{-1}(y)$ is the CDF of the distribution on $y \in \{0,1,...,(\prod_i |M_i|)-1 \}$
whose PMF is given by $\h(f^{-1}(y)).$

\paragraph{Proof of (iii):}
We will reduce (iii) to (ii).
{
To do this, we use Smith normal form, a canonical factorization of integer matrices:}

\begin{lemma}[See~\cite{storjohannthesis},~\cite{storjohann96})] \label{lem:smith}
Given any integer matrix $M \in \Z^{k \times k}$,
we can factorize $M$ as $M=U \cdot D \cdot V$,
where $U,D,V \in \Z^{k\times k}$
with $D$ diagonal and $U$,$V$ unimodular,
i.e., $|\det (U)|=|\det (V)| = 1,$
and therefore $U^{-1},V^{-1} \in \Z^{k \times k}.$
This factorization can be computed in time $$\poly(k) \log \max_{i,j} |M_{i,j}|.$$
\end{lemma}

Note that the Smith normal form satisfies additional conditions on $D$ 
than those in Lemma \ref{lem:smith}, but we are only interested in finding such a decomposition where $D$ is a diagonal integer matrix.

We note that the integer lattices $M\Z^k$ and $UD\Z^k$ are identical,
since if $x=Mb$ for $b \in \Z^k$, then $x=UDc$ for $c=Vb \in \Z^k$.
For any $\xi \in (M^T)^{-1} \Z^k$, $\xi=(U^T)^{-1} (D^T)^{-1} (V^T)^{-1} b$, for $b \in \Z$.
Then, if we take $\nu = U^T \xi$, we have $\nu \in (DV)^{-1}\Z^k=D^{-1}\Z^k$.

Hence, we can re-write (\ref{eq:inverse-DFT-H}) as follows:
$$\h(x) = \frac{1}{|\det M|} \sum_{\nu \in U^T T} e\left(-\nu U^{-1} x\right) \widehat{\h}((U^T)^{-1} \nu) \;.$$
Since $U^T T \subseteq  D^{-1}\Z^k$, substituting $y=U^{-1} x$ almost gives us the conditions which would allow us to apply (ii).
The issue is that for $x \in m+M(-1/2,1/2]^k$, we have $U^{-1} X \in U^{-1}m + DV(-1/2,1/2]^k,$
but we do not necessarily have $U^{-1} X \in U^{-1}m +  D(-1/2,1/2]^k.$
The following claim gives the details of how to change the fundamental domain:
\begin{claim} \label{clm:change-fund-domain}
Given a non-singular $M \in \Z^{k \times k}$ and $m,x \in \R^k,$
then $x'= x + M R\left(M^{-1} (m-x)\right)$ is the unique $x' \in m+M(-1/2,1/2]^k$ with $x-x' \in M\Z^k,$
where $R(x)$ is $x$ with each coordinate rounded to the nearest integer, rounding half integers up,
i.e., $(R(x))_i := \frac{1}{2} + \lceil x_i - \frac{1}{2} \rceil.$
\end{claim}
So we take $y=g(x) \eqdef U^{-1}x+D R\left((UD)^{-1} m - (UD)^{-1} x\right),$
which using Claim \ref{clm:change-fund-domain} has $g(x) \in U^{-1}m +  D(-1/2,1/2]^k.$
We need the inverse function of $g: m+M(-1/2,1/2]^k \rightarrow U^{-1}m +  D(-1/2,1/2]^k$.
Note that $g^{-1}(y) = U y + D^{-1} R(U^{-1} m - y),$ which again by Claim~\ref{clm:change-fund-domain} is in $m+M(-1/2,1/2]^k.$

So, if $y=g(x),$ since $|\det (M)| = |\det (U)| \cdot |\det (D)| \cdot |\det (V)| = |\det (D)|$, we have:
\begin{equation} \label{eq:straight-DFT}
\h(g^{-1}(y)) =  \frac{1}{|\det (D)|} \sum_{\nu \in U^T T} e(\nu \cdot y) \widehat{\h}((U^T)^{-1} \nu) \;.
\end{equation}
Now, we can take $\h(g^{-1}(y)$ to be a function of $y$ supported on $U^{-1}m +  D(-1/2,1/2]^k$
with a sparse DFT modulo $D$ supported on $U^T T \subseteq D^{-1} \Z^k.$
At this point, we can apply the algorithm of (ii), which gives a way to compute the CDF of $\h(g^{-1}(y))$
with respect to the lexicographic ordering on $y.$ Note that $g$ and $g^{-1}$ can be computed in time
$\poly(k),$ or more precisely the running time of matrix multiplication and inversion.

For the ordering on $x \in S,$ which has $x_1 \le_g x_2$ when $g(x_1) \leq_{\mathrm{lex}} g(x_2),$
we can compute $c_{\h}(x)=\sum_{y \in S: y \leq_g x} \h(y)$ by applying the algorithm in (ii) to the
function given in (\ref{eq:straight-DFT}) applied at $g(x).$
So, we can compute $c_{\h}(x)$ in time $O(k^2|T|+\poly(k)).$
Again, the function given by $f(g(x)),$ where $f$ is as in (ii)
is a monotone bijection from $S$ to $\{0,1, \ldots, \det (M) -1\},$
and we can calculate this function and its inverse in time $\poly(k).$
\end{proof}

Now we can prove the main theorem of this subsection.
\begin{proof}[Proof of Theorem \ref{thm:sampler}]
By Proposition \ref{prop:CDF} (iii), there is a bijection $f$ which takes the support $S$ of $\h$ to the integers
$\{0,1, \ldots ,|S|-1\},$ and we can efficiently calculate the CDF of the distribution considered on this set of integers.
So, we can apply Lemma~\ref{lem:bs} to this CDF on this distribution.
This gives us an $\eps$-sampler for this distribution,
which we can then apply $f^{-1}$ to each sample to get an $\eps$-sampler for $\h.$
To find the time it takes to compute each sample, we need to substitute
$D=O(\poly(k)+k^2|T|\log(|\det(M)|/\epsilon))$ from the running time of the CDF in Proposition~\ref{prop:CDF} (iii) into the bound
in Lemma \ref{lem:bs}, yielding $$O(\log(|\det (M)|) \log(|\det(M)|/\epsilon) \cdot |T| \cdot \poly(k))$$ time. This completes the proof.
\end{proof}

\subsection{Using our Learning Algorithm to Obtain a Cover} \label{sec:cover-from-alg}

As an application of our learning algorithm in Section~\ref{ssec:algo-dft}, 
we provide a simple proof that 
the space $\mathcal{M}_{n,k}$ of all $(n,k)$-PMDs has an $\eps$-cover under the total variation distance of size 
$n^{O(k^2)} \cdot 2^{O\left(k\log(1/\eps)\right)^{O(k)}}.$ Our argument is constructive, 
yielding an efficient algorithm to construct a non-proper $\eps$-cover of this size.

Two remarks are in order: (i) the non-proper cover construction in this subsection does not suffice
for our algorithmic applications of Section~\ref{sec:cover-nash}. These applications require the efficient construction
of a {\em proper} $\eps$-cover plus additional algorithmic ingredients.
(ii) The upper bound on the cover size obtained here is nearly optimal, 
as follows from our lower bound  in Section~\ref{sec:cover-lb}.

The idea behind using our algorithm to obtain a cover is 
quite simple. In order to determine its hypothesis, $\h$, our algorithm  {\tt Efficient-Learn-PMD} 
requires the following quantities:
\begin{itemize}
\item A vector $\wh{\mu} \in \R^k$ and a PSD matrix $\wh{\Sigma} \in \R^{k \times k}$ 
satisfying the conclusions of Lemma \ref{lem:sample-mean-and-covariance}.
\item Values of the DFT $\wh{\h}(\xi)$, for all $\xi\in T,$ so that
$\sum_{\xi\in T}|\wh{\h}(\xi)-\wh{\p}(\xi)| < \eps/10.$
Recall that $T  \eqdef  \left\{ \xi = (M^T)^{-1} v ~ \mid~   \left(v\in \Z^k\right)\land  \left(\|v\|_2 \leq C^2 k^2 \ln(k/\eps)\right) \right\}.$
\end{itemize}
Given this information, the analysis in Section \ref{ssec:algo-dft} carries over immediately. 
The algorithm {\tt Efficient-Learn-PMD} works by estimating the mean and covariance using samples, and then 
taking $\wh{\h}(\xi)$ to be the sample Fourier transform. If we instead guess the values of these 
quantities using an appropriate discretization, we obtain an $\eps$-cover for $\mathcal{M}_{n,k}.$
More specifically, we discretize the above quantities as follows:
\begin{itemize}
\item To discretize the mean vector, we consider a $1$-cover $\mathcal{Y}_1$ 
of the set $\mathcal{Y} \eqdef \{y = (y_1, \ldots, y_k) \in \R^k: (y_i \geq 0)\land (\sum_{i=1}^k y_i = k) \}$ with respect to the $L_2$ norm.
It is easy to construct such a cover with size $|\mathcal{Y}_1| \leq O(n)^k.$

\item To discretize the covariance matrix, we consider a $1/2$-cover $\mathcal{S}_{1/2}$ of the set of matrices 
$\mathcal{S} \eqdef \{ A \in \R^{k \times k}: (A \succeq \mathbf{0}) \land \|A\|_2 \leq n\},$
with respect to the spectral norm $\| \cdot \|_2.$ 
Note that we can construct such a $1/2$-cover with size 
$|\mathcal{S}_{1/2}| \leq (4n+1)^{k(k+1)/2}.$ This is because any maximal $1/4$-packing of this space 
(i.e., a maximal set of matrices in $\mathcal{S}$ with pairwise distance under the spectral norm at least $1/4$) 
is such a $1/2$-cover. Observe that for any such maximal packing, the balls of radius $1/4$ centered at these points  
are disjoint and contained in the ball (under the spectral norm) about the origin of radius $(n+1/4).$
Since the ball of radius $(n+1/4)$ has volume $(4n+1)^{k(k+1)/2}$ times as much as the ball of radius $1/4$, 
a simple volume argument completes the proof. 

\item Finally, to discretize the Fourier transform, we consider a $\delta$-cover $\mathcal{C}_{\delta}$ 
of the unit disc on the complex plane $\C$, 
with respect to the standard distance on $\C,$
where $$\delta \eqdef \eps (2C^2k^2\log(k/\eps))^{-k}/10 = \eps/(10 t) \leq \eps/(10|T|).$$ We note that 
$t \eqdef  (2C^2k^2\log(k/\eps))^{k}$ is an upper bound on $|T|.$
\end{itemize}

We claim that there is an $\eps$-cover of the space of $(n,k)$-PMDs
indexed by $\mathcal{Y} _1 \times \mathcal{S}_{1/2} \times \mathcal{C}_{\delta}^{t}.$ 
Such a cover is clearly of the desired size. The cover is constructed as follows: 
We let $\wh{\mu}$ and $\wh{\Sigma}$ be the selected elements 
from $\mathcal{Y}_1$ and $\mathcal{S}_{1/2},$ respectively. 
We use these elements to define the matrix $M \in \Z^{k \times k}$ as in the algorithm description. 
We then use our elements of $\mathcal{C}_{\delta}$ as the values of $\wh{\h}(\xi)$ for $\xi\in T$ (noting that $|T|\leq  t$).

We claim that for any $(n,k)$-PMD $\p$ there exists a choice of parameters, 
so that the returned distribution $\h$ is within total variation distance 
$\eps$ of $\p.$ We show this as follows: Let $\mu$ and $\Sigma$ be the true mean and covariance matrix of $\p.$ 
We have that $\mu \in \mathcal{Y}$ and that $\Sigma \in \mathcal{S}.$ 
Therefore, there exist $\wh{\mu} \in\mathcal{Y}_1$ and $\wh{\Sigma}\in \mathcal{S}_{1/2}$ 
so that $|\mu-\wh{\mu}|_2 \leq 1$ and $I/2 \succeq \Sigma - \wh{\Sigma} \succeq -I/2.$ 
It is easy to see that these conditions imply the conclusions of Lemma~\ref{lem:sample-mean-and-covariance}.
Additionally, we can pick elements of $\mathcal{C}_{\delta}$ in order to make $|\wh{\h}(\xi)-\wh{\p}(\xi)| < \eps/(10|T|)$ for each $\xi\in T.$ 
This will give that $\sum_{\xi\in T}|\wh{\h}(\xi)-\wh{\p}(\xi)| < \eps/10.$ 
In particular, the hypothesis $\h$ indexed by this collection of parameters will be within variation distance $\eps$ of $\p.$
Hence, the set we have constructed is an $\eps$-cover, and our proof is complete.

\section{Efficient Proper Covers and Nash Equilibria in Anonymous Games} \label{sec:cover-nash}

In this section, we give our efficient proper cover construction for PMDs,
and our EPTAS for computing Nash equilibria in anonymous games. 
These algorithmic results are based on new structural results for PMDs that we establish.
The structure of this section is as follows:
In Section~\ref{ssec:moments-struct}, we show the desired sparsity property of the continuous
Fourier transform of PMDs, and use it to prove our robust moment-matching lemma.
Our dynamic-programming algorithm for efficiently constructing 
a proper cover relies on this lemma, and is given in Section~\ref{ssec:cover-dp}.
By building on the proper cover construction, in Section~\ref{ssec:anonymous} we give our EPTAS for
Nash equilibria in anonymous games. In Section~\ref{sec:distinct}, we combine our moment-matching lemma
with recent results in algebraic geometry, to show that any PMD is close to another PMD
with few distinct CRV components. Finally, in Section~\ref{sec:cover-lb} we prove out cover size lower bound.

\subsection{Low-Degree {Parameter} Moment Closeness Implies Closeness in Variation Distance}
\label{ssec:moments-struct}

In this subsection, we establish the sparsity of the continuous Fourier transform of PMDs,
and use it to prove our robust moment-matching lemma, 
translating closeness in the low-degree {parameter} moments to closeness in total variation distance. 

At a high-level, our robust moment-matching lemma (Lemma~\ref{lem:moments-imply-dtv})
is proved by combining the sparsity of the continuous Fourier transform of PMDs (Lemma~\ref{lem:ft-es})
with very careful Taylor approximations of the logarithm of the Fourier transform (log FT) of our PMDs.
For technical reasons related to the convergence of the log FT, we will need one additional property
from our PMDs. In particular, we require that each component $k$-CRV has the same most likely outcome.
This assumption is essentially without loss of generality. There exist at most $k$ such outcomes,
and we can express an arbitrary PMD as a sum of $k$ independent component PMDs whose $k$-CRV components satisfy this property.
Formally, we have the following definition:
\begin{definition}
We say that a $k$-CRV $W$ is $j$-maximal, for some $j \in [k]$,
if for all $\ell \in [k]$ we have $\Pr[W = e_j] \ge \Pr[W = e_{\ell}].$
We say that an $(n, k)$-PMD $X = \sum_{i=1}^n X_i$
is $j$-maximal, for some $j \in [k]$, if for all $1 \le i \le n$
$X_i$ is a $j$-maximal $k$-CRV.
\end{definition}


Any $(n, k)$-PMD $X$ can be written as $X = \sum_{i=1}^k X^i$,
where $X^i$ is an $i$-maximal $(n_i, k)$-PMD,
with $\sum_i n_i = n.$ 
For the rest of this intuitive explanation, we focus on two $(n, k)$-PMDs $X, Y$ that are promised to be 
$i$-maximal, for some $i \in [k].$ 

To guarantee that $\wh{X}$, $\wh{Y}$ have roughly the same effective support, 
we also assume that they have roughly the same variance in each direction.
We will show that if the low-degree {parameter} moments of $X$ and $Y$ are close to each other, then $X$ and $Y$ are close in total variation distance. 
We proceed by partitioning the $k$-CRV components of our PMDs into groups, based on their maximum probability element $e_j$, with $j \neq i.$
The maximum probability of a $k$-CRV quantifies its maximum contribution to the variance of the PMD in some direction.
Roughly speaking, the smaller this contribution is, the fewer terms in the Taylor approximation are needed to achieve a given error. 
More specifically, we consider three different groups, partitioning the component $k$-CRVs into ones with small, medium, and large 
contribution to the variance in some direction. For the PMD (defined by the CRVs) of the first group, we only need to approximate the first $2$ {parameter} moments.
For the PMD of the second group, we approximate the low-degree {parameter} moments up to degree $O_k (\log(1/\eps) / \log \log(1/\eps)).$
Finally, the third group is guaranteed to have very few component $k$-CRVS, hence we can afford to approximate the individual parameters.

To quantify the above, we need some more notation and definitions.
To avoid clutter in the notation, we focus without loss of generality on the case $i=k$, i.e., 
our PMDs are $k$-maximal.
For a $k$-maximal $(n, k)$-PMD, $X$, let $X = \sum_{i=1}^n X_i$,
where the $X_i$ is a $k$-CRV with $p_{i, j} = \Pr[X_i = e_j]$ for $1\le i \le n$ and $1 \le j \le k.$
Observe that $\sum_{j=1}^k p_{i,j}=1,$ for $1 \le i \le n,$ hence the
definition of $k$-maximality implies that $p_{i,k}\geq 1/k$ for all $i.$
Note that the $j^{th}$ component of the random vector $X$ is a PBD with parameters $p_{i, j}$, $1 \le i \le n.$
Let $s_j(X)=\sum_{i=1}^n p_{i,j}$ be the expected value of the $j^{th}$ component of $X$.
We can assume that $s_j(X) \ge \eps/k$, for all $1 \le j \le k-1$; otherwise, we can remove
the corresponding coordinates and introduce an error of at most $\eps$ in variation distance.

Note that, for $j \neq k$, the variance of the $j^{th}$ {coordinate} of $X$ is in $[s_j(X)/2, s_j(X)]$.
Indeed, the aforementioned variance equals $\sum_{i=1}^n p_{i, j} (1- p_{i, j})$, which is clearly at most
$s_j(X)$. The other direction follows by observing that, for all $j \ne k$, we have
$1 \ge p_{i, k} + p_{i, j} \ge 2 p_{i, j}$, or $p_{i, j} \le 1/2$,
where we again used the $k$-maximality of $X$.
Therefore,
by Bernstein's inequality and a union bound, there is a set $S \subseteq [n]^k$
of size
$$|S| \le \prod_{j=1}^{k-1} \left(1+ 12 s_j(X)^{1/2} \ln(2k/\eps)\right)
\le O\left(\log(k/\eps)^{(k-1)} \right) \cdot  \prod_{j=1}^{k-1} (1+ 12 s_j(X)^{1/2}) \;, $$
so that $X$ lies in $S$ with probability at least $1-\eps$.


We start by showing that the continuous Fourier transform of a PMD is approximately sparse,
namely it is effectively supported on a small set $T.$ More precisely, we prove
that there exists a set $T$ in the Fourier domain such that  the 
integral of the absolute value of the Fourier transform outside $T$ 
multiplied by  the size of the effective support $|S|$ of our PMD is small.

\begin{lemma}[Sparsity of the FT for PMDs]  \label{lem:ft-es}
Let $X$ be $k$-maximal $(k,n)$-PMD with effective support $S.$ 
Let $$T \eqdef \left\{\xi \in [0, 1]^k : [\xi_j-\xi_k]<Ck (1+12s_j(X))^{-1/2}\log^{1/2}(1/\eps) \right\},$$
where $[x]$ is the distance between $x$ and the nearest integer,
and $C >0$ is a sufficiently large universal constant.
Then, we have that $$\int_{\overline{T}} |\wh{X}| \ll \eps/ |S|.$$
\end{lemma}
\begin{proof}
To prove the lemma, we will need the following technical claim: 
\begin{claim} \label{clm:ub-ft-crv}
For all $\xi = (\xi_1, \ldots, \xi_k)  \in [0, 1]^k$,  for all $1 \leq i \leq n$ and $1 \leq j \leq k-1$,
it holds: $$|\wh X_i(\xi)| = \exp(-\Omega(p_{i,j}[\xi_j-\xi_k]^2/k)) \;.$$
\end{claim}
\begin{proof}
The claim follows from the following sequence of (in-)equalities:
\begin{align*}
 |\wh{X_i}(\xi)|^2
 =&  \Big(\sum_{j=1}^k p_{i,j}e(\xi_j) \Big) \Big(\sum_{j'=1}^k p_{i,j'}e(-\xi_{j'}) \Big)
 =  \sum_{j,j'} p_{i,j}p_{i,j'} e(\xi_j-\xi_{j'}) \\
 =&  \sum_{j,j'} p_{i,j}p_{i,j'} \cos(2\pi(\xi_j-\xi_{j'}))
 =  1-\sum_{j\neq j'} p_{i,j}p_{i,j'}\big(1- \cos(2\pi(\xi_j-\xi_{j'}))\big) \\
 = & {1-\sum_{j\neq j'} p_{i,j}p_{i,j'}\big(1- \Omega([\xi_j-\xi_{j'}])^2\big) \tag*{(by Claim \ref{claim:cos-claim} with $\delta=\frac{1}{2}$)}} \\
 =&  \exp\left(-\Omega\left( \sum_{j \neq j'} p_{i,j'}p_{i,j} [\xi_j-\xi_{j'}]^2 \right)\right) \\
 =&  \exp\left(-\Omega\left( \sum_{j < k} p_{i,j}p_{i,k} [\xi_j-\xi_k]^2 \right)\right) \\
 =&  \exp\left(-\Omega\left( p_{i,j} [\xi_j-\xi_k]^2/k \right)\right) \;,
\end{align*}
where the last lines uses the fact $p_{i, k} \ge 1/k$, which follows from $k$-maximality.
\end{proof}
As a consequence of Claim~\ref{clm:ub-ft-crv}, we have that
\begin{equation} \label{eqn:ub-ft-pmd}
|\wh X(\xi)| = \prod_{i=1}^n |\wh{X_i}(\xi)| =  \exp(-\Omega(s_j(X)[\xi_j-\xi_k]^2/k)) \;.
\end{equation}
Let $\overline{T} = [0, 1]^k \setminus T$ be the complement of $T$.
To bound $\int_{\overline{T}} |\wh{X}|,$ we proceed as follows:
For $\ell \in \Z_+$, we define the sets
$$\overline{T}_{\ell} = \left\{\xi:\max([\xi_j-\xi_k]/(Ck (1+12s_j(X))^{-1/2}\log^{1/2}(1/\eps))) \in [2^{\ell},2^{\ell+1}]\right\} \;,$$ and observe
that $\overline{T} \subseteq \cup_{\ell \in \Z_+} \overline{T}_{\ell}$.
Now, Equation (\ref{eqn:ub-ft-pmd}) implies that
for $\xi\in \overline{T}_{\ell}$ it holds $|\wh{X}(\xi)|\leq \eps^{10k 2^{\ell}}$, where we used the assumption that the constant
$C$ is sufficiently large.
It is easy to see
that the volume of $\overline{T}_{\ell}$ is at most
$O(2^{\ell} k \log^{1/2}(1/\eps))^k\prod_{j<k} (1+12s_j(X))^{-1/2}$.
We can bound $ \int_{\overline{T}} |\wh{X}| $ from above by
the sum over $\ell$ of the maximum value of  $|\wh{X}|$ within  $\overline{T}_{\ell}$ times the volume of  $\overline{T}_{\ell}$, namely
$$\int_{\overline{T}} |\wh{X}| \le \sum_{\ell= 0}^{\infty} \int_{\overline{T}_{\ell}} |\wh{X}| \le
\sum_{\ell= 0}^{\infty} \left( \sup_{\xi \in \overline{T}_{\ell}} |\wh{X}(\xi)| \right) \mathrm{Vol}(\overline{T}_{\ell})
\le \eps^k  \prod_{j<k} (1+12s_j(X))^{-1/2} \ll \eps/|S|.$$
This completes the proof of Lemma~\ref{lem:ft-es}.
\end{proof}

We now use the sparsity of the Fourier transform to show that if two $k$-maximal PMDs,
with similar variances in each direction, have Fourier transforms
 that are pointwise sufficiently close to each other in this effective support,
 then they are close to each other in total variation distance.

\begin{lemma} \label{lem:ft-es-dtv}
Let $X$ and $Y$ be $k$-maximal $(k,n)$-PMDs, satisfying 
{$1/2 \le (1+s_j(X)) / (1+s_j(Y)) \le 2$}
for all $j$, $1 \le j \le k-1.$
Let $T \eqdef \left\{\xi \in [0, 1]^k : [\xi_j-\xi_k]<Ck (1+12s_j(X))^{-1/2}\log^{1/2}(1/\eps) \right\},$
where $[x]$ is the distance between $x$ and the nearest integer,
and $C >0$ is a sufficiently large universal constant.
Suppose that for all $\xi\in T$ it holds
$|\wh X(\xi)-\wh Y(\xi)|\leq \eps (Ck \log(k/\eps))^{-2k}$.
Then, $\dtv(X,Y) \le \eps.$
\end{lemma}
\begin{proof}
We start with an intuitive explanation of the proof.
Since  $1+s_j(X)$, $1+s_j(Y)$ are within a factor of $2$ for $1 \le j \le k-1$,
it follows from the above that $X$ and $Y$ are both effectively supported on a set $S \subseteq [n]^k$ of size
$|S| \le O\left(\log(k/\eps)\right)^{k-1} \cdot \prod_{j=1}^{k-1} \left(1+ 12 s_j(X)^{1/2}\right)$.
Therefore, to prove the lemma,
it is sufficient to establish that $\|X-Y\|_\infty \leq O(\eps/|S|$).

We prove this statement in two steps by analyzing the continuous Fourier transforms $\wh X$ and $\wh Y$.
The first step of the proof exploits the fact
that the Fourier transforms of $X$ and $Y$
are each essentially supported on the set $T$ of the lemma statement.
Recalling the assumption that $(1+s_j(X))/(1+s_j(Y)) \in [1/2, 2]$, $1 \le j \le k-1$,
an application of Lemma~\ref{lem:ft-es} yields that
$\int_{\overline{T}} |\wh{X}|$ and $\int_{\overline{T}} |\wh{Y}|$ are both at most $\eps/|S|.$
Thus, we have that
$$\int_{\overline{T}} |\wh{X}-\wh{Y}| \le \int_{\overline{T}} |\wh{X}|+ \int_{\overline{T}} |\wh{Y}| \ll \eps/|S| \;.$$

In the second step of the proof, we use the assumption that
the absolute difference $|\wh X(\xi)-\wh Y(\xi)|$, $\xi \in T$, is small,
and the fact that {$\int_{\overline{T}} |\wh{X}|$ and
 $\int_{\overline{T}} |\wh{Y}|$} are individually small,
to show that $\|\wh{X}-\wh{Y}\|_1 \le O(\eps/|S|)$.
The straightforward inequality  $\|X-Y\|_\infty \le \|\wh{X}-\wh{Y}\|_1$ combined with the concentration of
$X, Y$ completes the proof.

Given the aforementioned, in order to bound $\|\wh{X}-\wh{Y}\|_1$,
it suffices to show that the integral over $T$ is small.
By the assumption of the lemma, we have that $|\wh{X}(\xi)-\wh{Y}(\xi)|$ is point-wise
at most $\eps (Ck \log(k/\eps))^{-2k}$ over $T$.
We obtain an upper bound on $\int_{T} |\wh{X}-\wh{Y}|$
by multiplying this quantity by the volume of $T$. Note
that the volume of $T$ is at most $O(k\log^{1/2}(1/\eps))^{k-1} \prod_{j < k} (1+12s_j(X))^{-1/2})$.
Hence,
$$\int_{T} |\wh{X}-\wh{Y}| \le \eps \cdot \Theta(\log(k/\eps))^{-k} \cdot \prod_{j<k} (1+12s_j(X))^{-1/2} \ll \eps / |S|.$$
Combining the above, we get that $\|X-Y\|_\infty \le \|\wh{X}-\wh{Y}\|_1 =  O(\eps / |S|),$
which implies that the $L_1$ distance between $X$ and $Y$ over $S$ is $O(\eps)$.
The contribution of  $\overline{S}$ to the $L_1$ distance is at most $\eps$,
since both $X$ and $Y$ are in $S$ with probability at least $1-\epsilon$.
This completes the proof of Lemma~\ref{lem:ft-es-dtv}.
\end{proof}

We use this lemma as technical tool for our robust moment-matching lemma. 
As mentioned in the beginning of the section, 
we will need to handle separately the component $k$-CRVs that have a significant
contribution to the variance in some direction. 
This is formalized in the following definition:

\begin{definition}
Let $X$ be a $k$-maximal $(n, k)$-PMD with $X=\sum_{i=1}^n X_i$ and $0 < \delta \le 1.$
For a given $\ell \in [n]$,
we say that a particular component $k$-CRV $X_\ell,$ with $p_{\ell, j} = \Pr[X_{\ell} = e_j],$
is $\delta$-exceptional if there exists a coordinate $j$,  with $1 \le j \le k-1,$
such that $p_{\ell,j} \geq \delta  \cdot \sqrt{1+s_j(X)}.$ We will denote by $E(\delta, X) \subseteq [n]$ the set of
$\delta$-exceptional components of $X$.
\end{definition}

Recall that the variance of the $j^{th}$ coordinate of $X$ is in $[s_j(X)/2, s_j(X)]$.
Therefore, the above definition states that the $j^{th}$ coordinate of $X_i$ has probability mass
which is at least a $\delta$-fraction of the standard deviation across the $j^{th}$ coordinate of $X$.

We remark that for any $(n, k)$-PMD $X$, at most $k/\delta^2$ of its component $k$-CRVs
are $\delta$-exceptional. To see this, we observe that the number of $\delta$-exceptional components
is at most $k-1$ times the number of $(\delta, j)$-exceptional components, i.e., the $k$-CRVs
$X_i$ which are $\delta$-exceptional for the same value of $j$.
We claim that for any $j$, $1\le j \le k-1$, the number of  $(\delta, j)$-exceptional components is at most
$1/\delta^2$. Indeed, let $E_j \subseteq [n]$ denote the corresponding set. Then, we have that
$\sum_{i \in E_j} p_{i, j}^2 \ge \delta^2 |E_j| s_j(X) = \delta^2 |E_j| \sum_{i=1}^n p_{i, j}.$
Noting that $\sum_{i \in E_j} p_{i, j}^2 \le  \sum_{i=1}^n p_{i, j}^2 \le  \sum_{i=1}^n p_{i, j}$, we get
that $ \delta^2 |E_j| \le 1$, thus yielding the claim.

We now have all the necessary ingredients for our robust moment-matching lemma.
Roughly speaking, we partition the coordinate $k$-CRVs of our $k$-maximal PMDs into three groups.
For appropriate values $0< \delta_1 < \delta_2,$ we have: (i) $k$-CRVs that are {\em not} $\delta_1$-exceptional,
(ii) $k$-CRVs that are $\delta_1$-exceptional, but {\em not} $\delta_2$-exceptional, and (iii) 
 $\delta_2$-exceptional $k$-CRVs. For group (i), we will only need to approximate the first two {parameter} moments 
 in order to get a good Taylor approximation, and for group (ii) we need to approximate
 as many as $O_k(\log(1/\eps)/\log\log(1/\eps))$ degree {parameter} moments. 
Group (iii) has $O_k(\log^{3/2}(1/\eps))$ coordinate $k$-CRVs, hence we simply approximate 
the individual (relatively few) parameters each to high precision. Formally, we have:



\begin{lemma} \label{lem:moments-imply-dtv}
Let $X$ and $Y$ be $k$-maximal $(n,k)$-PMDs, satisfying
$1/2 \le (1+s_j(X)) / (1+s_j(Y)) \le 2$ for all $j$, $1\le j \le k-1$.
Let $C$ be a sufficiently large constant.
Suppose that the component $k$-CRVs of $X$ and $Y$
can be partitioned into three groups,
so that $X=X^{(1)}+X^{(2)}+X^{(3)}$
and $Y=Y^{(1)}+Y^{(2)}+Y^{(3)}$,
where $X^{(t)}$ and $Y^{(t)}$, $1 \le t \le 3$,
are PMDs over the same number of $k$-CRVs.
Additionally assume the following: (i) for $t \leq 2$ the random variables $X^{(t)}$ and $Y^{(t)}$
have no $\delta_t$-exceptional components, where
$\delta_1 = \delta_1(\eps) \eqdef  \eps (Ck \log(k/\eps))^{-3k-3}$ and $\delta_2 = \delta_2 (\eps) \eqdef  k^{-1} \log^{-3/4}(1/\eps)$,
and (ii) there is a bijection between the component $k$-CRVs of $X^{(3)}$
with those in $Y^{(3)}$, so that corresponding $k$-CRVs have total variation distance at most {$\eps/3n_3$},
where $n_3$ is the number of such $k$-CRVs.

Finally, suppose that for $t \leq 2$, and all vectors $m \in \Z^{k}_+$ with {$m_k = 0$ and}  $|m|_1\leq K_t$ it holds
$$
|M_m(X^{(t)})-M_m(Y^{(t)})|(2k)^{|m|_1} \leq  \gamma = \gamma(\eps) \eqdef \eps (Ck \log(k/\eps))^{-2k-1} \;,
$$
where $K_1=2$ and $K_2 = K_2(\eps) = C(\log(1/\eps)/\log\log(1/\eps)+k)$.
Then $\dtv(X,Y) \le \eps.$
\end{lemma}
\begin{proof}
First, note that $\dtv(X, Y) \le \sum_{t=1}^3 \dtv(X^{(t)}, Y^{(t)})$, so it suffices to show that
$\dtv(X^{(t)}, Y^{(t)}) = {\eps/3}$, for $t= 1, 2, 3$. This holds trivially for $t=3$, by assumption.
To prove the statement for $t=1, 2$, by Lemma~\ref{lem:ft-es-dtv},
it is sufficient to show that $\wh{X^{(t)}}$ and $\wh{Y^{(t)}}$ are point-wise close on the set $T$,
namely that for all $\xi \in T$ it holds
$|\wh{X^{(t)}}(\xi)-\wh{Y^{(t)}}(\xi)|\leq \eps (Ck \log(k/\eps))^{-2k}$.
To show this, we show separately that $\wh{X^{(t)}}$ is close to $\wh{Y^{(t)}}$ for each $t =1, 2$.

Let $X^{(t)} = \sum_{i \in A_t} X_i$, where $A_t \subseteq [n]$ with $|A_t| = n_t.$
We have the following formula for the Fourier transform of $X^{(t)}$:
\begin{align}
\notag
\wh{X^{(t)}}(\xi) &= \prod_{i \in A_t} \sum_{j=1}^k e(\xi_j)p_{i,j}\\ \notag
&= e(n_t\xi_k)\prod_{i \in A_t} \left(1-\sum_{j=1}^{k-1} (1-e(\xi_j-\xi_k))p_{i,j} \right)\\ \notag
&= e(n_t\xi_k) \exp\left(\sum_{i \in A_t} \log\left(1-\sum_{j=1}^{k-1} (1-e(\xi_j-\xi_k))p_{i,j} \right) \right)\\ \notag
&= e(n_t\xi_k) \exp\left(-\sum_{i \in A_t} \sum_{\ell=1}^{\infty} \frac{1}{\ell}\left(\sum_{j=1}^{k-1} (1-e(\xi_j-\xi_k))p_{i,j} \right)^{\ell} \right)\\
&= e(n_t\xi_k) \exp\left(-\sum_{m \in \Z^{k-1}_+} \binom{|m|_1}{m} \frac{1}{|m|_1} M_m(X^{(t)}) \prod_{j=1}^{k-1}(1-e(\xi_j-\xi_k))^{m_j} \right). \label{FTEqn}
\end{align}
An analogous formula holds for $\wh{Y^{(t)}}$. To prove the lemma, we will show that, {for all $\xi \in T$,}
the two corresponding expressions
inside the exponential of (\ref{FTEqn}) agree for $X^{(t)}$ and $Y^{(t)}$, up to a sufficiently small error.

We first deal with the terms with $|m|_1\leq K_t$, $t=1, 2$.
By the statement of the lemma, for any two such terms we have that
$|M_m(X^{(t)})-M_m(Y^{(t)})| \leq (2k)^{-|m|_1}  \cdot  \eps (Ck \log(k/\eps))^{-2k}$.
Hence, {for any $\xi \in [0, 1]^k$}, the contribution of these terms to the difference is at most
$$
 \eps (Ck \log(k/\eps))^{-2k-1} \sum_{m \in \Z^{k-1}_+, \textrm{ } |m|_1 \le K_t} \binom{|m|_1}{m} (2k)^{-|m|_1} 2^{|m|_1}
\leq K_t  \eps (Ck \log(k/\eps))^{-2k-1} \leq \eps (Ck \log(k/\eps))^{-2k}   \;.
$$
To deal with the remaining terms, we need the following technical claim:
\begin{claim} \label{clm:tech}
Let $X^{(t)}$ be as above.
For $t\leq 2$ and $m \in \Z^{k-1}_+$ with $|m|_1\geq 2$, we have that
$$
|M_m(X^{(t)})|\prod_{j=1}^{k-1} \left({(1+s_j(X))}^{-1/2}\log^{1/2}(1/\eps)\right)^{m_j} \leq \log^{|m|_1/2}(1/\eps) \cdot \delta_t^{|m|_1-2}.
$$
\end{claim}
\begin{proof}
By definition we have that $M_m(X^{(t)}) = \sum_{i \in A_t} \prod_{j=1}^{k-1} p_{i, j}^{m_j}.$
Thus, the claim is equivalent to showing that
$$
\sum_{i \in A_t} \prod_{j=1}^{k-1} \left( p_{i,j} \cdot {(1+s_j(X))}^{-1/2}(X) \right)^{m_j}  \leq  \delta_t^{|m|_1-2}.
$$
Since, by definition, $X^{(t)}$ does not contain any $\delta_t$-exceptional $k$-CRV components,
we have that for all $i \in A_t$ and all $j \in [k-1]$ it holds $p_{i,j} \cdot {(1+s_j(X))}^{-1/2}(X) \le \delta_t.$
Now observe that decreasing any component of $m$ by $1$ decreases
the left hand side of the above by a factor of at least $\delta_t$.
Therefore, it suffices to prove the desired inequality for $|m|_1=2$, i.e., to show that
$$
\sum_{i \in A_t} \left( p_{i,j_1} {(1+s_{j_1}(X))^{-1/2}} \right) \left( p_{i,j_2} {(1+s_{j_2}(X))^{-1/2}}\right) \leq 1.
$$
Indeed, the above inequality holds true, as follows from an application of the Cauchy-Schwartz inequality,
and the fact that
$$
\sum_{i \in A_t} p_{i,j}^2 \leq \sum_{i=1}^n p_{i,j} = s_j(X) {\leq s_j(X) + 1}.
$$
This completes the proof of Claim~\ref{clm:tech}.
\end{proof}
Now, for $\xi \in T$, the contribution to the exponent of (\ref{FTEqn}), coming from terms with $|m|_1 > K_t$, is at most
\begin{equation} \label{eqn:high-moments}
\sum_{\ell>K_t} \sum_{m \in \Z^{k-1}_+ : |m|_1 = \ell} \left( \binom{\ell}{m} \frac{1}{\ell} M_m(X^{(t)}) \prod_{j=1}^{k-1}O([\xi_j-\xi_k])^{m_j} \right) \le
\sum_{\ell>K_t} k^{\ell} \cdot \log^{\ell/2}(1/\eps) \cdot \delta_t^{\ell-2}  \;.
\end{equation}
Equation (\ref{eqn:high-moments}) requires a few facts to be justified.
First, we use the multinomial identity $\sum_{m \in \Z^{k-1}_+ : |m|_1 = \ell} \binom{\ell}{m} = (k-1)^{\ell}.$
We also require the fact that $$|1 - e(\xi_j-\xi_k)| \le O([\xi_j - \xi_k]) \;,$$ $\xi \in [0, 1]^T,$ and recall that
$[\xi_j-\xi_k]<Ck (1+12s_j(X))^{-1/2}\log^{1/2}(1/\eps)$, for $\xi \in T.$ Combining the above with Claim~\ref{clm:tech}
gives (\ref{eqn:high-moments}).

Finally, we claim that
$$
\sum_{\ell>K_t} k^{\ell} \cdot \log^{\ell/2}(1/\eps) \cdot \delta_t^{\ell-2} \leq   \eps (Ck \log(k/\eps))^{-2k} \;,
$$
where the last inequality holds for both $t=1, 2$, as can be readily verified
from the definition of $\delta_1, \delta_2, K_1, K_2.$
Combining with the bounds for smaller $|m|_1$,
we get that the absolute difference between $\wh{X^{(t)}}$ and $\wh{Y^{(t)}}$ on $T$
is at most $\eps (Ck \log(k/\eps))^{-2k}.$
Therefore, Lemma \ref{lem:ft-es} implies that {$\dtv(X^{(t)},Y^{(t)}) \leq \eps.$ 
A suitable definition of $C$ in the statement of the lemma to make this $\eps/3$
completes the proof of Lemma~\ref{lem:moments-imply-dtv}.}
\end{proof}


\begin{remark}
\emph{We note that the quantitive statement of Lemma~\ref{lem:moments-imply-dtv} is crucial for our algorithm:
(i) The set of non $\delta_1$-exceptional components can contain up to $n$ $k$-CRVs. 
Since we only need to approximate only the first $2$ {parameter} moments for this set, this only involves $\poly(n)$ possibilities. 
(ii) The set of $\delta_1$-exceptional but not $\delta_2$-exceptional $k$-CRVs has size $O(k/\delta_1^2)$, which is independent of $n$.
In this case, we approximate the first $O_k (\log (1/\eps)/\log\log(1/\eps))$ {parameter} moments, and the total number of possibilities is independent of $n$ 
and bounded by an appropriate quasipolynomial function of $1/\eps$. (ii) The set of $\delta_2$-exceptional components is sufficiently small, so that
we can afford to do a brute-force grid over the parameters.}
\end{remark}

\subsection{Efficient Construction of a Proper Cover} \label{ssec:cover-dp}

As a warm-up for our proper cover algorithm, we use the structural 
lemma of the previous section to show the following upper bound
on the cover size of PMDs.

\begin{proposition} \label{prop:cover-upper-maximal}
For all $n, k \in \Z_+,$ $k>2,$ and $\eps>0$, there exists an $\eps$-cover
of the set of $(n, k)$-PMDs of size
$n^{O(k^3)} (1/\eps)^{O(k\log(1/\eps)/\log\log(1/\eps))^{k-1}}.$
\end{proposition}

\begin{remark}
\emph{We remark that, for the sake of simplicity, we have not optimized the dependence of 
our cover upper bound on the parameter $n.$
With a slightly more careful argument, one can easily obtain
a cover size upper bound $n^{O(k^2)}(1/\eps)^{O(k\log(1/\eps)/\log\log(1/\eps))^{k-1}}.$ 
On the other hand, the asymptotic dependence of our upper bound 
on the error parameter $\eps$ is optimal. In Section~\ref{sec:cover-lb}, we show a lower bound of 
$(1/\eps)^{\Omega_k(\log(1/\eps)/\log\log(1/\eps))^{k-1}}.$}
\end{remark}

\begin{proof}[Proof of Proposition~\ref{prop:cover-upper-maximal}.]
Let $X$ be an arbitrary $(n, k)$-PMD. We can write $X$ as $\sum_{i=1}^k X^{i},$
where $X^i$ is an $i$-maximal $(n^{(i)}, k)$-PMD, where $\sum_{i=1}^k n^{(i)} = n.$
By the subadditivity of the total variation distance for independent random variables, 
it suffices to  show that the set of $i$-maximal $(n, k)$-PMDs has an $\eps/k$-cover of size
$n^{O(k^2)} (1/\eps)^{O(k\log(k/\eps)/\log\log(k/\eps))^{k-1}}.$

To establish the aforementioned upper bound 
on the cover size of $i$-maximal PMDs, we focus without loss of generality 
on the case $i=k$. The proof proceeds by an appropriate application of Lemma~\ref{lem:moments-imply-dtv}
and a counting argument. The idea is fairly simple: for a $k$-maximal $(n, k)$-PMD
$X$, we start by approximating the means $s_j(X)$, $1 \le j \le k-1$, within a factor of $2$,
and then impose an appropriate grid on its low-degree {parameter} moments.

We associate to such a $k$-maximal $(n, k)$-PMD $X$ the following data, 
and claim that if $X$ and $Y$ are two $k$-maximal PMDs with the same data, 
then their total variational distance is at most $\eps' \eqdef \eps/k.$ 
An  $\eps'$-cover for the set of $k$-maximal $(n, k)$-PMDs 
can then be obtained by taking one representative $X$ for each possible setting of the data in question. 
{Let us denote $\delta'_1 \eqdef \delta_1(\eps')$, $\delta'_2 \eqdef \delta_2(\eps')$,
$\gamma' \eqdef \gamma(\eps')$, $K'_1 \eqdef K_1 = 2$, and $K_2' \eqdef K_2(\eps')$, 
where the functions $\delta_1(\eps)$, $\delta_2(\eps)$, $\gamma(\eps)$, and $K_2(\eps)$
are defined in the statement of Lemma~\ref{lem:moments-imply-dtv}.}

In particular, for any $X,$ we partition the coordinates of $[n]$ into the sets 
$A_1 = \overline{E(\delta'_1, X)}$, $A_2 = \overline{E(\delta'_2, X)} \setminus A_1,$ 
and $A_3 = E(\delta'_2, X).$ We use these subsets to define $X^{(1)}$, $X^{(2)}$ and $X^{(3)}$ 
on $n_1,n_2,n_3$ $k$-CRVs respectively.

Now, to $X$ we associate the following data:
\begin{itemize}
\item $n_1,n_2,n_3$.
\item The nearest integer to $\log_2(s_j(X)+1)$ for each $j$, $1 \le j \le k-1.$
\item The nearest integer multiple of $\gamma'/(2k)^{|m|_1}$ to each of the $M_m(X^{(1)})$ for $|m|_1\leq 2$.
\item The nearest integer multiple of $\gamma'/(2k)^{|m|_1}$ to $M_m(X^{(2)})$ for $|m|_1\leq K'_2$.
\item Rounding of each of the $p_{i,j}$ for $i\in A_3$ to the nearest integer multiple of $\eps'/(kn_3)$.
\end{itemize}
First, note that if $X$ and $Y$ are $k$-maximal $(n, k)$-PMDs with the same associated data, 
then they are partitioned as $X=X^{(1)}+X^{(2)}+X^{(3)},$ $Y=Y^{(1)}+Y^{(2)}+Y^{(3)},$ 
where $X^{(t)},Y^{(t)}$, $t \le 2$, have no $\delta'_t$-exceptional variables and have 
the same number of component $k$-CRVs. 
Furthermore, $1+s_j(X)$ and $1+s_j(Y)$ differ by at most a factor of $2$ for each $j$. 
We also must have that $|M_m(X^{(t)})-M_m(Y^{(t)})|(2k)^{|m|_1} \leq \gamma'$ for $|m|_1 \leq K'_t$, 
and there is a bijection between the variables in $X^{(3)}$ and those in $Y^{(3)}$ 
so that corresponding variables differ by at most $\eps'/(kn_3)$ in each parameter 
(and, thus, differ by at most $\eps'/n_3$ in total variation distance). 
Lemma~\ref{lem:moments-imply-dtv} implies that if $X$ and $Y$ have the same data, then $\dtv(X,Y) \le \eps'.$ 
Hence, this does provide an $\eps'$-cover for the set of $k$-maximal $(n, k)$-PMDs.

We are left to prove that this cover is of the appropriate size. 
To do that, we need to prove a bound on the number of possible values 
that can be taken by the above data. We have at most $n$ choices for each $n_i$, 
and $O(\log(n))$ choices for each of the $k$ rounded values of $\log_2(s_j(X)+1)$ 
(since each is an integer between $0$ and $\log_2(n)+1$). 
$X^{(1)}$ has $O(k^2)$ {parameter} moments with $|m|_1\leq 2$, 
and there are at most $O(kn/\gamma')$ options for each of them (since each {parameter} moment is at most $n$). 
There are $O((k+K'_2)^{k-1})$ {parameter} moments of $X^{(2)}$ that need to be considered.
By Claim~\ref{clm:tech}, each such {parameter} moment has magnitude at most $O(k/{\delta'_1}^2)$, 
and, by our aforementioned rounding, needs to be evaluated to additive accuracy
at worst $\gamma'/(2k)^{K'_2}.$ 
Finally, note that $n_3  = |A_3| \le k/{\delta'_2}^2,$
since the coordinates of $A_3$ are
 $\delta'_2$-exceptional under $X$.
Each of the corresponding $O(k^2/{\delta'_2}^2)$ parameters $p_{i,j}$ for $i\in A_3$
need to be approximated to precision $\eps'/(kn_3)$. 
We remark that the number of such parameters is less than 
$O(k\log(1/\eps')/\log\log(1/\eps'))^{k-1}$, since $k\geq 3$.
Putting this together, we obtain that the number of possible values for this data is at most
$n^{O(k^2)} (1/\eps')^{O(k\log(1/\eps')/\log\log(1/\eps'))^{k-1}}.$
This completes the proof of Proposition~\ref{prop:cover-upper-maximal}.
\end{proof}

The proof of Proposition~\ref{prop:cover-upper-maximal} can be made algorithmic 
using Dynamic Programming, yielding an efficient construction of a proper $\eps$-cover for the set of all $(n, k)$-PMDs.

\begin{theorem} \label{thm:dp}
Let $S_1,S_2,\ldots,S_n$ be sets of $k$-CRVs. 
Let $\mathcal{S}$ be the set of $(n, k)$-PMDs of the form $\sum_{\ell=1}^n X_{\ell}$, 
where $X_{\ell}\in S_{\ell}.$ There exists an algorithm that runs in time
{$$n^{O(k^3)} \cdot (k/\eps)^{O(k^3\log(k/\eps)/\log\log(k/\eps))^{k-1}} \cdot  \max_{\ell \in [n]} |S_\ell| \;,$$}
and returns an $\eps$-cover of $\mathcal{S}.$
\end{theorem}

Observe that if we choose each $S_i$ to be a $\delta$-cover for the set of all $k$-CRVs, with $\delta = \eps/n,$ by the 
subadditivity of the total variation distance for independent random variables, we obtain an $\eps$-cover
for $\mathcal{M}_{n, k}$, the set of all $(n, k)$-PMDs. It is easy to see that the set of $k$-CRVs has 
an explicit $\delta$-cover of size $O(1/\delta)^k.$ 
This gives the following corollary:

\begin{corollary} \label{cor:cover-proper-alg}
There exists an algorithm that, on input $n, k \in \Z_+,$ $k>2,$ and $\eps>0,$ computes a proper $\eps$-cover
for the set  $\mathcal{M}_{n, k}$ and runs in time 
{$n^{O(k^3)} \cdot (k/\eps)^{O(k^3\log(k/\eps)/\log\log(k/\eps))^{k-1}} \;.$}
\end{corollary}

\begin{proof}[Proof of Theorem~\ref{thm:dp}]
The high-level idea is to split each such PMD into its $i$-maximal PMD components 
and approximate each to total variation distance $\eps' \eqdef \eps/k.$ 
We do this by keeping track of the appropriate data, 
along the lines of Proposition~\ref{prop:cover-upper-maximal},
and using dynamic programming.

{For the sake of readability, we start by introducing the notation that is used throughout this proof.
We use $X$ to denote a generic $(n, k)$-PMD, and $X^{i}$, $1 \le i \le k$,
to denote its $i$-maximal PMD components. For an $(n, k)$-PMD and a vector 
$m = (m_1, \ldots, m_k) \in \Z^k_+$, we denote its $m^{th}$ {parameter} moment by 
$M_m(X) = \sum_{\ell=1}^n \sum_{j=1}^k p_{\ell, j}^{m_j}.$ Throughout this proof,
we will only consider {parameter} moments of $i$-maximal PMDs, in which case the 
vector $m$ of interest will by construction satisfy 
$m_i = 0$, i.e., $m = (m_1, \ldots, m_{i-1}, 0, m_{i+1}, \ldots, m_k).$
}

In the first step of our algorithm,
we guess approximations to the quantities $1+ s_j(X^i)$ to within a factor of $2$, 
where $X^i$ is intended to be the $i$-maximal PMD component of our final PMD $X$. 
We represent these guesses in the form of a  matrix $G = (G_{i, j})_{1 \le i \neq j \le k}$. {Specifically, 
we take $G_{i, j}=(2^{a_{i, j}}+3)/4$ for each integer $a_{i, j} \geq 0$, where each $a_{i, j}$ is bounded from above by
$O(\log n).$} 
For each fixed guess $G$, we proceed as follows:
For $h \in [n],$ we denote by $\mathcal{S}_{h}$ the set of all $(h, k)$-PMDs
of the form $\sum_{\ell=1}^{h} X_{\ell}$, where $X_{\ell} \in S_{\ell}.$
For each $h \in [n],$ we compute the set of all possible (distinct) data 
$D_G(X)$, where $X \in \mathcal{S}_{h}.$ The data $D_G(X)$ consists of the following:
\begin{itemize}
\item The number of $i$-maximal $k$-CRVs of $X$, for each $i$, $1 \le i \le k.$

\item Letting $X^i$ denote the $i$-maximal PMD component of $X$, 
we partition the $k$-CRV components of $X^i$ into three sets based on whether or not they are 
$\delta'_1$-exceptional or $\delta'_2$-exceptional {\em with respect to our guess matrix $G$} for $1+s_j(X^i).$
Formally, we have the following definition:
\begin{definition}
Let $X^{i}$ be an $i$-maximal $(h_i, k)$-PMD with $X^{i} = \sum_{\ell \in A^i} X_{\ell}$ and $0 < \delta \le 1.$
We say that a particular component $k$-CRV $X_{\ell}$, $\ell \in A^i$, is $\delta$-exceptional with respect to $G = (G_{i, j}),$ 
if there exists a coordinate $j \neq i$, $1 \leq j \leq k$, such that {$p_{\ell, j} \geq \delta \cdot \sqrt{G_{i, j}}.$} We will denote by
$E(\delta, G) \subseteq A^i$ the set of $\delta$-exceptional coordinates of $X^i.$
\end{definition}
With this notation, we partition $A^i$ into the following three sets:
$A^i_1 = \overline{E(\delta'_1, G)}$,  $A^i_2 = \overline{E(\delta'_2, G)} \setminus A^i_1$, and $A^i_3 = E(\delta'_2, G).$
For each $i$, $1 \leq i \leq k,$ 
we store the following information:
\begin{itemize}
\item $n^i_1 = |A^i_1|,$ $n^i_2 = |A^i_2|,$ and $n^i_3 = |A^i_3|.$

\item Approximations $\widetilde{s}_{j, i}$ of the quantities 
$s_j(X^i)$, for each $j \neq i$, $1 \le j \le k$ to within an additive error of $(h/4n)$.

\item Approximations of the {parameter} moments  $M_m\left( (X^i)^{(1)} \right)$, for all $m = (m_1, \ldots, m_k) \in \Z^k_+$
with $m_i = 0$ and  $|m|_1\leq 2$, to within an additive $\left(\gamma'/(2k)^{|m|_1}\right) \cdot (n^i_1/n).$

\item Approximations of the {parameter} moments $M_m\left((X^i)^{(2)}\right)$, for all $m = (m_1, \ldots, m_k) \in \Z^k_+$
with $m_i = 0$ and $|m|_1\leq K'_2$ to within an additive $\left(\gamma'/(2k)^{|m|_1}\right) \cdot (n^i_2 \cdot {\delta'_1}^2/2k)$.
\item Rounding of each of the parameters $p_{\ell,j},$ for each $k$-CRV $X_{\ell}$, $\ell \in A^i_3$, 
to the nearest integer multiple of $\eps' {\delta'_2}^2/2k^2.$
\end{itemize}
\end{itemize}
Note that $D_G(X)$ can be stored as a vector of counts and moments. 
In particular, for the data associated with $k$-CRVs in $A^i_3,$ $1 \leq i \leq k,$ 
we can store a vector of counts of the possible roundings of the parameters using a sparse representation. 

We emphasize that our aforementioned approximate description needs to satisfy the following property:
for independent PMDs $X$ and $Y$, we have that  $D_G(X+Y)=D_G(X)+D_G(Y)$. 
This property is crucial, as it allows us to store only one PMD as a representative for each distinct data vector.
This follows from the fact that, if the property is satisfied, 
then $D_G(X+Y)$ only depends on the data associated with $X$ and $Y.$ 

To ensure this property is satisfied, for a PMD $X = \sum_{\ell=1}^n X_{\ell},$ where $X_{\ell}$ is a $k$-CRV, 
we define  $D_G(X) = \sum_{\ell=1}^n D_G(X_{\ell}).$
We now need to define $D_G(W)$ for a $k$-CRV $W$. 
For $D_G(W)$, we store the following information:
\begin{itemize}

\item The value of $i$ for which $W$ is $i$-maximal. 

\item Whether or not $W$ is $\delta'_1$-exceptional and $\delta'_2$-exceptional with respect to $G.$

\item  {$s_j(W)= \Pr[W = j]$ rounded down to a multiple of $1/4n$,} for each $j \neq i$, $1 \le j \le k$.

\item If $W$ is not $\delta'_1$-exceptional with respect to $G$, 
the nearest integer multiple of $\gamma'/(n(2k)^{|m|_1})$ 
to $M_m(W)$ for each $m \in \Z^k_+$, with $m_i = 0$ and $|m|_1\leq 2.$

\item If $W$ is $\delta'_1$-exceptional but not $\delta'_2$-exceptional with respect to $G$, 
the nearest integer multiple of  $\left(\gamma'/(2k)^{|m|_1}\right) \cdot ({\delta'_1}^2/2k)$ 
to $M_m(W),$ for each $m \in \Z^k_+$ with $m_i = 0$ and $|m|_1\leq K'_2.$

\item If $W$ is $\delta'_2$-exceptional with respect to $G$, we store roundings 
of each of the probabilities $\Pr[W=j]$ to the nearest integer multiple of $\eps' {\delta'_2}^2/2k.$
\end{itemize}


Given the above detailed description, we are ready to describe our dynamic programming based algorithm.
Recall that for each $h$, $1\leq h \leq n$, we compute sets of all possible (distinct) data 
$D_G(X)$, where $X \in \mathcal{S}_{h}.$
We do the computation by a dynamic program that works as follows: 
\begin{itemize}

\item At the beginning of step $h,$ we have a set $\mathcal{D}_{h-1}$ of all possibilities of $D_G(X)$ for PMDs 
of the form $X=\sum_{\ell=1}^{h-1} X_{\ell},$ where $X_{\ell} \in S_{\ell},$ { that have $1+\widetilde{s}_{j, i}^{D_G(X)} \leq 2 G_{i,j}$}. Moreover,
for each $D \in \mathcal{D}_{h-1}$, we have a representative PMD $Y_D$, 
given in terms of its $k$-CRVs, that satisfies $D_G(Y_D)=D.$ 

\item To compute $\mathcal{D}_{h}$, we proceed as follows:
We start by computing $D_G(X_h)$, for each $X_h \in S_h.$ 
Then, we compute a list of possible data for  $\mathcal{D}_{h}$ as follows: 
For each $D \in \mathcal{D}_{h-1}$ and $X_h \in S_h$, 
we compute the data $D+D_G(X_h)$, and the $k$-CRVs of the 
PMD $Y_D+X_h$ that has this data, 
since $D_G(Y_D+X_h)=D+D_G(X_h)$. We then remove duplicate data from this list, 
arbitrarily keeping one PMD that can produce the data. 
{We then remove data $D$ where $1+\widetilde{s}_{j, i}^{D} \geq 2 G_{i,j}$.}
This gives our set of possible data $\mathcal{D}_{h}.$ 
Now, we note that  $\mathcal{D}_{h}$ contains all possible
data of PMDs of the form $\sum_{\ell=1}^{h} X_{\ell},$ where each $X_{\ell} \in S_{\ell},$  { that have $1+\widetilde{s}_{j, i}^{D_G(X)} \leq 2 G_{i,j}$},
and for each distinct possibility, we have an explicit PMD that has this data.

\item After step $n$, for each $D \in \mathcal{D}_n,$ 
we output the data and the associated explicit PMD, 
if the following condition is satisfied:

\begin{condition} \label{cond:guess-approx}
For each $i, j \in \Z_+$, with $1 \leq i \neq j \leq k,$ 
it holds (a) $G_{i, j} \leq  1 + \max \{0, \widetilde{s}_{j, i}^{D_G(X)}-1/4 \}$ and 
(b) $1+\widetilde{s}_{j, i}^{D_G(X)} \leq 2 G_{i,j},$
where ${\widetilde{s}_{j, i}}^{D_G(X)}$ is the approximation to $s_j(X^i)$ in $D_G(X),$
and $G_{i, j}$  is the guess for $1+s_j(X^i)$  in $G.$
\end{condition}

\end{itemize}
We claim that the above computation, performed for all values of $G$,
outputs an $\eps$-cover of the set $\mathcal{S}$.
This is formally established using the following claim:
\begin{claim} \label{clm:data-gives-cover}
\begin{itemize}
\item[(i)] For any $X,Y \in \mathcal{S}$, if  $D_G(X)=D_G(Y)$ 
and $D_G(X)$ satisfies Condition~\ref{cond:guess-approx}, 
then $\dtv(X,Y) \leq \eps$.

\item[(ii)] For any $X \in \mathcal{S}$, there exists a $G$ such that $D_G(X)$ satisfies
for $i, j \in \Z_+$, with $1 \leq i \neq j \leq k,$ 
(a) $G_{i, j} \leq  1 + \max \{0, \widetilde{s}_{j, i}^{D_G(X)}-3/4 \}$ and 
(b) $1+\widetilde{s}_{j, i}^{D_G(X)} \leq 2 G_{i,j},$ hence also
Condition~\ref{cond:guess-approx}.
\end{itemize}
\end{claim}

\begin{remark}
\emph{Note that Condition (a) in statement (ii) of the claim above is slightly stronger than that in 
Condition~\ref{cond:guess-approx}. This slightly stronger condition will be needed for the anonymous games 
application in the following section.}
\end{remark}

\begin{proof}
To prove (i), we want to use Lemma \ref{lem:moments-imply-dtv} to show that that for all $i \in [k],$ 
the $i$-maximal components of $X$ and $Y$ are close, i.e., that $\dtv(X^i,Y^i) \leq \eps/k.$ 
To do this, we proceed as follows:

We first show that $\frac{1}{2} \leq (1+s_j(X^i))/(1+s_j(Y^i)) \leq 2$. 
By the definition of $\widetilde{s}_{j, i}^{D_G(X)}$, 
for each $i$-maximal $k$-CRV $X_{\ell} \in A^i,$ 
we have $s_j(X_{\ell}) - \frac{1}{4n} \leq \widetilde{s}_{j, i}^{D_G(X)} \leq s_j(X_{\ell}).$ 
Thus, for $X^i=\sum_{\ell \in A^i} X_{\ell},$ 
we have that $s_j(X^i) - (1/4) \leq \widetilde{s}_{j, i}^{D_G(X)} \leq s_j(X^i).$ 
Since $s_j(X^i) \geq 0,$ we have that $\max\{0, s_j(X^i) - 1/4 \} \leq \widetilde{s}_{j, i}^{D_G(X)} \leq s_j(X^i).$ 
Combining this with Condition~\ref{cond:guess-approx} 
yields that 
\begin{equation} \label{eq:Gs-bound} 
G_{i,j} \leq 1+s_j(X^i) \leq 2G_{i,j}. 
\end{equation} 
Since an identical inequality holds for $Y,$ we have that 
$\frac{1}{2} \leq (1+s_j(X^i))/(1+s_j(Y^i)) \leq 2.$

We next show that the set of coordinates $A^i_1$ for $X$ 
does not contain any $\delta'_1$ exceptional variables for $X^i.$ 
For all $\ell \in A^i_1,$ 
since $\ell$ is not $\delta'_1$-exceptional with respect to $G,$ 
using (\ref{eq:Gs-bound}), we have that 
$p_{\ell, j} \leq \delta'_1 \cdot \sqrt{G_{i, j}} \leq \sqrt{1+s_j(X^i)}.$ 
Similarly, it follows that $A^i_2$ for $X$ does not contain any $\delta'_2$-exceptional variables.
The same statements also directly follow for $Y.$

Now, we obtain bounds on the size of the $A^i_t$'s, $t=1, 2, 3.$ 
We trivially have $|A^i_1| \leq n.$ 
From (\ref{eq:Gs-bound}), we have that all variables in $A^i_2$ are $\delta'_1$-exceptional with respect to $G_{i,j}.$
 If we denote by $E_j \subseteq A^i_2$ the set of $\ell \in A^i_2$ with 
 $p_{\ell,j} \geq \delta'_1 \sqrt{G_{i,j}},$ then using (\ref{eq:Gs-bound}), we have that
\begin{eqnarray*}
s_j(X^i) = \sum_{\ell \in A^i} p_{\ell, j} \geq 
\sum_{\ell \in A^i} p_{\ell, j}^2 
\geq \sum_{\ell \in E_j} p_{\ell,j}^2 
\geq \delta'^2_1 |E_j| G_{i,j} 
\geq \delta'^2_1 |E_j| (1+s_j(X^i))/2 
\geq s_j(X^i) \cdot \delta'^2_1 |E_j|/2.
\end{eqnarray*} 
Thus, $|E_j| \leq 2/\delta'^2_1.$
Since $A^i_2 = \bigcup_{j=1}^k E_j,$ 
we have $|A^i_2| \leq 2k/\delta'^2_1.$
Similarly, we have $|A^i_3| \leq 2k/\delta'^2_2.$

For $\ell \in A^i_1$, $D_G(X_\ell)$ contains an approximation to $M_m(X_\ell)$ 
for each $m \in \Z^{k}_+$ {with $m_i=0$ and} $|m|_1 \leq 2,$ 
to within accuracy $\gamma'/(2n(2k)^{|m|_1})$. 
Since $|A^i_1| \leq n,$  we have that 
$D_G(X^i)$ contains an approximation to $M_m\left((X^i)^{(1)}\right)$ to within $\gamma'/(2(2k)^{|m|_1}).$ 
Since a similar bound holds for $(Y^i)^{(1)},$ 
and $D_G\left((Y^i)^{(1)}\right) = D_G\left( (X^i)^{(1)} \right),$ 
we have that $|M_m\left((X^i)^{(1)}\right)-M_m\left((Y^i)^{(1)}\right)|  \leq \gamma'/(2k)^{|m|_1}.$

Similarly, for $\ell \in A^i_2$, $D_G(X_\ell)$ contains an approximation to $M_m(X_\ell)$ for each 
$m \in \Z^{k}_+$ {with $m_i=0$ and} $|m|_1 \leq K'_2$ to within accuracy 
$(1/2) \cdot \left(\gamma'/(2k)^{|m|_1}\right) \cdot ({\delta'_1}^2/2k).$ 
Since $|A^i_2| \leq 2k/\delta'^2_2$, $D_G(X^i)$ contains an approximation 
to $M_m\left((X^i)^{(2)}\right)$ to within $\gamma'/(2(2k)^{|m|_1}).$
Since a similar bound holds for $(Y^i)^{(2)}$ 
and $D_G\left((Y^i)^{(2)}\right) = D_G\left((X^i)^{(2)}\right),$ 
we have that $|M_m\left((X^i)^{(2)}\right)-M_m\left((Y^i)^{(2)}\right)|   \leq \gamma'/(2k)^{|m|_1}.$

Finally, for $\ell \in A^i_3$, $D_G(X_\ell)$ contains an approximation to $p_{\ell,j}$ for all $j \neq i$ to within $\eps' {\delta'_2}^2/4k^2.$ 
The counts of variables with these approximations are the same in $(X^i)^{(3)}$ and $(Y^i)^{(3)}.$ 
So, there is bijection $f$ from $A^i_3(X)$ to $A^i_3(Y)$ 
such that an $\ell'  =f(\ell)$ has $D_G(X_\ell)=D_G(Y_{\ell'}).$
Then, we have that $\dtv(X_\ell, Y_{\ell'}) \leq \sum_{j \neq i} \eps' {\delta'_2}^2/2k^2 \leq \eps' {\delta'_2}^2/2k \leq \eps'/|A^i_3|.$

We now have all the necessary conditions to apply Lemma \ref{lem:moments-imply-dtv}, 
yielding that $\dtv(X^i,Y^i) \leq \eps/k.$ By the sub-additivity of total variational distance, 
we have $\dtv(X,Y) \leq \eps,$ proving statement (i) of the claim.

To prove (ii), it suffices to show that for any $i, j$ there is a $G_{i,j}$ that satisfies the inequalities claimed. 
Recall that $G_{i,j}$ takes values of the form $(2^a+3)/4$ for an integer $a \geq 0.$ 
For $a=0,$ $G_{i,j}=1$ and the inequality $G_{i, j} \leq  1 + \max \{0,\widetilde{s}_{j, i}^{D_G(X)}-3/4 \}$ 
 is satisfied for any value of $\widetilde{s}_{j, i}^{D_G(X)}.$ 
 When $a \geq 1$, $G_{i,j} > 1,$ so the inequality $G_{i, j} \leq  1 + \max \{0,\widetilde{s}_{j, i}^{D_G(X)}-3/4 \}$
 is only satisfied when $G_{i, j} \leq  1 + \widetilde{s}_{j, i}^{D_G(X)}-3/4,$ 
 i.e., when $\widetilde{s}_{j, i}^{D_G(X)} \geq G_{i, j} - 1/4 = (2^a + 2)/4 = (2^{a-1}+1)/2 $.
The second inequality in (i) is satisfied when 
$\widetilde{s}_{j, i}^{D_G(X)} \leq 2 G_{i,j} - 1 = 2 \cdot(2^{a} + 1)/4 = (2^{a} + 1)/2$. 

Summarizing, for $a=0$, we need that $\widetilde{s}_{j, i}^{D_G(X)} \in [0,1],$ 
and for $a \geq 1$, we need that $\widetilde{s}_{j, i}^{D_G(X)} \in [(2^{a-1}+1)/2, (2^{a} + 1)/2]$.
So, there is a $G_{i,j}=(2^{a_{i,j}}+3)/4$ for which the required inequalities are satisfied.
Thus, there is a $G$ for which we get the necessary inequalities for all $1 \leq i,j \leq k$ with $i \neq j.$ 
This completes the proof of (ii).
\end{proof}

We now bound the running time:

\begin{claim} \label{clm:dp-cover-size}
For a generic $(n, k)$-PMD $X$,
the number of possible values taken by $D_G(X)$ considered is at most ${n}^{O(k^3)}({k}/\eps)^{O({k^3}\log(1/\eps)/\log\log(1/\eps))^{k-1}}.$
\end{claim}
\begin{proof}

For a fixed $G$, we consider the number of possibilities for $D_G(X^i)$ for each $1 \leq i \leq k$.

For each $j \neq i,$ we approximate $s_j(X^i)$ up to an additive $1/(4n).$ 
Since $0 \leq s_j(X_i) \leq n,$ there are at most $4n^2$ possibilities. 
For all such $j$ we have $O(n^{2k})$ possibilities.

We approximate the {parameter} moments of $M_m\left( (X^i)^{(1)} \right)$ as an integer 
multiple of $\gamma'/(n(2k)^{|m|_1})$ for all $m$ with $m_1 \leq 2.$
For each such $m,$ we have $0 \leq M_m\left( (X^i)^{(1)} \right) \leq n,$ 
so there are at most $n^2(2k)^{|m|_1}/\gamma'= n^2 (k\log(1/\eps))^{O(k)} (1/\eps)$ possibilities.
 There are $O(k^2)$ such $m,$ so we have $n^{O(k^2)} \cdot (k \log(1/\eps)^{O(k^3)} (1/\eps)^{O(k^2)}$ possibilities.

We approximate the {parameter} moments of $M_m\left( (X^i)^{(2)} \right)$ as a multiple of $\left(\gamma'/(2k)^{|m|_1}\right) \cdot ({\delta'_1}^2/2k)$ 
for each $m$ with $|m|_1 \leq K'_2.$ The number of $k$-CRVs in $(X^i)^{(2)}$ is $|A^i_2| \leq 2k/\delta'^2_1$ from the proof of Claim \ref{clm:data-gives-cover}. 
So, for each $m$, we have $0 \leq M_m\left( (X^i)^{(2)} \right) \leq |A^i_2|,$ 
and there are at most $(2k)^{K'_2+2}/(\gamma'\delta'^2_1 \delta'^2_2)=k^{O(k+\ln(k/\eps)/\ln\ln(k/\eps))} \ln(1/\eps)^{O(k)}/\eps=(k/\eps)^{O(k)}$ possibilities. 
Since there are at most $$K'^{k-1}_2=O((\ln(k/\eps)/\ln\ln(k/\eps)+k)^{k-1}$$ such moments, there are  $(k/\eps)^{O(k\ln(k/\eps)/\ln\ln(k/\eps)+k^2)^{k-1}}$ possibilities.

We approximate each $X_\ell$ for $\ell \in A^i_3$ as a $k$-CRV 
whose probabilities are multiples of $\eps {\delta'_2}^2/2k^2.$ 
So, there are $(2k^2/(\eps \delta'_2)^k=(k/\eps)^{O(k)}$ possible $k$-CRVs. 
Since there may be $|A^i_3| \leq 2k/\delta'^2_2=2k^2 \log^{3/2}(k/\eps)$ such $k$-CRVs, 
there are $(k/\eps)^{O(k^3\log^{3/2}(k/\eps))}$ possibilities.

Multiplying these together, for every $G,$ 
there are at most $n^{O(k^2)} (k/\eps)^{O(k\ln(k/\eps)/\ln\ln(k/\eps)+k^2)^{k-1}}$ possible values of $D_G(X^i).$ 
Hence, there are at most $n^{O(k^3)} (k/\eps)^{O(k^3 \ln(k/\eps)/\ln\ln(k/\eps))^{k-1}}$ possible values of $D_G(X)$ for a given $G.$
Finally, there are $O(\log n)^{k^2}$ possible values of $G_{i,j},$ 
since  $G_{i,j}=(2^{a_{i,j}}+3)/4,$ for integers $a_{i,j},$ and we do not need to consider $G_{i,j} > n.$ 
Therefore, the number of possible values of $D_G(X)$ is at most  $n^{O(k^3)} \cdot (k/\eps)^{O(k^3\ln(k/\eps)/\ln\ln(k/\eps))^{k-1}}.$
\end{proof}

The runtime of the algorithm is dominated by the runtime of the substep of each step $h,$ 
where we calculate $D+D_G(X_h)$ for all $D \in \mathcal{D}_{h-1}$ and $X_h \in S_h.$
Note that $D$ and $D_G(X_h)$ are vectors with $O(K_2'^k)=O(\log(k/\eps)/\log \log(k/\eps)+k)^k$ non-zero coordinates. 
So, the runtime of step $h$ is at most 
$$|\mathcal{D}_{h-1}| \cdot |S_h| \cdot O((K'_2)^k)= |S_h| \cdot n^{O(k^3)} \cdot (k/\eps)^{O(k^3\ln(k/\eps)/\ln\ln(k/\eps))^{k-1}} \;,$$ 
by Claim \ref{clm:dp-cover-size}. 
The overall runtime of the algorithm is thus $n^{O(k^3)} \cdot (k/\eps)^{O(k^3\ln(k/\eps)/\ln\ln(k/\eps))^{k-1}} \cdot \max_h |S_h|.$
This completes the proof of Theorem~\ref{thm:dp}.
\end{proof}

\subsection{An EPTAS for Nash Equilibria in Anonymous Games} \label{ssec:anonymous}

In this subsection, we describe our EPTAS for computing Nash equilibria in anonymous games:

\begin{theorem} \label{thm:anon}
There exists an {$n^{O(k^3)} \cdot (k/\eps)^{O(k^3\log(k/\eps)/\log\log(k/\eps))^{k-1}}$}-time 
algorithm for computing a (well-supported) $\eps$-Nash Equilibrium in an $n$-player, $k$-strategy anonymous game.
\end{theorem}

This subsection is devoted to the proof of Theorem~\ref{thm:anon}.

We compute a well-supported $\eps$-Nash equilibrium, 
using a procedure similar to~\cite{DaskalakisP2014}. 
We start by using a dynamic program very similar to that of our Theorem~\ref{thm:dp} 
in order to construct an $\eps/10$-cover. We iterate over this $\eps/10$-cover. 
For each element of the cover, we compute a set of possible $\eps/5$-best responses. 
Finally, we again use the dynamic program of Theorem~\ref{thm:dp} to check if we can construct 
this element of the cover out of best responses. 
If we can, then we have found an $\eps$-Nash equilibrium. 
Since there exists an $\eps/5$-Nash equilibrium in our cover, this procedure must produce an output. 

In more detail, to compute the aforementioned best responses, 
we use a modification of the algorithm in Theorem~\ref{thm:dp}, 
which produces output at the penultimate step. 
The reason for this modification is the following: 
For the approximate Nash equilibrium computation, 
we need the data produced by the dynamic program, not just the cover of PMDs. 
Using this data, we can subtract the data corresponding to each candidate best response. 
This allows us to approximate the distribution of the sum of the other players strategies, 
which we need in order to calculate the players expected utilities.

Recall that a mixed strategy profile for a $k$-strategy anonymous game can be represented as a set of $k$-CRVs, 
$\{ X_i \}_{i \in [k]},$ where the $k$-CRV $X_i$ describes the mixed strategy for player $i.$ 
Recall that a mixed strategy profile is an $\eps$-approximate Nash equilibrium, 
if for each player $i$ we have $\E[u^i_{X_i}(X_{-i})] \geq \E[ u^i_{\ell}(X_{-i})]-\eps,$ for $\ell \in [k],$ 
where $X_{-i}=\sum_{j \in [n] \setminus \{ i \}} X_j$ is the distribution of the sum of other players strategies. 
A strategy profile is an  $\eps$-well-supported Nash equilibrium if for each player $i$, $\E[u^i_{\ell'}(X_{-i})] \geq \E[ u^i_{\ell}(X_{-i})]-\eps$ for each $\ell \in [k]$ and $e_{\ell'}$ in the support of $X_i$.
If this holds for one player $i$, then we call $X_i$ an $\eps$(-well-supported) best response to $X_{-i}$.
\begin{lemma} \label{lem:best-response-dtv} 
Suppose that $X_i$ is a $\delta$-best response to $X_{-i}$ for player $i$. 
Then, if an $n-1$ PMD $Y_{-i}$ has $\dtv(X_{-i},Y_{-i}) \leq \eps$, 
$X_i$ is a $(\delta+2\eps)$-best response to $Y_{-i}$. 
If, additionally, a $k$-CRV $Y_i$ has $\Pr[Y_i=e_j]=0$ 
for all $j$ with $\Pr[X_i=e_j]=0,$ 
then $Y_i$ is a $(\delta+2\eps)$-best response to $Y_{-i}.$ 
 \end{lemma}
\begin{proof}
Since $u^i_\ell(x) \in [0,1]$ for $\ell \in [k]$ and any $x$ in the support of $X_{-i}$, 
we have that $\E[u^i_{\ell}(X_{-i})] - \E[u^i_{\ell}(Y_{-i})] \leq \dtv(X_{-i}, Y_{-i})).$ 
Similarly, we have $\E[u^i_{\ell}(X_{-i})] - \E[u^i_{\ell}(Y_{-i})] \leq \dtv(X_{-i}, Y_{-i}).$ 
Thus, for all $e_{\ell'}$ in the support of $X_i$ and all $\ell \in [k],$ we have
$$\E[u^i_{\ell'}(Y_{-i})] \geq \E[u^i_{\ell'}(X_{-i})]-\eps \geq \E[ u^i_{\ell}(X_{-i})] - \eps - \delta \geq \E[ u^i_{\ell}(Y_{-i})] - 2\eps - \delta .$$
That is, $X_i$ is a $(\delta+2\eps)$-best response to $Y_{-i}.$ 
Since the support of $Y_i$ is a subset if the support of $X_i,$ 
$Y_i$ is also a $(\delta+2\eps)$-best response to $Y_{-i}.$
\end{proof}

We note that by rounding the entries of an actual Nash Equilibrium, 
there exists an $\eps/5$-Nash equilibrium 
where all the probabilities of all the strategies are integer multiples of $\epsilon/(10kn)$:
\begin{claim} \label{clm:nash} 
There is an $\eps/5$-well-supported Nash equilibrium $\{ X_i\},$ 
where the probabilities $\Pr[X_i=e_j]$ 
are multiples of $\eps/(10kn),$ for all $1 \leq i \leq n$ and $1 \leq j \leq k.$
\end{claim}
\begin{proof}
By Nash's Theorem, there is a Nash equilibrium $\{ Y_i \}$. 
We construct $\{X_i\}$ from $\{Y_i\}$ as follows: 
If $Y_i$ is $\ell$-maximal, then for every $j \neq \ell$, we set $\Pr[X_i=e_j]$ to be $\Pr[Y_i=e_j]$ 
rounded down to a multiple of  $\eps/(10kn)$ and $\Pr[X_i=e_\ell] = 1 - \sum_{j \neq \ell} \Pr[X_j = e_j].$
Now, we have $\dtv(X_i,Y_i) \leq \eps/(10n)$ and the support of $X_i$ is a subset of the support of $Y_j.$ 
By the sub-additivity of total variational distance, for every $i$ we have 
$\dtv(X_{-i},Y_{-i}) \leq \eps/10.$ Since $\{ Y_i \}$ is a Nash equilibrium, 
for all players $i$, $Y_i$ is a $0$-best response  to $Y_{-i}.$
By Lemma \ref{lem:best-response-dtv}, we have that $X_i$ is a $2\eps/10$-best response to $X_{-i}$ for all players $i.$
Hence, $\{X_i\}$ is an $\eps/5$-well supported Nash equilibrium.
\end{proof}

Let $S$ be the set of all $k$-CRVs whose probabilities are multiples of $\epsilon/(10kn).$
We will require a modification of the algorithm from Theorem \ref{thm:dp} 
(applied with $S_i \eqdef S$ for all $i,$ and $\eps \eqdef \eps/5$), 
which produces output at both step $n$ and step $n-1.$
Specifically, in addition to outputting a subset 
$V_{G,n} \subseteq \mathcal{D}_{G,n}$ of the data of possible $(n,k)$-PMDs 
that satisfy conditions (a) and (b) of Claim`\ref{clm:data-gives-cover} (ii), 
we output the subset $V_{G,n-1} \subseteq \mathcal{D}_{G,n-1}$ of the data 
of possible $(n-1,k)$-PMDs that satisfy the slightly weaker conditions 
(a) and (b) of Condition~\ref{cond:guess-approx} .

In more detail, we need the following guarantees about the output of our modified algorithm:
\begin{claim} \label{clm:reqs} 
For every PMD $X=\sum_{i=1}^n X_i$ and $X_{-j}  = \sum_{i \in [n] \setminus j} X_i$, for some $1 \leq j \leq n,$ 
and any $X_i \in S,$ for $1 \leq i \leq n$, we have:
\begin{itemize}
 \item There is a guess $G,$ such that $D_G(X) \in V_{G,n}.$
 
 \item For any $G$ such that $D_G(X) \in V_{G,n},$ 
 we also have $D_G(X)-D_G(X_j) = D_G(X_{-j}) \in V_{G,n-1}.$
  
 \item If $D_G(X) \in V_{G,n},$ for any PMD $Y$ with $D_G(Y)=D_G(X)$ 
 or $D_G(Y)=D_G(X_{-j}),$ 
 we have $\dtv(X,Y) \leq \eps/5$ or $\dtv(X_{-j},Y) \leq \eps/5$ respectively. 
 \end{itemize}
 \end{claim}
 \begin{proof}
 By Claim \ref{clm:data-gives-cover} (ii), there is a $G$ such that $D_G(X)$ satisfies conditions (a) and (b) and so $D_G(X)  \in V_{G,n}$.
 
 We note that by the correctness of the dynamic program, since $X_{-j}$ is a sum of 
 $n-1$ many $k$-CRVs in $S,$ we have $D_G(X_{-j}) \in \mathcal{D}_{G,n-1}.$ 
 To show that it is in $V_{G,n-1},$ we need to show that all 
 ${\widetilde{s}_{h, i}}^{D_G(X_{-j})}$ satisfy Condition~\ref{cond:guess-approx}, 
 for all $1 \leq i,h \leq k$ and $h \neq i$. 
 We know that ${\widetilde{s}_{h, i}}^{D_G(X)}$ satisfies the stronger conditions (a) and (b) of Claim \ref{clm:data-gives-cover} (ii). 
 All we need to show is that $${\widetilde{s}_{h, i}}^{D_G(X)} - 1/2 \leq {\widetilde{s}_{h, i}}^{D_G(X_{-j})} \leq  {\widetilde{s}_{h, i}}^{D_G(X)}.$$ 
 This condition is trivial unless $X_j$ is $i$-maximal. 
 If it is, we note that $\Pr[X_j=e_h] \leq \Pr[X_j=e_i],$ 
 and so ${\widetilde{s}_{j, i}}^{D_G(X_j)} \leq \Pr[X_j=e_h] \leq 1/2.$ 
 Thus, $D_G(X_{-j}) = D_G(X)-D_G(X_j) \in V_{G,n-1}.$

 We now have that both $D_G(X)$ and $D_G(X_{-j})$ satisfy Condition~\ref{cond:guess-approx}. 
 Therefore, Claim \ref{clm:data-gives-cover} (i) yields the third claim.
 \end{proof}
 
 We note that we can calculate the expected utilities efficiently to sufficient precision:
\begin{claim}
Given an anonymous game $(n, k, \{u^i_{\ell}\}_{i \in [n], \ell \in [k]})$ with each utility given to within an additive $\eps/2$ using $O(\log(1/\eps))$ bits, 
and given a PMD $X$ in terms of its constituent $k$-CRVs $X_i$, 
we can approximate the expected utility $\E[u^i_{\ell}(\sum_{j \neq i} X_i)]$ for any player $i$ 
and pure strategy $\ell$ to within $\eps$ in time $O(n^{k+1} \cdot k\log(n) \cdot \polylog(1/\eps)).$
\end{claim}
\begin{proof}
We can compute the probability mass function of $X_{-i}=\sum_{j \neq i} X_j$ by using the FFT on $[n]^k$. 
We calculate the DFT of each $X_i$, $\wh{X}_i$, calculate the DFT of $X_{-i}$, $\wh{X_{-i}}(\xi)=\prod_{j \neq i} \wh{X}_j,$ 
and finally compute the inverse DFT. To do this within $\eps/2$ total variational error needs time $O(n^{k+1} \cdot k\log(n) \polylog(1/\eps)),$ 
since we need to use the FFT algorithm $n+1$ times. 
We then use this approximate pmf to compute the expectation $\E[u^i_{\ell}(X_{-i})]= \sum_{x} u^i_{\ell}(x) X_i(x).$ 
This takes time $O(n^k)$ and gives error $\eps$.
\end{proof}
Henceforth, we will assume that we can compute these expectations exactly, 
but it should be clear that computing them to within a suitably small $O(\eps)$ error suffices.
  
 \begin{proof}[Proof of Theorem \ref{thm:anon}]
 We use the modified dynamic programming algorithm 
 given above to produce an $\eps/5$-cover with explicit sets $V_{G, n}$, $V_{G, n-1}$ 
 of data and PMDs which produce each output data.
 
Then, for each $G$ and for each $D \in V_{G,n},$ 
we try to construct an $\eps$-Nash equilibrium whose associated PMD $X$ has $D_G(X)=D.$ 
Firstly, for each player $i$ we compute a set $S_i \subseteq S$ of best responses to $X.$ 
To do this, we check each $X_i \in S$ individually. 
We first check if $D_G(X)-D_G(X_i) \in V_{G,n-1}.$ 
If it is not, then Claim \ref{clm:reqs} implies that there is no set of strategies for the other players 
$X_j \in S,$ for $j \neq i,$ such that $D_G(\sum_{i=1}^n X_i)=D.$ 
In this case, we do not put this $X_i \in S_i.$ 
If we do have $D_{-i} := D - D_G(X_i) \in V_{G,n-1},$
then we recall that the algorithm gives us an explicit $Y_{D_{-i}}$ such that $D(Y_{D_{-i}})=D_{-i}.$ 
Now, we calculate the expected utilities $\E[u^i_{\ell}(Y_{D_{-i}})]$ for each $1 \leq \ell \leq k.$ 
If $X_i$ is a $3\eps/5$-best response to $Y_{D_{-i}},$ then we add it to $S_i.$

When we have calculated the set of best responses $S_i$ for each player, 
we use the algorithm from Theorem \ref{thm:dp} with these $S_i$'s and this guess $G.$ 
If the set of data it outputs contains $D,$ 
then we output the explicit PMD $X := Y_D$ that does so 
in terms of its constituent CRVs $X=\sum_{i=1}^n X_i$ and terminate.

To prove correctness, we first show that $\{ X_i \}$ is an $\eps$-Nash equilibrium, 
and second that that the algorithm always produces an output. 
We need to show that $X_i$ is an $\eps$-best response to $X_{-i} =  \sum_{j \in [n] \setminus i} X_j.$ 
When we put $X_i$ in $S_i,$ we checked that $X_i$ was a $3\eps/5$-best response to $Y_{D_{-i}},$ 
where $D_{-i}=D-D_G(X_i).$ But note that 
$$D_G(Y_{D_{-i}})=D-D_G(X_i)=D_G(X)-D_G(X_i)=D_G(X_{-i}).$$ 
Since $D \in V_{G,n}$ and $D_G(Y_{D_{-i}})=D_G(X_{-i}),$ Claim \ref{clm:reqs} yields that 
$\dtv(X_{-i},Y_{D_{-i}}) \leq \eps/5.$ 
So, by Lemma~\ref{lem:best-response-dtv}, we indeed have that $X_i$ is 
an $\eps$-best response to $X_{-i}$. Since this holds for all $X_i,$   $X$ is an $\eps$-Nash equilibrium.

By Claim~\ref{clm:nash}, there exists an $\eps/5$-Nash equilibrium $\{ X'_i \},$ 
with each $X'_i \in S.$ By Claim \ref{clm:reqs}, we have that for $X'= \sum_{i=1}^n X'_i,$
there is a guess $G$ with $D_G(X') \in V_{G,n}.$
So, if the algorithm does not terminate successfully first, 
it eventually considers $G$ and $D := D_G(X').$ 
We next show that that the algorithm puts $X'_i$ in $S_i.$ 
For each $1 \leq i \leq n,$ $X'_{-i} =  \sum_{j \in [n] \setminus i} X'_j$ 
has $D_G(X'_{-j}) \in  V_{G,n-1}$ by Claim \ref{clm:reqs}, 
since $D_G(X') \in V_{G,n}.$
So, $D_{-i} = D-D_G(X'_i)=D_G(X')-D_G(X'_i)=D_G(X'_{-i}),$ 
and we have $D_{-i} \in V_{G,n-1}.$
Hence, the algorithm will put $X'_i$ in $S_i$ 
if $X'_i$ is an $4\eps/5$-best response to $Y_{D_{-i}}.$ 
By Claim \ref{clm:reqs}, since $D_G(X) \in V_{G,n}$ and $D_G(X_{-i})=D_G(Y_{D_{-i}}),$ 
we have $\dtv(Y_{D_{-i}},X_{-i}) \leq \eps/5.$
Since $\{ X'_i \}$ is an $\eps/5$-Nash equilibrium,
$X'_{-i}$ is an $\eps/5$-best response to $X'_{-i}.$
Since $\dtv(Y_{D_{-i}},X_{-i}) \leq \eps/5,$ 
by Lemma \ref{lem:best-response-dtv}, this implies that $X'_i$ is a 
$3 \eps/5$-best response to $Y_{D_{-i}}.$
Thus, the algorithm puts $X'_i$ in $S_i.$ 
Since each $X'_i$ satisfies $X'_i \in S_i,$ 
by Theorem \ref{thm:dp}, the algorithm from that theorem 
outputs a set of data that includes $D_G(X')=D.$
Therefore, if the algorithm does not terminate successfully first, 
when it considers $G$ and $D,$ it will produce an output.  
This completes the proof of Theorem~\ref{thm:anon}.
\end{proof}


\paragraph{Threat points in anonymous games.}
As an additional application of our proper cover construction, we give an EPTAS
for computing threat points in anonymous games~\cite{BorgsCIKMP08}.

\begin{definition}
The threat point of an anonymous game $(n, k, \{u^i_{\ell}\}_{i \in [n], \ell \in [k]})$ is the vector $\theta$ with
$$ \theta_i = \min_{X_{-i} \in {\cal M}_{n-1,k}} \max_{1 \leq j \leq k} \E[u^i_j(X_{-i})].$$
\end{definition}
Intuitively, If all other players cooperate to try and punish player $i,$ 
then they can force her expected utility to be $\theta_i$ but no lower, 
so long as player $i$ is trying to maximize it. 
This notion has applications in finding Nash equilibria in repeated anonymous games.

\begin{corollary} 
Given an anonymous game $(n, k, \{u^i_{\ell}\}_{i \in [n], \ell \in [k]})$ with $k > 2$, 
we can compute a $\tilde{\theta}$ with $\|\theta-\tilde{\theta}\|_\infty \leq \eps$ 
in time $n^{O(k^3)} \cdot (k/\eps)^{O(k^3\log(k/\eps)/\log\log(k/\eps))^{k-1}}.$ 
Additionally, for each player $i$, we obtain strategies $X_{i,j}$ for all other players $j \neq i$ 
such that $\max_{1 \leq \ell \leq k} \E[u^i_\ell(\sum_{j \neq i} X_{i,j})] \leq \theta_i + \eps$.
\end{corollary}
\begin{proof} 
Using the dynamic programming algorithm of Theorem~\ref{thm:dp}, 
we can construct an $\eps$-cover $\mathcal{C}$ of ${\cal M}_{n-1,k}$. 
For each player $i$, we then compute $\tilde{\theta}_i=\min_{X_{-i} \in {\cal C}} \max_{1 \leq j \leq k} \E[u^i_j(X_{-i})]$ by brute force. 
Additionally, we return the $k$-CRVs $X_{i,j}$ that the algorithm gives us as the explicit parameters of the PMD $X_{-i}$ 
which achieves this minimum, i.e., with $\tilde{\theta}_i =  \max_{1 \leq j \leq k} \E[u^i_j(X_{-i})].$ 
The running time of this algorithm is dominated by the dynamic programming step.

We now show correctness.
Let $Y_{-i} \in {\cal M}_{n-1,k}$ be such that $ \theta_i = \max_{1 \leq j \leq k} \E[u^i_j(X_{-i})].$ 
Then, there exists a $Y'_{-i} \in {\cal C}$ with $\dtv(Y_{-i},Y'_{-i}) \leq \eps,$ and so we have $|\E[u^i_j(Y_{-i})]-\E[u^i_j(Y'_{-i})]| \leq \eps.$ 
Therefore,
$$\theta_i = \max_{1 \leq j \leq k} \E[u^i_j(Y_{-i})] \geq \max_{1 \leq j \leq k} \E[u^i_j(Y'_{-i})] - \eps \geq \tilde{\theta}_i-\eps \;.$$
Similarly, there is an $X_{-i} \in \mathcal{C}$ with $\tilde{\theta}_i=\max_{1 \leq j \leq k} \E[u^i_j(X_{-i})] \leq \theta_i.$ 
And so we have $|\tilde{\theta}_i -\theta_i| \leq \eps,$ as required. 
Additionally, for the $\sum_{j \neq i} X_{i,j} = X_{-i}$ we have $\max_{1 \leq \ell \leq k} \E[u^i_\ell(\sum_{j \neq i} X_{i,j})] = \tilde{\theta_i} \leq \theta_i + \eps.$
\end{proof}

\subsection{Every PMD is close to a PMD with few distinct parameters} \label{sec:distinct}

In this section, we prove our structural result that states that any PMD
is close to another PMD which is the sum of $k$-CRVs with a small number of 
distinct parameters.

\begin{theorem} \label{thm:distinct}
Let $n, k \in \Z_+,$ $k>2,$ and $\eps>0.$
For any $(n, k)$-PMD $X$, there is an $(n, k)$-PMD $Y$ such that $\dtv(X, Y) \le \eps$ satisfying the following property: 
We can write $Y=\sum_{i=1}^n Y_i$ where each $k$-CRV $Y_i$ is distributed as one of 
$$O\left((\log(k/\eps)/(\log \log (k/\eps))+k)\right)^k$$ distinct $k$-CRVs. 
\end{theorem}

The main geometric tool used to prove this is the following result from \cite{GRW15}:

\begin{lemma}[Theorem 14 from \cite{GRW15}] \label{lem:nD-Reiner}
Let $f(x)$ be a multivariate polynomial with variables $x_{i,j},$ 
for $1 \leq i \leq n$ and $1 \leq j \leq k,$ which is symmetric up to permutations of the $i$'s, 
i.e., such that for any permutation $\sigma \in \mathbb{S}_n$, 
we have that, for $(x_\sigma)_{i,j}:=x_{\sigma(i),j},$ for all $1 \leq i \leq n$ and $1 \leq j \leq k$, $f(x_\sigma)=f(x)$. 
Let $w \in \Z_{>0}^k$. Suppose that $f$ has weighted $w$ degree at most $d,$ 
i.e., each monomial $\prod_{i,j} x_{i,j}^{a_{i,j}}$ has $\sum_{i,j} w_j a_{i,j} \leq d.$  
Suppose that the minimum of $f(x)$ is attained by some $x' \in \R^n,$ 
i.e., that $f(x')= \min_{x \in \R^{n \times k}} f(x).$ 
Then, there is a point $x^{\ast}$ with $f(x^{\ast})= \min_{x \in \R^{n \times k}} f(x),$ 
such that the number of distinct $y \in \R^k$ of the form $y_j=x^{\ast}_{i,j},$ 
for some $i,$ is at most $\prod_{j=1}^k \left\lfloor \frac{d}{w_j} \right\rfloor.$
\end{lemma}

\begin{proof}[Proof of Theorem \ref{thm:distinct}]
Firstly we're going to divide our PMD into $i$-maximal PMDs. We assume wlog that $X$ is $k$-maximal below.

We divide this PMD $X$ into component PMDs $X^{(1)}$, $X^{(2)}$, $X^{(3)}$ 
according to whether these are $\delta_1$ and $\delta_2$, as in the proof of Proposition~\ref{prop:cover-upper-maximal}. 
We want to show that there exists a $Y^{(1)}$,$Y^{(2)}$ such that $X^{(1)}$ and $Y^{(1)}$ agree on the first $2$ 
moments, $X^{(2)}$ and $Y^{(2)}$ agree on the first $K_2$ moments, but each has few distinct CRVs. 
Then $Y=Y^{(1)}+Y^{(2)}+X^{(3)}$ is close to $X$ by Lemma \ref{lem:moments-imply-dtv} 
(because the first moments agree, i.e., we have $s_j(X) = s_j(Y)$).

We are going to use Lemma \ref{lem:nD-Reiner} to show that 
we can satisfy some polynomial equations $p_l(x)=0$ by setting $f$ to be a sum of squares $f(x)= \sum_l p_l(x)^2$. 
Then if the polynomial equations have a simultaneous solution at $x$, 
$f$ attains its minimum of $0$ at $x.$ 
Some of these $p_l$'s are going to be symmetric in terms of $i$. 
For the rest, we are going to have identical equations that hold for each individual $i,$ 
so $f$ overall will be symmetric.

We have $X^{(t)}$ for $t=1, 2$, and we want to construct a $Y^{(t)}$ with few distinct $k$-CRVs. 
That is, we want to find $p_{i,j}$, 
the probability that $Y_i=p_{i,j}$, for $1 \leq i \leq n$, $1 \leq j \leq k.$ 
These $p_{i,j}$'s have to satisfy certain inequalities to ensure  each $Y^{(t)}_i$ 
is a non-$\delta_t$ exceptional $k$-CRV and $p_{i,1} \leq p_{i,2} \leq \ldots \leq p_{i,k}.$ 
To do this, we will need to introduce variables whose square is the slack in each of these inequalities. 

The free variables of these equations will be $p_{i,1}, \dots, p_{i,k}, x_{i,1}, \ldots, x_{i,3k}.$ 
The equations we consider are as follows:

The following two equations mean that $Y^{(t)}_i$ is a $k$-CRV with the necessary properties:
For each $1 \leq i \leq n$ and $1 \leq j \leq k-1$, 
\begin{equation} \label{eq:0}
p_{i,j}=x_{i,j}^2
\end{equation}
\begin{equation} \label{eq:2a}
p_{i,j} + x_{i,j+k}^2 = p_{i,k}
\end{equation}
and
\begin{equation} \label{eq:3}  
p_{i,j} +x_{i,j+2k}^2 = \delta_t \sqrt{1+s_j(X)} \;.	
\end{equation}
For each $1 \leq i \leq n$,
\begin{equation} \label{eq:4} 
\sum_{j=1}^k p_{i,j}=1  \;.
\end{equation}
We need an equation that the $m^{th}$ moment of $Y^{(t)}$ is identical to the $m^{th}$ moment of $X^{(t)},$ 
i.e., 
\begin{equation} \label{eq:1a}  
\left(\sum_i \prod_j p_{i,j}^{m_j}\right)-M_m(X^{(t)})=0 \;,
\end{equation}
for each moment $m$ with $|m|_1 \leq K_t.$		

If these equations have a solution for real $p_{i,j}$'s  and $x_{i,j}$'s, 
then the $p_{i,j}$'s satisfy all the inequalities we need. 
We square all these expressions and sum them giving $f.$
Note that the slack variables $x_{i,j}$ only appear in monomials of degree 4 in $f$. 
We set the weights $w_j$ of the $p_{i,j}$ to be $1$ and the weights of the $x_{i,j}$ to be $K_t/2.$
Then, $f$ has $w$ degree $2K_t$: (\ref{eq:1a}) has degree $K_t$ in terms of $p_{i,j},$  so when we square it to put it in $f$, it has degree $2K_t$. 
So we have that, for $d=2K_t$, $\prod_{j=1}^k \left\lfloor \frac{d}{w_j} \right\rfloor =(2K_t)^{k}4^{3k}=O(K_t)^k $
Now $f$ is symmetric in terms of the $n$ different values of $i,$ 
so we can apply Lemma \ref{lem:nD-Reiner}, 
which yields that there is a minimum with $O(K_t)^k$ distinct $(k+1)$-vectors 
provided that there is any minimum. 

However, note that if we set $p'_{i,j} = \Pr[X_i=e_j]$ 
and define the $x'_{i,j}$ appropriately, we obtain an $x'$ such that $f(x')=0.$ 
Since $f$ is a sum of squares $f(x) \geq 0.$ 
So, there is an $x^{\ast}$ with $f(x^{\ast})=0,$ 
but such that $x^{\ast}$ has $O(K_t)^k$ 
distinct $4k$-vectors $(p_{i,1}^{\ast},\ldots,p_{i,k}^{\ast},x^{\ast}_{i,1}, \ldots, x^{\ast}_{i,k}).$
 
Using  the $p_{i,j}^{\ast}$'s in this solution, 
we have a $Y^{(t)}$ with $O(K_t)^k$ distinct CRVs. 
So, the $Y$ which is $O(\eps)$ close to $X$ has $(O(K_1)^k +O(K_2)^k + k (\log 1/\eps)^2)$ distinct $k$-CRVs. 
Overall, we have that any PMD is $O(k \eps)$-close to one with 
$$k \cdot O(K_2)^k=O\left(( \log(1/\eps)/(\log \log (1/\eps))+k)\right)^k$$
distinct constituent $k$-CRVs. Thus, every PMD is $\eps$-close to one 
with $k \cdot O\left((\log(k/\eps)/(\log \log (k/\eps))+k)\right)^k$ 
distinct constituent $k$-CRVs. This completes the proof.
 \end{proof}


\subsection{Cover Size Lower Bound for PMDs} \label{sec:cover-lb}

In this subsection, we prove our lower bound on the cover size of PMDs, which is restated below:

\begin{theorem} \label{thm:cover-lb-precise}(Cover Size Lower Bound for $(n, k)$-PMDs)
Let $k>2$, $k \in \Z_+$, and $\eps$ be sufficiently small as a function of $k$. 
For $n =\Omega((1/k) \cdot \log(1/\eps)/ \log\log(1/\eps))^{k-1}$
{\em any} $\eps$-cover of ${\cal M}_{n,k}$
under the total variation distance must be of size 
$n^{\Omega(k)}  \cdot (1/\eps)^{\Omega((1/k) \cdot  \log(1/\eps)/ \log\log(1/\eps))^{k-1}}.$ 
\end{theorem}

Theorem~\ref{thm:cover-lb-precise} will follow from the following theorem:

\begin{theorem} \label{thm:cover-lb-sparse}
Let $k>2$, $k \in \Z_+$, and $\eps$ be sufficiently small as a function of $k$. 
Let $n = \Omega((1/k) \cdot  \log(1/\eps)/ \log\log(1/\eps))^{k-1}.$ 
There exists a set $\mathcal{S}$ of $(n, k)$-PMDs so that 
for $x, y \in \mathcal{S}$, $x\neq y$ implies that $\dtv(x,y)\geq \eps$, and
$|\mathcal{S}| \geq (1/\eps)^{\Omega((1/k) \cdot  \log(1/\eps)/ \log\log(1/\eps))^{k-1}}.$
\end{theorem}

The proof of Theorem~\ref{thm:cover-lb-sparse} is quite elaborate and is postponed
to the following subsection. We now show how Theorem~\ref{thm:cover-lb} follows from it.
Let $n_0 \eqdef  \Theta((1/k) \cdot  \log(1/\eps)/ \log\log(1/\eps))^{k-1}.$ By Theorem~\ref{thm:cover-lb-sparse},
there exists a set  $\mathcal{S}_{n_0}$ of size
$(1/\eps)^{\Omega(n_0)}$
consisting of $(n_0, k)$-PMDs that are $\eps$-far from each other.

We construct $(n/n_0)^{\Omega(k)}$ appropriate ``shifts'' of the set $\mathcal{S}_{n_0}$, by selecting 
appropriate sets of $n-n_0$ deterministic component $k$-CRVs. These sets   
shift the mean vector of the corresponding PMD, while the remaining $n_0$ components form 
an embedding of the set $\mathcal{S}_{n_0}$. We remark that the PMDs corresponding to different shifts
have disjoint supports. Therefore, any $\eps$-cover must contain disjoint $\eps$-covers 
for each shift, which is isomorphic to $\mathcal{S}_{n_0}$. Therefore, any $\eps$-cover
must be of size 
$$(n / n_0)^{\Omega(k)} \cdot (1/\eps)^{\Omega((1/k) \cdot  \log(1/\eps)/ \log\log(1/\eps))^{k-1}} = 
n^{\Omega(k)} \cdot  (1/\eps)^{\Omega((1/k) \cdot  \log(1/\eps)/ \log\log(1/\eps))^{k-1}} \;,   $$
where the last inequality used the fact that $n_0^k = o((1/\eps)^{n_0})$,
if the parameter $\eps$ is sufficiently small as a function of $k.$ This completes the proof.
The following subsection is devoted to the proof of Theorem~\ref{thm:cover-lb-sparse}.

\subsubsection{Proof of Theorem~\ref{thm:cover-lb-sparse}.}
Let $k>2$, $k \in \Z_+$, and $\eps$ be sufficiently small as a function of $k$. 
Let $n = \Theta((1/k) \cdot  \log(1/\eps)/ \log\log(1/\eps))^{k-1}.$ 

We express an $(n, k)$-PMD $X$ as a sum of independent $k$-CRVs 
$X_s$, where $s$ ranges over some index set.
For $1\leq j \leq k-1$, we will denote $p_{s,j} = \Pr[X_s=e_j]$. 
Note that $\Pr[X_s = e_k] = 1- \sum_{j=1}^{k-1} p_{s,j}.$

We construct our lower bound set $\mathcal{S}$ explicitly as follows. 
Let $0< c < 1$ be an appropriately small universal constant. 
We define the integer parameters $a \eqdef \lfloor c \ln(1/\eps)/ 2k \ln\ln(1/\eps)) \rfloor$ and $t \eqdef \lfloor \eps^{-c} \rfloor$.
We define the set $\mathcal{S}$ to have elements indexed by a function
$ f:[a]^{k-1} \rightarrow [t],$
where the function $f$ corresponds to the PMD
$$
X^f \eqdef \sum_{s\in [a]^{k-1}} X^f_s \;,
$$
and the $k$-CRV $X^f_s,$  $s = (s_1, \ldots, s_{k-1}) \in [a]^{k-1}$, has the following parameters:
\begin{equation} \label{eqn:def-p}
p^f_{s,j} = \frac{s_j+\delta_{j,1}\eps^{3c}f(s)}{\ln^k(1/\eps)},
\end{equation}
for $1\le j \le k-1$. {(Note that we use $\delta_{i, j}$ to denote the standard Kronecker  delta function, i.e., $\delta_{i, j} = 1$ if and only if $i=j$).}

{Let $\mathcal{F} = \{f \mid  f:[a]^{k-1} \rightarrow [t]\}$ be the set of all functions from $[a]^{k-1}$ to $[t].$
Then, we have that 
$$ \mathcal{S} \eqdef \{X^f: f \in \mathcal{F} \}.$$
That is, each PMD in $\cal {S}$ is the sum of $a^{k-1}$ many $k$-CRVs, and there are $t$ possibilities
for each $k$-CRV. Therefore,
$$
|\mathcal{S}| = t^{a^{k-1}} = (1/\eps)^{\Omega((1/k) \cdot  \log(1/\eps)/ \log\log(1/\eps))^{k-1}}.
$$
Observe that all PMDs in $\mathcal{S}$ are $k$-maximal.
In particular, for any $f \in \mathcal{F},$ $s \in [a]^{k-1},$ and $1 \le j \le k-1,$ the above definition implies that
\begin{equation} \label{eqn:small-prob}
p^f_{s,j} \leq \frac{1}{k} \cdot \frac{1}{\ln^k(1/\eps)}.
\end{equation}
An important observation, that will be used throughout our proof, is that for each $k$-CRV $X^f_s,$
only the first out of the $k-1$ parameters $p^f_{s,j},$ $1 \leq j \leq k-1,$ depends on the function $f$.
More specifically, the effect of the function $f$ on $p^f_{s,1}$ is a very small perturbation of the numerator.
Note that the first summand in the numerator of (\ref{eqn:def-p})
is a positive integer, while the summand corresponding to $f$ is at most $\eps^{2c} = o(1).$
We emphasize that this perturbation term is an absolutely crucial ingredient of our construction. As will 
become clear from the proof below, this term allows us to show that distinct PMDs in $\cal{S}$
have a {parameter} moment that is substantially different.

}


The proof proceeds in two main conceptual steps that we explain in detail below.

\paragraph{First Step.}
In the first step, we show that for any two distinct PMDs in $\mathcal{S},$
there exists a {parameter} moment in which they differ by a non-trivial amount.
For $m \in \Z^{k-1}_+$, we recall that the $m^{th}$ {parameter} moment of a $k$-maximal PMD $X=\sum_{s \in S}X_s$ is defined to be
$M_m(X) \eqdef \sum_{s\in S} \prod_{j=1}^{k-1} p_{s,j}^{m_j}.$
In Lemma~\ref{lem:moment-separation} below, we show that for any distinct PMDs $X^f, X^g \in \mathcal{S},$
there exists  $m \in [a]^{k-1}$ such that their $m^{th}$ {parameter} moments differ by at least $\poly(\eps).$

\begin{lemma}\label{lem:moment-separation}
If $f,g:[a]^{k-1}\rightarrow [t]$, with $f\neq g$, 
then there exists $m \in [a]^{k-1}$ so that $$|M_{m}(X^f) - M_m(X^g)| \geq \eps^{4c}\;.$$
\end{lemma}

{We now give a brief intuitive overview of the proof.
It is clear that, for $f \neq g$, the PMDs $X^f$ and $X^g$ have distinct parameters.
Indeed, since $f \neq g$, there exists an $s \in [a]^{k-1}$ such that $f(s) \neq g(s)$, which implies that 
the $k$-CRVs  $X^f_s$ and $X^g_s$ have $p^f_{s,1} \neq p^g_{s,1}.$

We start by pointing out that if two arbitrary PMDs have distinct parameters, 
there exists a {parameter} moment where they differ. 
This implication uses the fact that PMDs are determined by their moments,
which can be established by showing that the Jacobian matrix of the moment function is non-singular.
Lemma~\ref{lem:moment-separation} is a a robust version of this fact, that applies to PMDs in $\mathcal{S}$, 
and is proved by crucially exploiting the structure of the set $\cal {S}.$

Our proof of Lemma~\ref{lem:moment-separation} proceeds as follows: 
We start by approximating the {parameter} moments $M_m(X^f)$, $X^f \in \cal{S}$,
from above and below, using the definition of the parameters of $X^f.$ This approximation step allows us to express
the desired difference $M_{m}(X^f) - M_m(X^g)$ (roughly) as the product of two terms: the first term is always positive 
and has magnitude $\poly(\eps),$ while the second term is $L\cdot (f-g),$ for a certain linear transformation (matrix)
$L.$ We show that $L$ is the tensor product of matrices $L_i$, 
where each $L_i$ is a Vandermonde matrix on distinct integers.
Hence, each $L_i$ is invertible, which in turn implies that $L$ is invertible.
Therefore, since $f \neq g$, we deduce that $L\cdot (f-g) \neq \mathbf{0}.$ 
Noting that the elements of this vector are integers, yields the desired lower bound.
}

\begin{proof}[Proof of Lemma~\ref{lem:moment-separation}]
We begin by approximating the $m^{th}$ {parameter} moment of $X^f$. We have that
\begin{align*}
M_m(X^f) & = \sum_{s\in [a]^{k-1}} \prod_{j=1}^{k-1} \left( \frac{s_j+\delta_{j,1}\eps^{3c}f(s)}{\ln^k(1/\eps)}\right)^{m_j} \\
& = \ln^{-k\|m\|_1}(1/\eps) \sum_{s\in [a]^{k-1}} (s_1+\eps^{3c}f(s))^{m_1} \prod_{j=2}^{k-1} s_j^{m_j} \;.
\end{align*}
Note that in the expression $(s_1+\eps^{3c}f(s))^{m_1}=\sum_{i=0}^{m_1} {m_1 \choose i} s_1^{m_1-i} (\eps^{3c}f(s))^i$, 
the ratio of the $(\eps^{3c}f(s))^{i+1}$ term to the $(\eps^{3c}f(s))^i$ term is $(m_1-i) \eps^{3c}f(s)/s_1 i \leq a \eps^{2c} \leq 1/2$. 
So, we have 
\begin{eqnarray*} 
(s_1+\eps^{3c}f(s))^{m_1} & = & \sum_{i=0}^{m_1} {m_1 \choose i} s_1^{m_1-i} (\eps^{3c}f(s))^i \\
&\leq& s_1^{m_1} + m_1 s_1^{m_1-1} \eps^{3c} f(s) + (m_1(m_1-1)/2) s_1^{m_1-2} (\eps^{3c} f(s))^2 \sum_{i=0}^{m_1-2} 2^{-i} \\
&\leq& s_1^{m_1} + m_1 s_1^{m_1-1} \eps^{3c} f(s) + a^a \eps^{4c} \;.
\end{eqnarray*}
We can therefore write
\begin{align*} 
M_m(X^f) &= \ln^{-k\|m\|_1}(1/\eps) \sum_{s\in [a]^{k-1}} \left(s_1+\eps^{3c}f(s)\right)^{m_1} \prod_{j=2}^{k-1} s_j^{m_j}\\
& \leq  \ln^{-k\|m\|_1}(1/\eps) \sum_{s\in [a]^{k-1}} \left(s_1^{m_1} + m_1 s_1^{m_1-1} \eps^{3c} f(s) + a^a \eps^{4c}\right) \left(\prod_{j=2}^{k-1} s_j^{m_j}\right)\\
& \leq \ln^{-k\|m\|_1}(1/\eps) \left( \left(\sum_{s\in [a]^{k-1}} \prod_{j=1}^{k-1}s_j^{m_j}\right) + \left(\eps^{3c}\sum_{s\in [a]^{k-1}}m_1f(s) s_1^{m_1-1}\prod_{j=2}^{k-1}s_j^{m_j}\right) + a^{ka} \eps^{4c}\right) \;.
\end{align*}
Note that
$a^{ka} = \exp{(ak \ln a)} \leq \exp{(ak \ln \ln 1/\eps)} \leq \exp{(c \ln \eps/2)} = (1/\eps)^{c/2} \;,$
and so finally we have
$$M_m(X^f) \leq \ln^{-k\|m\|_1}(1/\eps) \left( \left(\sum_{s\in [a]^{k-1}} \prod_{j=1}^{k-1}s_j^{m_j}\right) 
+ \left(\eps^{3c}\sum_{s\in [a]^{k-1}}m_1f(s) s_1^{m_1-1}\prod_{j=2}^{k-1}s_j^{m_j}\right) + \eps^{7c/2}\right) \;,$$
and that
$$M_m(X^f) \geq \ln^{-k\|m\|_1}(1/\eps) \left( \left(\sum_{s\in [a]^{k-1}} \prod_{j=1}^{k-1}s_j^{m_j}\right) + 
\left(\eps^{3c}\sum_{s\in [a]^{k-1}}m_1f(s) s_1^{m_1-1}\prod_{j=2}^{k-1}s_j^{m_j}\right)\right) \;.$$
An analogous formula holds for the {parameter} moments of $X^g$ and therefore 
$$
M_m(X^f)-M_m(X^g) = \ln^{-k\|m\|_1}(1/\eps)\left(\eps^{3c}\sum_{s\in [a]^{k-1}}m_1 s_1^{m_1-1}\prod_{j=2}^{k-1}s_j^{m_j}(f(s)-g(s)) + O(\eps^{7c/2}) \right).
$$
We need to show that, for at least one value of $m \in \Z^{k-1}_+$, the integer
$$
\sum_{s\in [a]^{k-1}} \prod_{j=1}^{k-1}s_j^{m_j-1}(f(s)-g(s))
$$
is non-zero, since $\log^{-k\|m\|_1}(1/\eps) \cdot \eps^{3c} m_1 \prod_{j=1}^{k-1} s_j > 0$ for all $s$ and $m$.

We observe that these integers are the coordinates of $L(f-g)$, 
where $L:\R^{[a]^{k-1}}\rightarrow \R^{[a]^{k-1}}$ is the linear transformation with
$$
L(h)_m \eqdef \sum_{s\in [a]^{k-1}} \prod_{j=1}^{k-1}s_j^{m_j-1}(h(s)) \;,
$$
for $m \in [a]^{k-1}.$
It should be noted that $L$ is the tensor product of the linear transformations $L_i:\R^m\rightarrow \R^m$, with
$$
L_i(h)_{m_i} = \sum_{s_i \in [a]} s_i^{m_i-1} h(s) \;,
$$
for $1 \leq i \leq k-1$. Moreover, each $L_i(h)$ is given by the Vandermonde matrix 
on the distinct integers $1,2,\ldots,a$, which is non-singular. 
Since each $L_i$ is invertible,  the tensor product $L$ is also invertible. 
Therefore, $L(f-g)$ is non-zero. That is, there exists an $m \in [a]^{k-1}$ with $(L(f-g))_m \neq 0$, 
and so 
$$M_m(X^f)-M_m(X^g) = \log^{-k\|m\|_1}(1/\eps) \cdot \eps^{3c} \cdot m_1 \prod_{j=1}^{k-1} s_j (L(f-g))_m \neq 0 \;.$$
Since $m_1 \prod_{j=1}^{k-1} s_j (L(f-g))_m$ is an integer, $|m_1 \prod_{j=1}^{k-1} s_j (L(f-g))_m| \geq 1$. So, we get
$$|M_m(X^f)-M_m(X^g)| \geq \ln^{-k\|m\|_1}(1/\eps) \eps^{3c} \;.$$
Finally, we note that 
$\ln^{-k\|m\|_1}(1/\eps) = \exp (-k\|m\|_1 \ln \ln 1/\eps) \geq \exp (-ka \ln \ln 1/\eps) \geq \eps^{c/2}$. 
We therefore conclude that $|M_m(X^f)-M_m(X^g)| \geq \eps^{4c}$, as required.
\end{proof}

\paragraph{Second Step.}
In the second  step of the proof,
we show that two PMDs in $\mathcal{S}$ 
that have a {parameter} moment that differs by a non-trivial amount, 
must differ significantly in total variation distance. In particular, we prove:

\begin{lemma}\label{lem:differentMoment-implies-dtv}
Let $f,g:[a]^{k-1}\rightarrow [t]$, with $f\neq g$.
If $|M_m(X^f) - M_m(X^g)| \geq \eps^{4c}$  for some $m \in [a]^{k-1},$
then $\dtv (X^f, X^g) \geq \eps.$
\end{lemma}

{We establish this lemma in two sub-steps: 
We first show that if the $m^{th}$ parameter moments of two PMDs in $\mathcal{S}$ differ by a non-trivial amount, 
then the corresponding probability generating functions (PGF) must differ by a non-trivial amount at a point. 
An intriguing property of our proof of this claim is that it is 
non-constructive: we prove that there exists a point where  the PGF's differ, but we do not explicitly
find such a point. Our non-constructive argument makes essential use of Cauchy's integral formula. 
We are then able to directly translate a distance lower bound between the PGFs to a lower 
bound in total variation distance.
}

For a random variable $W = (W_1, \ldots, W_k) $ taking values in $\Z^k$ and $z=(z_1,\ldots,z_{k-1})\in \C^{k-1}$, 
we recall the definition of the probability generating function:
$P(W,z) \eqdef \E\left[\prod_{i=1}^{k-1} z_i^{W_i}\right].$
For a PMD $X^f$, we have that
$$P(X^f,z)  = \E\left[\prod_{s\in [a]^{k-1}} \prod_{i=1}^{k-1} z_i^{X^f_{s,i}} \right]
= \prod_{s\in [a]^{k-1}} \left(1+\sum_{i=1}^{k-1}p^f_{s,i}({z_i}-1) \right) \;.
$$

We start by establishing the following crucial claim:

\begin{claim} \label{claim:moment-to-pgf}
Let $f,g:[a]^{k-1}\rightarrow [t]$, with $f\neq g$.
If $|M_m(X^f) - M_m(X^g)| \geq \eps^{4c}$  for some $m \in [a]^{k-1},$
then there exists $z^{\ast} \in \C^{k-1}$ with $\|z^{\ast}\|_{\infty} \le 2$ such that 
$|P(X^f,z^{\ast})-P(X^g,z^{\ast})| \geq \eps^{5c}.$
\end{claim}

{

Before we proceed with the formal proof, 
we provide an intuitive explanation of the argument.
The proof of Claim~\ref{claim:moment-to-pgf} proceeds as follows:
We start by expressing $\ln\left(P(X^f, z)\right)$,
the logarithm of the PGF  of a PMD 
$X^f \in \mathcal{S}$, as a Taylor series whose coefficients depend on its {parameter} moments $M_m(X^f)$.
We remark that appropriate bounds on the {parameter} moments of $X^f \in \mathcal{S}$ imply that this series is in fact absolutely convergent 
in an appropriate region $R.$ Note that, using the aforementioned
Taylor expansion, we can express each {parameter} moment $M_m(X^f)$, $m \in \Z^{k-1}_+,$  as a partial derivative
of $\ln\left(P(X^f, z)\right)$. Hence, if $X^f$ and $X^g$ are distinct PMDs in $\mathcal{S},$
the difference $|M_m(X^f) - M_m(X^g)|$ can also be expressed 
as the absolute value of the partial derivative of the difference between the PGFs
$\left|\ln\left(P(X^f, z)\right) - \ln\left(P(X^g, z)\right)\right|.$
We then use Cauchy's integral formula to express this partial derivative
as an integral, which we can further be absolutely bounded from above by the difference 
$\left|\ln\left(P(X^f, z^{\ast})\right) - \ln\left(P(X^g, z^{\ast})\right)\right|, $ for some point $z^{\ast} \in R.$
Finally, we use the fact that $\ln\left(P(X^f, z)\right)$ is absolutely bounded for all $z\in R$ to complete the proof of the claim.
}

\begin{proof}[Proof of Claim~\ref{claim:moment-to-pgf}]
For $z \in \C^{k-1}$, let $w \in \C^{k-1}$ be defined by $w_i = z_i -1$, $1 \leq  i \leq k-1.$
When $|w_i| \leq 5$, for all $1 \leq i \leq k-1$, we take logarithms and obtain:
\begin{align*}
\ln(P(X^f,z)) & = \sum_{s\in [a]^{k-1}} \ln\left(1+\sum_{i=1}^{k-1}p^f_{s,i}w_i \right)\\ 
&= \sum_{s\in [a]^{k-1}} \sum_{\ell=1}^\infty \frac{(-1)^{\ell+1}}{\ell} \left( \sum_{i=1}^{k-1}p^f_{s,i}w_i \right)^\ell \\ 
& = \sum_{s\in [a]^{k-1}} \sum_{\ell=1}^\infty \frac{(-1)^{\ell+1}}{\ell} \sum_{\|m\|_1=\ell} \binom{\ell}{m} \prod_{i=1}^{k-1} (p^f_{s,i})^{m_i}\prod_{i=1}^{k-1} w_i^{m_i}\\ 
& = \sum_{\substack{a\in \Z_{\geq 0}^{k-1} \\a\neq 0}} \frac{(-1)^{\|m\|_1+1}}{\|m\|_1} \binom{\|m\|_1}{m}\left(\prod_{i=1}^{k-1} w_i^{m_i}\right)\sum_{s\in [a]^{k-1}}\left(\prod_{i=1}^{k-1} (p^f_{s,i})^{m_i}\right)\\ 
& = \sum_{\substack{a\in \Z_{\geq 0}^{k-1} \\a\neq 0}} \frac{(-1)^{\|m\|_1+1}}{\|m\|_1} \binom{\|m\|_1}{m}\left(\prod_{i=1}^{k-1} w_i^{m_i}\right)M_m(X^f).
\end{align*}
We note that
$$
|M_m(X^f)| = \left|\sum_{s\in [a]^{k-1}} \prod_{i=1}^{k-1} (p^f_{s,i})^{m_i}\right| \leq  {a}^{k-1} \ln^{-(k-1)\|m\|_1}(1/\eps) < (10k)^{-\|m\|_1} \;,
$$
where the first inequality follows from (\ref{eqn:small-prob}).
Therefore, for $\|z\|_\infty \leq 4$ and so $\|w\|_\infty \leq 5$, we obtain that
\begin{equation}\label{logBoundEqn}
|\ln(P(X^f,z))| \leq \sum_{\ell=1}^\infty \frac{1}{\ell}\sum_{\substack{a\in \Z_{+}^{k-1} \\\|m\|_1=\ell}} \binom{\|m\|_1}{m}\|w\|_\infty^{\ell}(10k)^{-\ell} \leq \sum_{\ell=1}^\infty(k-1)^{\ell}5^\ell(10k)^{-\ell} \leq \sum_{\ell=1}^\infty 2^{-\ell} = 1 \;.
\end{equation}
This suffices to show that all the series above are absolutely convergent when $\|z\|_\infty \leq 4$, 
and thus that their manipulations are valid, and that we get the principal branch of the logarithm.
To summarize, for all $z \in \C^{k-1}$ with $\|z\|_{\infty} \le 4$, and $w \in \C^{k-1}$ with $w_i = z_i - 1$, $i \in [k-1]$, we have:
\begin{equation}\label{logLaplaceEquation}
\ln(P(X^f,z)) = \sum_{\substack{a\in \Z_{\geq 0}^{k-1} \\a\neq 0}} \frac{(-1)^{\|m\|_1+1}}{\|m\|_1} \binom{\|m\|_1}{m}\left(\prod_{i=1}^{k-1} w_i^{m_i}\right)M_m(X^f) \;.
\end{equation}
In particular, we have that the $\prod_{i=1}^{k-1} w_i^{m_i}$ coefficient of $\ln(P(X^f,z))$ 
is an integer multiple of $M_m(X^f)/\|m\|_1$. This expansion is a Taylor series in the $w_i$'s, 
so this coefficient is equal to a partial derivative, 
which we can extract by Cauchy's integral formula. 
Now, suppose that $X^f,X^g$ are distinct elements of $\mathcal{S}$. 
We have that:
\begin{eqnarray*}
|M_m(X^f)-M_m(X^g)|/\|m\|_1 &\leq&  |M_m(X^f)-M_m(X^g)| \binom{\|m\|_1}{m}/\|m\|_1\\
& = & \left(\prod_{i=1}^{k-1} 1/m_i!\right) \left| \frac{\partial^{\|m\|_1} \ln(P(X^f,z)) -\ln(P(X^g,z))}{\partial w_1^{m_1} \cdots \partial w_{k-1}^{m_{k-1}}} \right| \\
 & = &  \left|(1/2 \pi i)^{k-1} \oint_{\gamma} \ldots \oint_{\gamma} (\ln(P(X^f,z)) - \ln(P(X^g,z)))/ \prod_{i=1}^{k-1} w_i^{m_i} dw_1 \ldots dw_{k-1} \right|\\
& \leq &  \max_{|w_1|=1, \ldots, |w_{k-1}|=1}  |\ln(P(X^f,z)) - \ln(P(X^g,z)))/ \prod_{i=1}^{k-1} w_i^{m_i}| \\
&=&  \max_{|w_1|=1, \ldots, |w_{k-1}|=1}  |\ln(P(X^f,z)) - \ln(P(X^g,z)))| \;,
 \end{eqnarray*}
where the third line above follows from Cauchy's integral formula, 
and $\gamma$ is the path round the unit circle. 

Now suppose that there exists an $m \in [a]^{k-1}$, i.e.,  with $\|m\|_1 \leq (k-1)a,$  
such that it holds $|M_m(X^f)-M_m(X^g)| \geq \eps^{4c}$. 
By the above, this implies that there is some $w^{\ast}=(w^{\ast}_1,\ldots, w^{\ast}_{k-1})$ 
with $|w^{\ast}_i|=1$ for all $i$ so that for the corresponding $z^{\ast}$,
\begin{equation}\label{largeGapEqn}
|\ln(P(X^f,z^{\ast}))-\ln(P(X^g,z^{\ast}))| \geq \eps^{4c}/\|m\|_1 \geq \eps^{4c}/(ka) \;.
\end{equation}
{Note that $\|z^{\ast}\|_{\infty} \leq \|w^{\ast}\|_{\infty}+1 = 2.$ Hence, $z^{\ast} \in R.$}
Applying (\ref{logBoundEqn}), at this $z^{\ast}$, 
we have $|\ln(P(X^f,z^{\ast}))| \leq 1$ and $|\ln(P(X^g,z^{\ast}))| \leq 1$.
Therefore, by Equation (\ref{largeGapEqn}), for this $z^{\ast}$ with $\|z^{\ast}\|_{\infty} \leq 2,$ 
we have that
$$
|P(X^f,z^{\ast})-P(X^g,z^{\ast})| = \Omega\left(\eps^{4c}/(ka)\right) \geq \eps^{5c} \;,
$$
where the last inequality follows from our definition of $a.$
This completes the proof of Claim~\ref{claim:moment-to-pgf}.
\end{proof}

We are now ready to translate a lower bound on the distance between the PGFs to 
a lower bound on total variation distance. Namely, we prove the following:

\begin{claim} \label{claim:pgf-to-dtv}
 If there exists $z^{\ast} \in \C^{k-1}$ with $\|z^{\ast}\|_{\infty} \le 2$ such that 
$|P(X^f,z^{\ast})-P(X^g,z^{\ast})| \geq \eps^{5c},$ then $\dtv(X^f, X^g) \geq \eps.$
\end{claim}

{The main idea of the proof of Claim~\ref{claim:pgf-to-dtv} is this: 
By Equation (\ref{logBoundEqn}), we know that $|\ln(P(X^f,z))| \leq 1,$ for all $z \in R.$
We use this fact to show that the contribution to the value of the PGF $P(X^f,z^{\ast})$
coming from the subset of the probability space $\{ X^f > T  \}$
is at most $O(2^{-T})$. On the other hand, the contribution to the difference 
$|\ln(P(X^f,z^{\ast})) - \ln(P(X^g,z^{\ast}))|$
coming from the set $\{X^f \leq T, X^g \leq T \}$ can be easily bounded from above by
 $2^T \cdot \dtv(X^f,X^g)$. The claim follows by selecting an appropriate value of $T = \Theta(\log(1/\eps)),$
 balancing these two terms.
}

\begin{proof}
First note that exponentiating Equation (\ref{logBoundEqn}) 
at $z=(4,4,\ldots,4),$ and using the definition of the PGF we get: 
$$
\E\left[4^{\sum_{i=1}^{k-1} X^f_i} \right],\E\left[4^{\sum_{i=1}^{k-1} X^g_i} \right] \leq e \;.
$$
Therefore, for any $z$ with $\|z\|_\infty \leq 2$ and any $T \in \Z_+$ we have that
$$
\left|\sum_{x, |x|_1\geq T} \prod_{i=1}^{k-1} z_i^{x_i} \Pr[X^f = x]\right| \leq (1/2)^{T}\sum_{x, |x|_1\geq T} \prod_{i=1}^{k-1} 4^{|x|_1} \Pr[X^f = x] \leq e(1/2)^T.
$$
A similar bound holds for $X^g$. By assumption, there exists such a $z^{\ast}$ so that
\begin{align*}
\eps^{5c} & \leq \left|P(X^f,z^{\ast})-P(X^g,z^{\ast})\right|\\
& =  \left|\sum_x  \prod_{i=1}^{k-1} (z^{\ast})_i^{x_i} \left(\Pr[X^f=x] - \Pr[X^g=x]\right) \right|\\
& =  \left|\sum_{|x|_1 < T}\prod_{i=1}^{k-1} (z^{\ast})_i^{x_i} \left(\Pr[X^f=x] - \Pr[X^g=x]\right) \right| + 2e(1/2)^T\\
& \leq  2^T \sum_{|x|_1 < T} \left|\Pr[X^f=x]-\Pr[X^g=x]\right|+2e (1/2)^T\\
& \leq  2^T \dtv(X^f,X^g) + 2e/2^{T} \;.
\end{align*}
Taking $T = \lceil 5 c \log_2 (1/\eps) \rceil$, we get
$$
\dtv(X^f,X^g) \geq \Omega(\eps^{10c}) \geq \eps.
$$
This completes the proof Claim~\ref{claim:pgf-to-dtv}.
\end{proof}

Lemma~\ref{lem:differentMoment-implies-dtv} follows by combining
Claims~\ref{claim:moment-to-pgf} and~\ref{claim:pgf-to-dtv}.
By putting together Lemmas~\ref{lem:moment-separation} and~\ref{lem:differentMoment-implies-dtv}, 
it follows that any two distinct elements of $\mathcal{S}$ are $\eps$-separated in total variation distance.
This completes the proof of Theorem~\ref{thm:cover-lb-sparse}, 
establishing the correctness of our lower bound construction. \qed

\section{A Size--Free Central Limit Theorem for PMDs} \label{sec:clt}

In this section, we prove our new CLT thereby establishing Theorem~\ref{thm:clt}.
For the purposes of this section, we define a discrete Gaussian in $k$ dimensions to be a probability distribution 
supported on $\Z^k$ so that the probability of a point $x$ is proportional to $e^{Q(x)},$ for some quadratic polynomial $Q.$
The formal statement of our CLT is the following:

\begin{theorem}\label{thm:clt-formal}
Let $X$ be an $(n,k)$-PMD with covariance matrix $\Sigma.$
Suppose that $\Sigma$ has no eigenvectors other than $\bone=(1,1,\ldots,1)$
with eigenvalue less than $\sigma.$
Then, there exists a discrete Gaussian $G$ so that
$$
\dtv(X,G) \leq O( k^{7/2} \sqrt{\log^3(\sigma)/\sigma}).
$$
\end{theorem}

We note that our phrasing of the theorem above is slightly different than 
the CLT statement of~\cite{VV10b}. 
More specifically, we work with $(n,k)$-PMDs directly, while~\cite{VV10b} work 
with projections of PMDs onto $k-1$ coordinates. 
Also, our notion of a discrete Gaussian is not the same as the one discussed in~\cite{VV10b}. 
At the end of the section, we show how our statement can be rephrased 
to be directly comparable to the~\cite{VV10b} statement.

\paragraph{Proof of Theorem~\ref{thm:clt-formal}.}

We note that unless $\sigma > k^{7}$ that there is nothing to prove, and thus we will assume this throughout the rest of the proof.

The basic idea of the proof will be to compare the Fourier transform of $X$ to that of the discrete Gaussian
with density proportional to the pdf of $\mathcal{N}(\mu, \Sigma)$ (where $\mu$ is the expectation of $X$).
By taking the inverse Fourier transform, we will be able to conclude that these distributions are pointwise close.
A careful analysis of this combined with the claim that both $X$ and $G$ have small effective support will yield our result.

We start by providing a summary of the main steps of the proof.
We start by bounding the effective support of $X$ under our assumptions 
(Lemma~\ref{lem:clt-support} and Corollary \ref{cor:clt-support-ellipse-and-count}). 
Then, we describe the effective support of its Fourier transform (Lemma  \ref{lem:outside-effective-support-ft}). 
We further show that the effective support of the distribution $X$ 
and the Fourier transform of the discrete Gaussian $G$ are 
similar (see Lemmas \ref{lem:clt-gaussian-support} and \ref{lem:clt-gaussian-ft-support}).
We then obtain an estimate of the error between the Fourier transforms of $X$ and a Gaussian with the same mean and covariance (Lemma \ref{FourierApproxLem}). 
The difference between the distributions of $X$ and $G$ at a point, as given by the inverse Fourier transform, 
is approximately equal to the integral of this error over the effective support of the Fourier transform of $X$ and $G.$ 
If we take bounds on the size of this integral naively, we get a weaker result than Theorem~\ref{thm:clt-formal}, 
concretely that  $\dtv(X,G) \leq O(\log(\sigma))^k \sigma^{-1/2}$ (Proposition \ref{prop:clt-naive-integral}). 
Finally, we are able to show the necessary bound on this integral by using the saddlepoint method.

We already have a bound on the effective support of a general PMD (Lemma \ref{lem:conc-pmd}). 
Using this lemma, we obtain simpler bounds that hold under our assumptions.
\begin{lemma} \label{lem:clt-support}
Let $X$ be an $(n,k)$-PMD with mean $\mu$ and covariance matrix $\Sigma,$
where all non-trivial eigenvalues of $\Sigma$ are at least $\sigma,$
then for any {$\eps > \exp(-\sigma/k)$}, with probability $1-\eps$ over $X$ we have that
$$
(X-\mu)^T (\Sigma+I)^{-1}(X-\mu) = {O(k\log(k/\eps))}.
$$
\end{lemma}
\begin{proof}
From Lemma \ref{lem:conc-pmd}, 
we have that $(X-\mu)^T (k\ln(k/\eps)\Sigma + k^2 \ln^2(k/\eps)I)^{-1}(X-\mu) = O(1)$ with probability at least $1-\eps/10.$

By our assumptions on $\Sigma,$ we can write $\Sigma=U^T \mathrm{diag}(\lambda_i) U,$
for an orthogonal matrix $U$ with $k^{th}$ column $\bone/\sqrt{k}$ and
$\lambda_i \geq \sigma$ for $1 \leq i \leq k-1,$ and $\lambda_k=0.$

By our assumptions on $\eps,$ we have that $\sigma \geq k \ln (1/\eps),$
and so for $1 \leq i \leq k-1$, we have $\lambda_i + 1 \geq \frac{1}{2}(\lambda_i + k \ln (1/\eps)).$

Since $ \bone^T (X-\mu)= 0,$ we have that $(U^T (X-\mu))_k=0,$ and so we can write
\begin{eqnarray*}
(X-\mu)\cdot (\Sigma+I)^{-1}(X-\mu) & = &  (U^T (X-\mu))^T \mathrm{diag}(1/(\lambda_i+1)) U^T (X-\mu) \\
& \leq &  (U^T (X-\mu))^T \mathrm{diag}(2/(\lambda_i+k \ln (1/\eps))) U^T (X-\mu) \\
& = & 2k \ln (k/\eps) \cdot (X-\mu)^T \left(k\ln(k/\eps)\Sigma + k^2 \ln^2(k/\eps)I\right)^{-1} (X-\mu) \\
& = & O(k \ln (k/\eps)).
\end{eqnarray*}
\end{proof}
Specifically, if we take $\eps=1/\sigma$, we have the following:
\begin{corollary} \label{cor:clt-support-ellipse-and-count}
Let $X$ be as above, and let $S$ be the set of points $x \in \Z^k$ where $(x-\mu)^T \bone=0$ and
$$
(x-\mu)^T (\Sigma+I)^{-1}(x-\mu) \leq (Ck\log(\sigma)) \;,
$$
for some sufficiently large constant $C.$ Then, $X\in S$ with probability at least $1-1/\sigma,$ and
$$
|S| = \sqrt{\det(\Sigma+I)} \cdot O(\log(\sigma))^{k/2}.
$$
\end{corollary}
\begin{proof}
Noting that $\ln (k \sigma)=O(\log \sigma)$ since $\sigma > k$, 
by Lemma \ref{lem:clt-support}, applied with $\eps=1/\sigma$, 
it follows that $x \in S$ with probability $1-1/\sigma.$

The remainder of the claim is a standard counting argument 
where we need to bound the number of integer lattice points 
within a continuous region (in this case, an ellipsoid). 
We deal with this by way of the standard technique 
of erecting a unit cube about each of the lattice points 
and bounding the volume of the union. 
Note that for $x \in \Z^k,$ the cubes $x + (-1/2,1/2)^k$ each have volume one and are disjoint. 
Thus, if we define $S'$ to be the set of $y$ such that there exists an $x \in S$ 
with $\|y-x\|_\infty < \frac{1}{2},$ then $|S| = \mathrm{Vol}(S').$ 
For any $y \in S',$ there is an $x \in S$ with:
\begin{eqnarray*}
(y-\mu) \cdot  (\Sigma+I)^{-1}(y-\mu) 
& = & (x-\mu) \cdot  (\Sigma+I)^{-1}(x-\mu) + (y-x) \cdot  (\Sigma+I)^{-1}(y-x) + \\
&&2  (y-x) \cdot  (\Sigma+I)^{-1}(x-\mu) \\
& \leq & O((x-\mu) \cdot  (\Sigma+I)^{-1}(x-\mu) + (y-x) \cdot  (\Sigma+I)^{-1}(y-x) ) \\
& \leq & O(Ck\log(\sigma) + (y-x) \cdot  I (y-x)) \\
& \leq & O(Ck\log(\sigma) + k \|y-x\|_\infty^2) \\
& \leq & O(Ck\log(\sigma) + k) \\
&=&  O(Ck\log(\sigma)) \;.
\end{eqnarray*}
That is, $S'$ is contained in the ellipsoid $(y-\mu) \cdot  (\Sigma+I)^{-1}(y-\mu) \leq O(Ck\log(\sigma)).$
The corollary follows by bounding the volume of this ellipsoid. We have the following simple claim:
\begin{claim} \label{clm:ellipse-volume-bound}
The volume of the ellipsoid $x^T A^{-1} x \leq ck$ for a symmetric $k \times k$ matrix $A$
and $c >0$ is $\sqrt{\det (A)} \cdot O(c)^{k/2}.$
\end{claim}
\begin{proof}
We can factorize $A= U^T \mathrm{diag}(\lambda_i) U^T$ for some orthogonal matrix $U.$
Then, the ellipsoid is the set of $x$ with $\| \mathrm{diag}(1/\sqrt{ck \lambda_i}) U^T x \|_2 \leq 1.$
The volume of the ellipsoid is  
$$\left|\det( \mathrm{diag}(1/\sqrt{ck \lambda_i}) U^T)^{-1} \right| V_k = \sqrt{\det (A)} \cdot (ck)^{k/2} V_k,$$
where $V_k$ is the volume of the unit sphere.
By standard results, $V_k = \pi^{k/2}/\Gamma(1+k/2)={\Omega(k)^{-k/2}},$
using Stirling's approximation $$\Gamma(1+k/2)= \sqrt{2 \pi / (1+k/2)} ((1+k/2)/e)^{1+k/2} (1+O(2/(k+2)).$$
Therefore, the volume is $O(\sqrt{\det (A)} \cdot (c)^{k/2}).$
\end{proof}
As a consequence, the volume of the ellipsoid $(y-\mu) \cdot  (\Sigma+I)^{-1}(y-\mu) \leq O(Ck\log(\sigma))$
is $\sqrt{\det (\Sigma + I)} \cdot O(C\log \sigma)^{k/2}.$
Thus, we conclude that $|S| \leq \mathrm{Vol}(S') \leq \sqrt{\det (\Sigma + I)} \cdot O(\log \sigma)^{k/2}.$
This completes the proof of the corollary.
\end{proof}

Next, we proceed to describe the Fourier support of $X.$
In particular, we show that $\wh{X}$ has a relatively small effective support, $T$.
Our Fourier sparsity lemma in this section is somewhat different than in previous section, 
but the ideas are similar.
The proof will similarly need Lemma \ref{lem:gaussian-bound-from-interval}.

\begin{lemma} \label{lem:outside-effective-support-ft}
Let $T \eqdef \{ \xi\in \R^k \mid {\xi}\cdot \Sigma {\xi} \leq C k \log(\sigma) \},$ for $C$ some sufficiently large constant. 
Then, we have that:
\begin{itemize}
\item[(i)] For all $\xi\in T,$ the entries of $\xi$ are contained in an interval of length $2\sqrt{Ck\log(\sigma)/\sigma}.$
\item[(ii)] Letting $T' = T \cap \{\xi \in \R^k \mid \xi_1\in[0,1]\},$ 
it holds $\mathrm{Vol}(T') = \det(\Sigma+I)^{-1/2} \cdot O(C \log(\sigma))^{k/2}.$
\item[(iii)] $\int_{[0,1]^k \setminus (T+\Z^k)} |\wh{X}(\xi)| d\xi \leq 1/(\sigma |S|).$
\end{itemize}
\end{lemma}
\begin{proof}
We define $\tilde{\xi}$ to be the projection of $\xi$ onto the plane where the coordinates sum to $0,$
i.e., $\tilde{\xi} = \xi + \alpha \bone$ for some $\alpha \in \R$ and $\tilde{\xi} \cdot \bone=0.$
Then, we have that $\xi\cdot \Sigma \xi \geq \sigma |\tilde{\xi}|^2_2.$
Hence, for $\xi \in T,$ we have that $|\tilde{\xi}|_\infty \leq |\tilde{\xi}|_2 \leq \sqrt{Ck\log(\sigma)/\sigma}.$
This implies that for any $i, j$ it holds 
$$|\xi_i - \xi_j| = |\tilde{\xi}_i - \tilde{\xi}_j| \leq 2|\tilde{\xi}|_\infty \leq 2\sqrt{Ck\log(\sigma)/\sigma}.$$
This proves (i).

In particular, for any $\xi \in T,$ and all $i,$  we have $|\xi_1 - \xi_i| \leq 2\sqrt{Ck\log(\sigma)/\sigma} \leq 2C.$
And so if $\xi \in T'$, then $\xi_i \in [-2\sqrt{C},1+ 2\sqrt{C}].$
Thus, for $\xi \in T'$, it holds $\xi \cdot \xi \leq O(C).$
Thus, for $\xi  \in T',$ we have $\xi \cdot (\Sigma+I) \cdot \xi \leq O(C k \log(\sigma)).$
By Claim~\ref{clm:ellipse-volume-bound}, we get that $\mathrm{Vol}(T') \leq \det(\Sigma+I)^{-1/2} O(C \log(\sigma))^{k/2}.$
This proves (ii).

By Claim \ref{clm:pigeon-hole-interval},
for every $\xi \in [0,1)^k,$
there is an interval $I_\xi$ of length $1-1/(k+1)$
such that $\xi'=\xi+b$, for some $b \in \Z^k,$ has coordinates in $I_\xi.$
Let $T_m$ be the set of $\xi$ such that there is such a $\xi'$ with
$$2^{m+1} C k \log(\sigma) \geq \xi'\cdot \Sigma \xi' \geq 2^m C k \log(\sigma) \;,$$
and $\xi_i = \xi'_i - \lfloor \xi'_i \rfloor$ for all $1 \leq i \leq k.$
Then, for every $\xi \in [0,1]^k,$ we either have $\xi \in T + \Z^k$
or else $\xi \in T_m$ for some $m \geq 0.$
Hence,
\begin{equation} \label{eq:power-2-clt}
\int_{[0,1]^k \setminus (T + \Z^k)}|\wh{X}(\xi)|d\xi \leq  \sum_{m=0}^\infty \mathrm{Vol}(T_m)\sup_{\xi\in T_m}|\wh{X}(\xi)| \;.
\end{equation}

To bound the RHS above, we need bounds on the volume of each $T_m.$ 
These can be obtained using a similar argument to (ii) along with some translation.
\begin{claim} We have that
$\mathrm{Vol}(T_m) \leq \det(\Sigma+I)^{-1/2} \cdot O(2^{m+1} C \log(\sigma))^{k/2}.$
\end{claim}
\begin{proof}
Let $U_m$ be the set of $\xi$ such that there is a $\xi'$
with $2^{m+1} C k \log(\sigma) \geq \xi'\cdot \Sigma \xi'$
and $\xi_i = \xi'_i - \lfloor \xi'_i \rfloor$ for all $1 \leq i \leq k.$
Note that $T_m \subseteq U_m.$
Let $U'_m$ be the set of $\xi'$ with $\xi_1 \in [0,1]$ and $2^{m+1} C k \log(\sigma) \geq \xi'\cdot \Sigma \xi'.$
Note that for any $\xi''$ with $2^{m+1} C k \log(\sigma) \geq \xi''\cdot \Sigma \xi''$,
we have that $\xi''+\lambda \bone$ also satisfies $2^{m+1} C k \log(\sigma) \geq (\xi'' + \lambda \bone) \cdot \Sigma (\xi''+\lambda \bone)$ for any $\lambda \in \R.$
In particular $\xi'=\xi'' - (\lfloor \xi''_1 \rfloor) \bone) \in U'_m$. Note that $\xi''_i - \lfloor \xi''_i \rfloor = \xi'_i$ for all $i$.
So $U_m$ is the set of $\xi$ such that there is a $\xi' \in U'_m$ with $\xi_i = \xi'_i - \lfloor \xi'_i \rfloor$.
Then, $\mathrm{Vol}(U_m) \le \mathrm{Vol}(U'_m)$ 
since $U_m = \cup_{b \in \Z^n} (U'_m \cap \prod_{i=1}^k [b_i,b_i+1))-b,$ 
and so $\mathrm{Vol}(U_m) \leq \sum_{b \in \Z^n} \mathrm{Vol}(U'_m \cap \prod_{i=1}^k [b_i,b_i+1))) = \mathrm{Vol}(U'_m).$

Note that by Lemma \ref{lem:outside-effective-support-ft} (ii) applied with $C := 2^{m+1}C$ gives the bound
$$\mathrm{Vol}(U'_m) \leq \det(\Sigma+I)^{-1/2} \cdot O(2^{m+1} C \log(\sigma))^{k/2}.$$
Therefore, we have $\mathrm{Vol}(T_m) \leq  \mathrm{Vol}(U_m) \leq  \mathrm{Vol}(U'_m) \leq \det(\Sigma+I)^{-1/2} \cdot O(2^{m+1} C \log(\sigma))^{k/2}.$
This completes the proof.
\end{proof}


Next, we obtain bounds on $\sup_{\xi\in T_m}|\wh{X}(\xi)|$ by using Lemma \ref{lem:gaussian-bound-from-interval}.
\begin{claim} 
For $\xi \in T_m,$ it holds $|\wh{X}(\xi)| \leq \exp(-\Omega(C 2^m \log(\sigma)/k)).$
If additionally we have $m \leq 4 \log_2 k,$ then $|\wh{X}(\xi)| = \exp(-\Omega(C 2^m k \log(\sigma)))$. 
\end{claim}
\begin{proof}
Note that $\xi'$ has coordinates in an interval of length $1 - 1/k,$ 
so we may apply Lemma \ref{lem:gaussian-bound-from-interval}, 
yielding 
$$|\wh{X}(\xi)|=|\wh{X}(\xi')| \leq  \exp(-\Omega(\xi'^T \cdot \Sigma \cdot \xi'/k^2)) = \exp(-\Omega(C 2^m \log(\sigma)/k)).$$
To get the stronger bound, we need to show that for small $m,$ all the coordinates of $\xi'$ are in a shorter interval.
As before, we consider $\tilde{\xi'},$ the projection of $\xi'$ onto the set of $x$ with $x \cdot \bone=0.$
Similarly, we have $\xi' \cdot \Sigma \xi' \geq \sigma |\tilde{\xi'}|^2_2.$
So, for any $i,j$, it holds 
$$|\xi'_i-\xi'_j| \leq |\tilde{\xi'}_i - \tilde{\xi'}_j| \leq 2 |\tilde{\xi'}|_\infty \leq 2 |\tilde{\xi'}|_2 \leq \sqrt{\xi' \cdot \Sigma \xi'/\sigma} \leq \sqrt{C 2^{m+1} k \log(\sigma)/\sigma}.$$ 
For $m \leq \log_2( \sigma/Ck \log(\sigma))-3$, we have that the coordinates of $\xi$ lie in an interval of length $1/2.$
Now, Lemma~\ref{lem:gaussian-bound-from-interval} gives that
$$|\wh{X}(\xi)|= |\wh{X}(\xi')| \leq  \exp(-\Omega(\xi'^T \cdot \Sigma \cdot \xi')) = \exp(-\Omega(C k 2^m \log(\sigma)/k)).$$
Finally, note that $4 \log_2 k \leq \log_2( \sigma/Ck \log(\sigma))-3,$ when $\sigma \geq Ck^3.$
This completes the proof of the claim.
\end{proof}
Using the above, we can write
\begin{eqnarray*}
 \int_{[0,1]^k \setminus (T + \Z^k)} |\wh{X}(\xi)|d\xi & \leq &  \sum_{m=0}^\infty \mathrm{Vol}(T_m)\sup_{\xi\in T_m}|\wh{X}(\xi)| \\
									& \leq & \det(\Sigma+I)^{-1/2} \cdot O(C \log(\sigma))^{k/2} \sum_{m=0}^\infty 2^{mk/2} \sup_{\xi\in T_m}|\wh{X}(\xi)| \;.
\end{eqnarray*}
We divide this sum into two pieces:
\begin{eqnarray*} 
\sum_{m=0}^{4 \log_2 k} 2^{mk/2} \sup_{\xi\in T_m}|\wh{X}(\xi)| & \leq & \sum_{m=0}^{\log_2( \sigma/Ck \log(\sigma))-3} 2^{mk/2} \exp(-\Omega(C 2^m k \log(\sigma))) \\
& \leq & \sum_{m=0}^{4 \log_2 k} \exp(-\Omega(C (2^m-m) k \log(\sigma))) \\
& \leq & \sum_{m=0}^{4 \log_2 k} 2^{-m} \exp(-\Omega(C k \log(\sigma))) \\
& \leq & \exp(-\Omega(C k \log(\sigma)))  =  \sigma^{-\Omega(Ck)} \;,
\end{eqnarray*}
and
\begin{align*} 
\sum_{m=4 \log_2 k}^{\infty} 2^{mk/2} \sup_{\xi\in T_m}|\wh{X}(\xi)|  \leq & \sum_{m=4 \log_2 k}^{\infty}  2^{mk/2} \exp(-\Omega(C 2^m  \log(\sigma)/k)) \\
 \leq & \sum_{m=4 \log_2 k}^{\infty}  \exp(-\Omega(C (2^m-m) \log(\sigma)/k)) \\
 \leq & \sum_{m=4 \log_2 k}^{\infty}  \exp(-\Omega(C (k^2+m) \log(\sigma)/k)) \\
 \leq & \sum_{m=4 \log_2 k}^{\infty}  2^{-m} \exp(-\Omega(C k \log(\sigma)))
 \leq  \sigma^{-\Omega(Ck)} \;.
\end{align*}
We thus have $\int_{[0,1]^k \setminus (T + \Z^k)} |\wh{X}(\xi)|d\xi \leq  \det(\Sigma+I)^{-1/2}O(C \log(\sigma))^{k/2} \log(\sigma)^{-O(Ck)}.$
\end{proof}

The previous lemma establishes that the contribution to the Fourier transform of $X$ coming from points outside of $T$ is negligibly small.
We next claim that, for $\xi\in T,$ it is approximated by a Gaussian.
\begin{lemma}\label{FourierApproxLem}
For $\xi\in T,$ we have that
$$
\wh{X}(\xi) = \exp\left( 2\pi i \mu\cdot \xi - 2\pi^2 \xi\cdot \Sigma \xi + O(C^{3/2} k^{7/2} \sqrt{\log^3(\sigma)/\sigma})\right).
$$
This also holds for complex $\xi,$ under the assumption that 
the coordinate-wise complex and real parts of $\xi$ are in $T,$ 
i.e., such that $\re(\xi)\cdot \Sigma\re(\xi),\im(\xi)\cdot \Sigma\im(\xi) \leq O(Ck\log(\sigma)).$
\end{lemma}
\begin{proof}
Recall that
$
\wh{X}(\xi) = \prod_{i=1}^n \sum_{j=1}^k e(\xi_j) p_{ij}.
$
Let $m_i$ be the element of $[k]$ so that $p_{im_i}$ is as large as possible for each $i$.
In particular, $p_{im_i}\geq 1/k.$
We will attempt to approximate the above product by approximating
the log of $\sum_{j=1}^k e(\xi_j)p_{ij}$ by its Taylor series expansion
around the point $(\xi_{m_i},\xi_{m_i},\ldots,\xi_{m_i}).$
In particular, by Taylor's Theorem, we find that
$$
\sum_{j=1}^k e(\xi_j)p_{ij} = \exp\left(2\pi i \left( \xi_{m_i} +\sum_{j=1}^k p_{ij}(\xi_j-\xi_{m_i})\right) -2\pi^2 \left(\sum_{j=1}^k p_{ij}(\xi_j-\xi_{m_i})^2 \right)+2\pi^2\left( \sum_{j=1}^k p_{ij}(\xi_j-\xi_{m_i})\right)^2 +E_i \right),
$$
where $E_i$ is the third directional derivative in the $\xi-(\xi_{m_i},\ldots,\xi_{m_i})$ direction of $\log(\wh{X_i}(\xi))$
at some point $\tilde{\xi}$ along the line between $\xi$ and $(\xi_{m_i},\ldots,\xi_{m_i}).$ Note that the above is exactly
$$
\exp\left(2\pi i (\xi\cdot \E[X_i]) - 2\pi^2 (\xi\cdot \cov(X_i)\xi) + E_i \right).
$$
Thus, taking a product over $i$, we find that
$$
\wh{X}(\xi) = \exp\left( 2\pi i \mu\cdot \xi - 2\pi^2 \xi\cdot \Sigma \xi +\sum_{i=1}^n E_i\right).
$$

We remark that the coefficients of this Taylor series are (up to powers of $-2\pi i$) the cumulants of $X$.

Since the coordinates of $\tilde{\xi}$ lie in an interval of length at most $1/2,$
we have that $\sum_{j=1}^k e(\tilde{\xi}_j) p_{ij}$ is bounded away from $0.$ Therefore, we get that
\begin{align*}
|E_i| = O\bigg(\sum_{j=1}^k p_{ij} |\tilde\xi_j-\xi_{m_i}|^3 & + \sum_{j_1,j_2=1}^k p_{ij_1}p_{ij_2}|\tilde\xi_{j_1}-\xi_{m_i}|^2|\tilde\xi_{j_2}-\xi_{m_i}|\\ & + \sum_{j_1,j_2,j_3=1}^k p_{ij_1}p_{ij_2}p_{ij_3}|\tilde\xi_{j_1}-\xi_{m_i}||\tilde\xi_{j_2}-\xi_{m_i}||\tilde\xi_{j_3}-\xi_{m_i}|\bigg).
\end{align*}
Next note that
$$
\var(X_i\cdot \tilde{\xi}) \geq p_{ij}p_{im_i} |\tilde\xi_j-\xi_{m_i}|^2 \geq p_{ij} |\tilde\xi_j-\xi_{m_i}|^2/k.
$$
Additionally, note that
$$
C k \log(\sigma) \geq \xi\cdot\Sigma \xi  = \var(X\cdot \xi) = \sum_{i=1}^n \var(X_i \cdot \xi).
$$
Therefore,
$$
\sum_{i=1}^n p_{ij} |\tilde\xi_j-\xi_{m_i}|^2 \leq \sum_{i=1}^n p_{ij} |\xi_j-\xi_{m_i}|^2 \leq C k^2 \log(\sigma).
$$

Thus, since $|\tilde\xi_j-\xi_{m_i}|=O(\sqrt{Ck\log(\sigma)/\sigma})$ for all $i,j$ we have that
$$
\sum_{i=1}^n \sum_{j=1}^k p_{ij} |\tilde\xi_j-\xi_{m_i}|^3 \leq O(C^{3/2} k^{7/2} \sqrt{\log^3(\sigma)/\sigma}).
$$
We have that
\begin{align*}
\sum_{i=1}^n \sum_{j_1,j_2=1}^k p_{ij_1}p_{ij_2}|\tilde\xi_{j_1}-\xi_{m_i}||\tilde\xi_{j_2}-\xi_{m_i}| & \leq \sum_{i=1}^n \left(\sum_{j=1}^k p_{i,j}|\tilde \xi_j -\xi_{m_i}|\right)^2\\
& \leq \sum_{i=1}^n \left(\sum_{j=1}^k p_{i,j} \right) \left(\sum_{j=1}^k p_{i,j}|\tilde \xi_j -\xi_{m_i}|^2 \right)\\
& = O(Ck^3\log(\sigma)).
\end{align*}
Therefore,
$$
\sum_{i=1}^n \sum_{j_1,j_2=1}^k p_{ij_1}p_{ij_2}|\tilde\xi_{j_1}-\xi_{m_i}|^2|\tilde\xi_{j_2}-\xi_{m_i}|
$$
and
$$
\sum_{j_1,j_2,j_3=1}^k p_{ij_1}p_{ij_2}p_{ij_3}|\tilde\xi_{j_1}-\xi_{m_i}||\tilde\xi_{j_2}-\xi_{m_i}||\tilde\xi_{j_3}-\xi_{m_i}|
$$
are both
$$
O(C^{3/2}k^{7/2}\sqrt{\log^3(\sigma)/\sigma}).
$$
\end{proof}

We now define $G$ to be the discrete Gaussian supported on the set of points in $\Z^k$ whose coordinates sum to $n,$ 
so that for such a point $p$ we have:
\begin{align*}
G(p) = (2\pi)^{-(k-1)/2}\det(\Sigma')^{-1/2}\exp((p-\mu)\cdot \Sigma^{-1} (p-\mu)/2) & = \int_{\xi,\sum \xi_j=0} e(-p\cdot \xi) \exp(2\pi i (\xi\cdot \mu) - 2\pi^2 \xi\cdot \Sigma \xi)\\
& = \int_{\xi,\xi_1\in[0,1]} e(-p\cdot \xi) \exp(2\pi i (\xi\cdot \mu) - 2\pi^2 \xi\cdot \Sigma \xi) \;,
\end{align*}
where $\Sigma'=\Sigma+\bone \bone^T$ restricted to the space of vectors whose coordinates sum to $0.$

We let $\wh{G}$ equal
$$
\wh{G}(\xi) := \exp(2\pi i (\xi\cdot \mu) - 2\pi^2 \xi\cdot \Sigma \xi).
$$
Next, we claim that $G$ and $X$ have similar effective supports and subsequently that $\wh{G}$ and $\wh{X}$ do as well.
Firstly, the effective support of the distribution of $G$ is similar to that of $X$, namely $S$:
\begin{lemma} \label{lem:clt-gaussian-support}
The sum of the absolute values of $G$ at points not is $S$ is at most $1/\sigma.$
\end{lemma}
\begin{proof}
For this it suffices to prove a tail bound for $G$ analogous to that satisfied by $X.$
In particular, assuming that $\Sigma$ has unit eigenvectors $v_i$ with eigenvalues $\lambda_i,$
it suffices to prove that $|(G-\mu)\cdot v_i| < \sqrt{\lambda_i} t$ except with probability at most $\exp(-\Omega(t^2)).$
Recall that
$$
G(p) = (2\pi)^{-(k-1)/2}\det(\Sigma')^{-1/2}\exp((p-\mu)\cdot \Sigma^{-1} (p-\mu)/2).
$$
Let $\tilde{G}$ be the continuous probability density defined by
$$
\tilde G(x) = (2\pi)^{-k/2}\det(\Sigma')^{-1/2}\exp((x-\mu)\cdot \Sigma'^{-1} (x-\mu)/2).
$$
Note that for any $p$ with $(p-\mu)\cdot \bone = 0$, and $x\in [-1/2,1/2]^k,$ we have that
$$
G(p) = O\left(\tilde G(p+x) + \tilde G(p-x) \right).
$$
Therefore, we have that
$$
G(p) = O\left(\int_{x\in p+ [-1/2,1/2]^k} \tilde G(x) \right).
$$
Applying this formula for each $p$ with $(p-\mu)\cdot v_i \geq \sqrt{\lambda_i}t$ 
and noting that $(x-\mu)\cdot v_i \geq (p-\mu)\cdot v_i - \sqrt{k} \geq \sqrt{\lambda_i}t-\sqrt{k}$ yields
$$
\Pr(|(G-\mu)\cdot v_i| > \sqrt{\lambda_i} t) = O(\Pr(|(\tilde G-\mu)\cdot v_i| > \sqrt{\lambda_i} t - \sqrt{k} = \exp(-\Omega(t^2)).
$$
Taking a union bound over $1\leq i \leq k$ yields our result.
\end{proof}
Secondly, the effective support of the Fourier Transform of $G$ is similar to that of $X$, namely $T$:
\begin{lemma} \label{lem:clt-gaussian-ft-support}
The integral of $|\wh{G}(\xi)|$ over $\xi$ with $\xi_1\in [0,1]$ and $\xi$ not in $T$ is at most $1/(|S|\sigma).$
\end{lemma}
\begin{proof}
We consider the integral over $\xi\in T_m,$ where
$$
T_m := \{\xi: \xi_1\in[0,1] \mid \xi\cdot \Sigma\xi \in [2^m C k \log(\sigma),2^{m+1} C k \log(\sigma)] \}.
$$
We note that it has volume $2^{mk}k^{O(k)}\log^{O(k)}(\sigma)/|S|,$
and that within $T_m$ it holds $|\wh{G}(\xi)| = \exp(-\Omega(Ck\log(\sigma)2^m)).$
From this it is easy to see that the integral over $T_m$ is at most $2^{-m-1}/(|S|\sigma).$
Summing over $m$ yields the result.
\end{proof}

We now have all that is necessary to prove a weaker version of our main result.
\begin{proposition} \label{prop:clt-naive-integral}
We have the following:
$$
\dtv(X,G) \leq O(\log(\sigma))^k \sigma^{-1/2}.
$$
\end{proposition}
\begin{proof}
First, we bound the $L^\infty$ of the difference.
In particular, we note that for any $p$ with integer coordinates summing to $n$ we have that
$$
X(p) = \int_{\xi\in\R^k,\xi_1\in[0,1],\xi_i\in [\xi_1-1/2,\xi_1+1/2]} e(-p\cdot \xi)\wh{X}(\xi)d\xi \;,
$$
and
$$
G(p) = \int_{\xi,\xi_1\in[0,1]} e(-p\cdot \xi) \exp(2\pi i (\xi\cdot \mu) - 2\pi^2 \xi\cdot \Sigma \xi).
$$
We note that in both cases the integral for $\xi$ not in $T$ is at most $1/(|S|\sigma).$
To show this, we need to note that any $\xi$ with $\xi_1\in[0,1],\xi_i\in [\xi_1-1/2,\xi_1+1/2]$
equivalent to a point of $T$ modulo $\Z^k$ must lie in $T$ itself.
This is because the element of $T$ must have all its coordinates differing by at most $1/4,$
and thus must differ from $\xi$ by an integer multiple of $(1,1,\ldots,1).$ Therefore, we have that
\begin{align*}
|X(p)-G(p)| & = \left| \int_{\xi\in T,\xi_1\in[0,1]}e(-p\cdot \xi) (\wh{X}(\xi)-\wh{G}(\xi))d\xi \right|+O(1/(|S|\sigma))\\
& \leq \int_{\xi\in T,\xi_1\in[0,1]}|\wh{X}(\xi)-\wh{G}(\xi)|d\xi+O(1/(|S|\sigma))\\
& \leq \int_{\xi\in T,\xi_1\in[0,1]}O(C^{3/2} k^{7/2} \sqrt{\log^3(\sigma)/\sigma})d\xi+O(1/(|S|\sigma))\\
& \leq \det(\Sigma+I)^{-1/2}(C^{3/2}\sqrt{\log(\sigma)/\sigma})O(C\log(\sigma))^{k/2}.
\end{align*}
Therefore, the sum of $|X(p)-G(p)|$ over $p\in S$ is at most
$$
(C^{3/2}\sqrt{\log(\sigma)/\sigma})O(C^2\log^2(\sigma))^{k/2}.
$$
The sum over $p \not\in S$ is at most $O(1/\sigma)$. This completes the proof.
\end{proof}

The proof of the main theorem is substantially the same as the above.
The one obstacle that we face is that above we are only able to prove $L^\infty$
bounds on the difference between $X$ and $G,$ and these bounds are too weak for our purposes.
What we would like to do is to prove stronger bounds on the difference between $X$ and $G$
at points $p$ far from $\mu.$ In order to do this, we will need to take advantage
of cancellation in the inverse Fourier transform integrals.
To achieve this, we will use the saddle point method from complex analysis.

\begin{proof}[Proof of Theorem~\ref{thm:clt-formal}.]
For $p\in S$ we have as above that
$$
|X(p)-G(p)| = \left| \int_{\xi\in T,\xi_1\in[0,1]}e(-p\cdot \xi) (\wh{X}(\xi)-\wh{G}(\xi))d\xi \right|+O(1/(|S|\sigma)).
$$
Let $\xi_0\in \R^k$ be such that $\xi_0.\bone=0$ and so that $\Sigma \xi_0 = (\mu-p)/(2\pi)$
(i.e., take $\xi_0 = (\Sigma+\bone \bone^T)^{-1} \mu-p)/(2\pi)$).
We think of the integral above as an iterated contour integral.
By deforming the contour associated with the innermost integral, we claim that it is the same as the sum of the integrals
over $\xi$ with $\re(\xi)\in T,\re(\xi_1)\in[0,1]$ and $\im(\xi)=\xi_0$
and the integral over $\re(\xi)\in \delta T,\re(\xi_1)\in[0,1]$
and $\im(\xi)=t\xi_0$ for some $t\in [0,1]$
(the extra pieces that we would need to add at $\re(\xi_1)=0$ and $\re(\xi_1)=1$ cancel out).

\begin{claim}
 $\int_{\xi\in \delta T,\xi_1\in[0,1]}e(-p\cdot \xi) (\wh{X}(\xi)-\wh{G}(\xi))d\xi$ equals $$\int_{\xi\in T,\xi_1\in[0,1]}e(-p\cdot (\xi+i\xi_0)) (\wh{X}(\xi+i\xi_0)-\wh{G}(\xi+i\xi_0))d\xi$$ plus
$$\int_{\xi\in \delta T,\xi_1\in[0,1]} \int_{t=0}^1 e(-p\cdot (\xi+it\xi_0)) (\wh{X}(\xi+it\xi_0)-\wh{G}(\xi+it\xi_0)) d(t\xi_0) \cdot d\xi.$$
\end{claim}
\begin{proof}
We write $f(\xi)= e(-p\cdot \xi) (\wh{X}(\xi)-\wh{G}(\xi)).$
Let $O$ be an orthogonal matrix with $k$th column $\xi_0/\|\xi_0\|_2$.
Then, we change variables from $\xi$ to $\nu = O^T \xi$, yielding
$$ \int_{\xi\in T'}f(\xi) d\xi = \int_{\nu \in O^T T'} f(O^T \nu) d\nu $$
We can consider this as an iterated integral where $\nu_i$ is integrated from $a_i(\nu_1,\ldots, \nu_{i-1})$ to $b_i(\nu_1,\ldots,\nu_{i-1}).$
$$\int_{\nu \in O^T T'} f(O^T \nu) d\nu = \int_{a_1}^{b_1} \int_{a_2(\nu_1)}^{b_2(\nu_1)} \ldots \int_{a_k(\nu_1,\ldots,\nu_{k-1})}^{b_k(\nu_1,\ldots,\nu_{k-1})} f(O^T \nu) d\nu_k \ldots d\nu_2 d\nu_1 \;.$$
We consider the innermost integral. The function  $f(O^T \nu)$ is a linear combination of exponentials and so is holomorphic on all of $\C^n.$
Let $\mathcal{C}$ be the contour which consists of three straight lines,
from $a_k(\nu_1,\ldots,\nu_{k-1})$ via $a_k(\nu_1,\ldots,\nu_{k-1})+i\|\xi_0\|_2$ and $b_k(\nu_1,\ldots,\nu_{k-1})+i\|\xi_0\|_2$ to $b_k(\nu_1,\ldots,\nu_{k-1}).$
Then, by standard facts of complex analysis, we have:
\begin{align*}
& \int_{a_k(\nu_1,\ldots,\nu_{k-1})}^{b_k(\nu_1,\ldots,\nu_{k-1})} f(O^T \nu) d\nu_k \\
= & \int_\mathcal{C} f(O^T \nu) d\nu_k \\
= & \int_0^1 f(O^T (\nu_1,\ldots, \nu_{k-1}, a_k(\nu_1,\ldots,\nu_{k-1})+i\|\xi_0\|_2 t)) i  \|\xi_0\|_2 dt \\
+ & \int_{a_k(\nu_1,\ldots,\nu_{k-1})}^{b_k(\nu_1,\ldots,\nu_{k-1})} f(O^T (\nu + i\|\xi_0\|_2 e_k)) d\nu_k \\
+ & \int_0^1 f(O^T (\nu_1,\ldots, \nu_{k-1}, b_k(\nu_1,\ldots,\nu_{k-1})+i\|\xi_0\|_2 (1-t'))) i  \|\xi_0\|_2 dt'
\end{align*}
The middle part of this path gives the first term in the statement of the claim:
\begin{align*}
& \int_{a_1}^{b_1} \int_{a_2(\nu_1)}^{b_2(\nu_1)} \ldots \int_{a_k(\nu_1,\ldots,\nu_{k-1})}^{b_k(\nu_1,\ldots,\nu_{k-1})} f(O^T (\nu + i\|\xi_0\|_2 e_k)) d\nu_k \ldots d\nu_2 d\nu_1 \\
= & \int_{\nu \in O^T T'} f(O^T (\nu+i\|\xi_0\|_2 e_k) ) \\
= & \int_{\xi \in T'} f(\xi + i \xi_0) \;.
\end{align*}
A change of variables allows us to express the sum of the contributions from the first and third part of the path:
\begin{align*}
& \int_0^1 f(O^T (\nu_1,\ldots, \nu_{k-1}, a_k(\nu_1,\ldots,\nu_{k-1})+i\|\xi_0\|_2 t)) i  \|\xi_0\|_2 dt \\
- & \int_0^1 f(O^T (\nu_1,\ldots, \nu_{k-1}, b_k(\nu_1,\ldots,\nu_{k-1})+i\|\xi_0\|_2 t')) i  \|\xi_0\|_2 dt' \;.
\end{align*}
Changing variables to replace $(\nu_1,\ldots, \nu_{k-1}, a_k(\nu_1,\ldots,\nu_{k-1}))$ 
or $(\nu_1,\ldots, \nu_{k-1}, b_k(\nu_1,\ldots,\nu_{k-1}))$ with $\xi\in \delta T$ or $\xi \in T \cap \{0,1\}$ 
we get an appropriate integral of $\pm i f(\xi+it\xi_0).$ 
We note that the volume form for $\xi_0$ assigns to a surface element 
the volume of the projection of that element in the $\xi_0$ direction. 
Multiplying by $\|\xi_0\|_2$ and the appropriate sign yields exactly 
the measure $\xi_0\cdot d\xi.$ Thus, we are left with an integral of 
$f(\xi+it\xi_0)d(t\xi_0)\cdot d\xi.$ However, it should be noted that the measures 
$\xi_0\cdot d\xi$ are opposite on $\xi_1=0$ and $\xi_1=1$ boundaries (as $d\xi$ is the outward pointing normal). 
Since $f(\xi+it\xi_0) = f(\xi+\bone+it\xi_0),$ the integrals over these regions cancel, leaving exactly with the claimed integral.
\end{proof}

In order to estimate this difference, we use Lemma \ref{FourierApproxLem}, which still applies. 
Furthermore, we note that
$$
\xi_0\cdot\Sigma\xi_0 = (p-\mu)\cdot \Sigma^{-1}(p-\mu)/(4\pi^2) = O(Ck\log(\sigma)) \;,
$$
because $p\in S.$

Therefore, we have that $|X(p)-G(p)|$ is $O(1/(|S|\sigma))$ plus
$$
O( k^{7/2} \sqrt{\log^3(\sigma)/\sigma})\int_{\mathcal{C}} \left|\exp\left(2\pi i(\mu-p)\cdot \xi -2\pi^2 \xi \cdot \Sigma \xi \right)\right|d\xi.
$$
Now, when $\im(\xi)=\xi_0,$ we have that
$$
\exp\left(2\pi i(\mu-p)\cdot \xi -2\pi^2 \xi \cdot \Sigma \xi \right) = \exp\left(-(p-\mu)\Sigma^{-1}(p-\mu)/2-2\pi^2 \re(\xi)\cdot\Sigma \re(\xi) \right).
$$
Integrating, we find that the difference over this region is at most times
$$
O( k^{7/2} \sqrt{\log^3(\sigma)/\sigma})\int \exp\left(-(p-\mu)\Sigma^{-1}(p-\mu)/2-2\pi^2 \re(\xi)\cdot\Sigma \re(\xi) \right)d\xi = O( k^{7/2} \sqrt{\log^3(\sigma)/\sigma})G(p).
$$
The contribution from the part of the contour where $\re(\xi)$ is on the boundary of $T$ 
is also easy to bound after noting that both $|\wh{X}(\xi)|$ and $|\wh{G}(\xi)|$ are $O(\sigma^{-k}).$ 
We furthermore claim that the total volume of the region of integration is $O(\sqrt{k}\textrm{Vol}(T)).$ 
Together, these would imply that the total integral over this region is $O(1/(\sigma|S|)).$ 
To do this, we note that the total volume of the region being integrated over
is at most the volume of the projection of $T$ in the direction perpendicular to $\xi_0$ 
times the length of $\xi_0.$ In order to analyze this we consider each slice of $T$ 
given by $\xi\cdot \bone = \alpha$ separately. Noting that $|\alpha|\leq 2$ for all $\xi \in \delta T,$ 
it suffices to consider only a single slice. In particular, since for all such $\alpha,$ 
we have that $\tilde \xi\in T,$ it suffices to consider the slice $\alpha=0.$ 
Along this slice, we have that $T$ is an ellipsoid. 
Note that $\xi_0 \cdot \Sigma \xi_0 = (p-\mu)\cdot (\Sigma +\bone \bone^T)^{-1} (p-\mu) \leq Ck\log(\sigma).$ Therefore $\xi_0\in T.$

Next, we claim that if $E$ is any ellipsoid in at most $k$ dimensions,
and if $v$ is a vector with $v\in E,$
then the product of the length of $v$ times the volume of the projection of $E$
perpendicular to $v$ is at most $O(\sqrt{k}\textrm{Vol}(E)).$
This follows after noting that the claim is invariant under affine transformations,
and thus it suffices to consider $E$ the unit ball for which it is easy to verify.

Therefore, for $p\in S$, we have that
$$
|X(p)-G(p)|\leq O(1/|S|\sigma)+O( k^{7/2} \sqrt{\log^3(\sigma)/\sigma})G(p).
$$
From this it is easy to see that it is also
$$
|X(p)-G(p)|\leq O(1/|S|\sigma)+O( k^{3/2} \sqrt{\log^3(\sigma)/\sigma})X(p).
$$
Summing over $p\in S$ gives a total difference of at most
$$
O( k^{7/2} \sqrt{\log^3(\sigma)/\sigma}).
$$
Combining this with the fact that the sum of $X(p)$ and $G(p)$ for $p$ not in $S$ is at most $1/\sigma$ gives us that
$$
\dtv(X,G) = O( k^{7/2} \sqrt{\log^3(\sigma)/\sigma}).
$$
This completes the proof of Theorem~\ref{thm:clt-formal}.
\end{proof}

\paragraph{Comparison to the ~\cite{VV10b} CLT.}
We note that the above statement of Theorem~\ref{thm:clt-formal}
is not immediately comparable to the CLT of~\cite{VV10b}.
More specifically, we work with PMDs directly, while~\cite{VV10b} works with projections of PMDs onto $k-1$ coordinates.
Also, our notion of a discrete Gaussian is not the same as the one discussed in~\cite{VV10b}. 
However, it is not difficult to relate the two results.
First, we need to relate our PMD (supported on integer vectors whose coordinates sum to $n$) 
to theirs (which are projections of PMDs onto $k-1$ coordinates). 
In particular, we need to show that this projection does not skew minimum eigenvalue in the wrong direction. 
This is done in the following simple proposition:

\begin{proposition} \label{prop:proj-covariance}
Let $X$ be an $(n, k)$-PMD, and $X'$ be obtained by projecting $X$ onto its first $k-1$ coordinates. 
Let $\Sigma$ and $\Sigma'$ be the covariance matrices of $X$ and $X',$ respectively, 
and let $\sigma$ and $\sigma'$ be the second smallest and smallest eigenvalues respectively of $\Sigma$ and $\Sigma'.$ 
Then,  we have $\sigma \geq \sigma'.$
\end{proposition}
\begin{proof}
Note that, since $\mathbf{1}$ is in the kernel of $\Sigma,$ 
$\sigma$ is the minimum of $v$ orthogonal to $\mathbf{1}$ of $\frac{v^T\Sigma v}{v^Tv}.$ 
Whereas, $\sigma'$ is the minimum over $w\in \R^{k-1}\backslash \{\mathbf{0}\}$ 
of $\frac{w^T\Sigma' w}{w^Tw}.$ This is the same as the minimum over $w$ in $\R^k$ 
with $k^{th}$ coordinate equal to $0$ of $\frac{w^T\Sigma w}{w^Tw}.$

Let the minimization problem defining $\sigma$ be obtained by some particular $v$ 
orthogonal to $\mathbf{1}.$ In particular, a $v$ so that $\sigma=\frac{v^T\Sigma v}{v^Tv}.$ 
Let $w$ be the unique vector of the form $v+a\mathbf{1}$ so that $w$ has last coordinate $0.$ 
Then, we have that
\begin{align*}
\sigma' \geq  \frac{w^T\Sigma w}{w^Tw}
=  \frac{v^T\Sigma v}{w^Tw}
\geq  \frac{v^T\Sigma v}{v^Tv}
=  \sigma.
\end{align*}
This completes the proof.
\end{proof}

Next, we need to relate the two slightly different notions of discrete Gaussian.
\begin{proposition} \label{prop:our-discrete-Gaussians-are-different}
Let $G$ be a Gaussian in $\R^k$ with covariance matrix $\Sigma,$ 
which has no eigenvalue smaller than $\sigma.$ 
Let $G'$ be the discrete Gaussian obtained by rounding the values of $G$ to the nearest lattice point. 
Let $G''$ be the discrete distribution obtained by assigning each integer lattice point 
mass proportional to the probability density function of $G.$ Then, we have that
$$ \dtv(G', G'') \leq O(k\sqrt{\log(\sigma)/\sigma}).$$
\end{proposition}
\begin{proof}
We note that the probability density function of $G$ 
is proportional to $\exp(-(x\cdot \Sigma^{-1} x)/2) dx.$ 
Suppose that $y$ is another vector with $\|x-y\|_\infty < 1.$ 
We would like to claim that the probability density function at $y$ 
is approximately the same as at $x.$ In particular, we write $y=x+z$ and note that
\begin{align*}
y\cdot \Sigma^{-1} y & = x\cdot \Sigma^{-1} x + 2 z\cdot \Sigma^{-1} x + z\cdot \Sigma^{-1} z\\
& = x\cdot \Sigma^{-1} x + 2 (\Sigma^{-1/2}x)\cdot(\Sigma^{-1/2} z) + O(|z|_2^2 \sigma^{-1})\\
& = x\cdot \Sigma^{-1} x + O( |\Sigma^{-1/2} x|_2 \sqrt{k/\sigma} + k \sigma^{-1})\\
& = x\cdot \Sigma^{-1} x + O(  \sqrt{(k/\sigma)x\cdot \Sigma^{-1} x} + k \sigma^{-1}).
\end{align*}
Note that, for lattice points $x,$ $G'(x)$ is the average over $y$ in a unit cube 
about $x$ of the pdf of $G$ at $y,$ while $G''(x)$ is just the pdf of $G$ at $x.$ 
These quantities are within a $1+O(\sqrt{(k/\sigma)x\cdot \Sigma^{-1} x} + k \sigma^{-1})$ 
multiple of each other by the above so long as the term in the ``$O$'' is $o(1).$ 
Therefore, for all $x$ with $x\cdot \Sigma^{-1} x \ll k \log(\sigma),$ 
we have that $G'(x) = G''(x)(1+O(k\sqrt{\log(\sigma)/\sigma})).$ 
We note however that $G'$ has only a $1/\sigma$ probability of $x$ being outside of this range. 
{Furthermore, we claim that $G''(x)=O(G'(x))$ for all $x.$ To see this, note that 
for any $v$ with $\|v\|_\infty \leq 1/2,$ we have 
$$G(x) = e^{-v^T \Sigma^{-1} v/2}\sqrt{G(x+v)G(x-v)} \leq e^{-(k/2\sigma)}\max\{G(x+v),G(x-v)\}.$$ 
We assume that $\sigma \geq k^2$ or else we have nothing to prove. 
Then, we have $G(x) = O(G(x+v)+G(x-v)),$ and by considering the integral that defines $G',$ 
we have $G''(x)=O(G'(x)).$ Thus, $G''$ similarly has $O(1/\sigma)$ mass outside 
of the range $x\cdot \Sigma^{-1} x \ll k \log(\sigma)$.} Therefore, the $L_1$ difference 
inside the range is $O(k\sqrt{\log(\sigma)/\sigma})$ and the $L_1$ error from outside is {$O(1/\sigma).$} This completes the proof.
\end{proof}

Armed with these propositions, we have the following corollary of Theorem~\ref{thm:clt-formal}:
\begin{corollary}
Let $X$ be an $(n, k)$-PMD, and $X'$ be obtained by projecting $X$ onto its first $k-1$ coordinates. 
Let $\Sigma'$  be the covariance matrix of $X'.$ Suppose that $\Sigma'$ has no eigenvectors with eigenvalue less than $\sigma'.$
Let $G'$ be the distribution obtained by sampling from $\mathcal{N}(\E[X'], \Sigma')$ and rounding to the nearest point in $\Z^k.$ 
Then, we have that
$$
\dtv(X',G') \leq O( k^{7/2} \sqrt{\log^3(\sigma')/\sigma'}).
$$
\end{corollary}
\begin{proof} 
Let $\Sigma$ be the covariance matrix of $X.$ 
Since $X$ is a PMD, $\bone$ is an eigenvector of $\Sigma$ with eigenvalue $0.$ 
By Proposition \ref{prop:proj-covariance}, the other eigenvalues of $\Sigma$ are at least $\sigma'.$ 
Theorem~\ref{thm:clt-formal} now yields that $\dtv(X,G) \leq O( k^{7/2} \sqrt{\log^3(\sigma')/\sigma'}),$ 
where $G$ is a discrete Gaussian in $k$ dimensions (as defined in the context of the theorem statement). 
Let $G''$ be the discrete Gaussian obtained by projecting $G$ onto the first $k-1$ coordinates. 
Then, we have that $\dtv(X',G'') \leq O( k^{7/2} \sqrt{\log^3(\sigma')/\sigma'}).$
From the proof of Theorem~\ref{thm:clt-formal}, $G$ is proportional to the pdf of $\mathcal{N}(\E[X],\Sigma).$
Note that $G''$ is proportional to the pdf of $\mathcal{N}(\E[X'], \Sigma').$  
Then, by Proposition~\ref{prop:our-discrete-Gaussians-are-different}, 
it follows that $\dtv(G',G'') \leq O(k\sqrt{\log(\sigma')/\sigma'}).$ 
So, by the triangle inequality, we have $\dtv(X',G') \leq O( k^{7/2} \sqrt{\log^3(\sigma')/\sigma'}),$ as required. 
\end{proof}

\section{Conclusions and Open Problems} \label{sec:conclusions}

In this work, we used Fourier analytic techniques to obtain a number of structural results
on PMDs.
As a consequence, 
we gave a number of applications in distribution learning, statistics, and game theory.
We believe that our techniques are of independent interest and may find other applications.

Several interesting open questions remain: 
\begin{itemize}

\item What is the precise complexity of learning PMDs? Our bound is nearly-optimal when the dimension $k$
is fixed. The case of high dimension is not well-understood, and seems to require different ideas.

\item Is there an efficient {\em proper} learning algorithm, 
i.e., an algorithm that outputs a PMD as its hypothesis? 
This question is still open even for $k=2$; 
see~\cite{DKS15b} for some recent progress.

\item What is the optimal error dependence in Theorem~\ref{thm:clt} as a function of the dimension 
$k$? 

\item Is there a fully-polynomial time approximation scheme (FPTAS) for computing $\eps$-Nash equilibria
in anonymous games? We remark that cover-based algorithms cannot lead to such a result, because
of the quasi-polynomial cover size lower bounds in this paper, as well as in our previous work~\cite{DKS15} for the case $k=2.$
Progress in this direction requires a deeper understanding of the relevant fixed points.

\end{itemize}

\bibliographystyle{alpha}

\bibliography{allrefs}

\newcommand{\etalchar}[1]{$^{#1}$}
\begin{thebibliography}{DDO{\etalchar{+}}13}

\bibitem[Bar88]{Barbour88}
A.~D. Barbour.
\newblock Stein's method and poisson process convergence.
\newblock {\em Journal of Applied Probability}, 25:pp. 175--184, 1988.

\bibitem[BCI{\etalchar{+}}08]{BorgsCIKMP08}
C.~Borgs, J.~T. Chayes, N.~Immorlica, A.~T. Kalai, V.~S. Mirrokni, and C.~H.
  Papadimitriou.
\newblock The myth of the folk theorem.
\newblock In {\em STOC}, pages 365--372, 2008.

\bibitem[BDS12]{BDS12}
A.~Bhaskara, D.~Desai, and S.~Srinivasan.
\newblock Optimal hitting sets for combinatorial shapes.
\newblock In {\em 15th International Workshop, {APPROX} 2012, and 16th
  International Workshop, {RANDOM} 2012}, pages 423--434, 2012.

\bibitem[Ben03]{Bentkus:03}
V.~Bentkus.
\newblock {On the dependence of the Berry-Esseen bound on dimension}.
\newblock {\em Journal of Statistical Planning and Inference}, 113:385--402,
  2003.

\bibitem[BHJ92]{BHJ:92}
A.D. Barbour, L.~Holst, and S.~Janson.
\newblock {\em Poisson Approximation}.
\newblock Oxford University Press, New York, NY, 1992.

\bibitem[Blo99]{Blonski1999}
M.~Blonski.
\newblock Anonymous games with binary actions.
\newblock {\em Games and Economic Behavior}, 28(2):171 -- 180, 1999.

\bibitem[Blo05]{Blonski2005}
M.~Blonski.
\newblock The women of cairo: Equilibria in large anonymous games.
\newblock {\em Journal of Mathematical Economics}, 41(3):253 -- 264, 2005.

\bibitem[CDO15]{CDO15}
X.~Chen, D.~Durfee, and A.~Orfanou.
\newblock On the complexity of nash equilibria in anonymous games.
\newblock In {\em STOC}, 2015.

\bibitem[CST14]{ChenST14}
X.~Chen, R.~A. Servedio, and L.Y. Tan.
\newblock New algorithms and lower bounds for monotonicity testing.
\newblock In {\em FOCS}, pages 286--295, 2014.

\bibitem[DDKT16]{DDKT15}
C.~Daskalakis, A.~De, G.~Kamath, and C.~Tzamos.
\newblock A size-free {CLT} for poisson multinomials and its applications.
\newblock In {\em Proceedings of STOC'16}, 2016.

\bibitem[DDO{\etalchar{+}}13]{DDOST13focs}
C.~Daskalakis, I.~Diakonikolas, R.~O'Donnell, R.A. Servedio, and L.~Tan.
\newblock {Learning Sums of Independent Integer Random Variables}.
\newblock In {\em FOCS}, pages 217--226, 2013.

\bibitem[DDS12]{DDS12stoc}
C.~Daskalakis, I.~Diakonikolas, and R.A. Servedio.
\newblock {Learning Poisson Binomial Distributions}.
\newblock In {\em STOC}, pages 709--728, 2012.

\bibitem[De15]{De15}
A.~De.
\newblock Beyond the central limit theorem: asymptotic expansions and
  pseudorandomness for combinatorial sums.
\newblock In {\em FOCS}, 2015.

\bibitem[DKS15a]{DKS15}
I.~Diakonikolas, D.~M. Kane, and A.~Stewart.
\newblock Optimal learning via the fourier transform for sums of independent
  integer random variables.
\newblock {\em CoRR}, abs/1505.00662, 2015.
\newblock To appear in COLT 2016.

\bibitem[DKS15b]{DKS15b}
I.~Diakonikolas, D.~M. Kane, and A.~Stewart.
\newblock Properly learning poisson binomial distributions in almost polynomial
  time.
\newblock {\em CoRR}, 2015.
\newblock To appear in COLT 2016.

\bibitem[DKT15]{DKT15}
C.~Daskalakis, G.~Kamath, and C.~Tzamos.
\newblock On the structure, covering, and learning of poisson multinomial
  distributions.
\newblock In {\em FOCS}, 2015.

\bibitem[DP07]{DaskalakisP07}
C.~Daskalakis and C.~H. Papadimitriou.
\newblock Computing equilibria in anonymous games.
\newblock In {\em FOCS}, pages 83--93, 2007.

\bibitem[DP08]{DaskalakisP08}
C.~Daskalakis and C.~H. Papadimitriou.
\newblock Discretized multinomial distributions and nash equilibria in
  anonymous games.
\newblock In {\em FOCS}, pages 25--34, 2008.

\bibitem[DP09]{DaskalakisP09}
C.~Daskalakis and C.~Papadimitriou.
\newblock {On Oblivious PTAS's for Nash Equilibrium}.
\newblock In {\em STOC}, pages 75--84, 2009.

\bibitem[DP14]{DaskalakisP2014}
C.~Daskalakis and C.~H. Papadimitriou.
\newblock Approximate nash equilibria in anonymous games.
\newblock {\em Journal of Economic Theory}, 2014.

\bibitem[GKM15]{GKM15}
P.~Gopalan, D.~M. Kane, and R.~Meka.
\newblock Pseudorandomness via the discrete fourier transform.
\newblock In {\em FOCS}, 2015.

\bibitem[GMRZ11]{GMRZ11}
P.~Gopalan, R.~Meka, O.~Reingold, and D.~Zuckerman.
\newblock Pseudorandom generators for combinatorial shapes.
\newblock In {\em STOC}, pages 253--262, 2011.

\bibitem[GRW15]{GRW15}
P.~Gorlach, C~Riener, and T.~Weisser.
\newblock Deciding positivity of multisymmetric polynomials.
\newblock {\em Journal of Symbolic Computation}, 2015.
\newblock Also available as arxiv report http://arxiv.org/abs/1409.2707.

\bibitem[GT14]{GT14}
P.~W. Goldberg and S.~Turchetta.
\newblock Query complexity of approximate equilibria in anonymous games.
\newblock {\em CoRR}, abs/1412.6455, 2014.

\bibitem[HJ85]{HornJohnson:85}
R.~A. Horn and C.~R. Johnson.
\newblock {\em Matrix Analysis}.
\newblock Cambridge University Press, 1985.

\bibitem[Loh92]{Loh92}
W.~Loh.
\newblock Stein's method and multinomial approximation.
\newblock {\em Ann. Appl. Probab.}, 2(3):536--554, 08 1992.

\bibitem[Mil96]{Milchtaich1996}
I.~Milchtaich.
\newblock Congestion games with player-specific payoff functions.
\newblock {\em Games and Economic Behavior}, 13(1):111 -- 124, 1996.

\bibitem[PC99]{PCZ98}
V.~Y. Pan and Z.~Q. Chen.
\newblock The complexity of the matrix eigenproblem.
\newblock In {\em Proceedings of the Thirty-first Annual ACM Symposium on
  Theory of Computing}, pages 507--516, 1999.

\bibitem[Poi37]{Poisson:37}
S.D. Poisson.
\newblock {\em Recherches sur la Probabilit\`{e} des jugements en mati\'{e}
  criminelle et en mati\'{e}re civile}.
\newblock Bachelier, Paris, 1837.

\bibitem[Roo99]{Roos99}
B.~Roos.
\newblock On the rate of multivariate poisson convergence.
\newblock {\em Journal of Multivariate Analysis}, 69(1):120 -- 134, 1999.

\bibitem[Roo02]{Roos02}
B.~Roos.
\newblock Multinomial and krawtchouk approximations to the generalized
  multinomial distribution.
\newblock {\em Theory of Probability \& Its Applications}, 46(1):103--117,
  2002.

\bibitem[Roo10]{Roos10}
B.~Roos.
\newblock Closeness of convolutions of probability measures.
\newblock {\em Bernoulli}, 16(1):23--50, 2010.

\bibitem[Sto96]{storjohann96}
A.~Storjohann.
\newblock Near optimal algorithms for computing smith normal forms of integer
  matrices.
\newblock In {\em Proceedings of the 1996 international symposium on Symbolic
  and algebraic computation}, pages 267--274, 1996.

\bibitem[Sto00]{storjohannthesis}
A.~Storjohann.
\newblock {\em Algorithms for matrix canonical forms}.
\newblock PhD thesis, Diss., Technische Wissenschaften ETH Z{\"u}rich, Nr.
  13922, 2001, 2000.

\bibitem[Val08]{Valiant08stoc}
P.~Valiant.
\newblock Testing symmetric properties of distributions.
\newblock In {\em STOC}, pages 383--392, 2008.

\bibitem[VV10]{VV10b}
G.~Valiant and P.~Valiant.
\newblock A {CLT} and tight lower bounds for estimating entropy.
\newblock {\em Electronic Colloquium on Computational Complexity {(ECCC)}},
  17(179), 2010.

\bibitem[VV11]{ValiantValiant:11}
G.~Valiant and P.~Valiant.
\newblock {Estimating the unseen: an $n/\log(n)$-sample estimator for entropy
  and support size, shown optimal via new CLTs}.
\newblock In {\em STOC}, pages 685--694, 2011.

\end{thebibliography}


\appendix

\section*{Appendix}

\section{Proof of Lemma \ref{lem:sample-mean-and-covariance}} \label{app:mean-cov-est}

Lemma~\ref{lem:sample-mean-and-covariance} follows directly from the following statement:

\begin{lemma}\label{lem:copy-paste} 
If we take $O(k^4/\eps^2)$ samples from an  $(n, k)$-PMD and let $\wh{\mu}$ and $\wh{\Sigma}$ 
be the sample mean and sample covariance matrix, 
then with probability $19/20,$ for any $y \in \R^k,$ we have:
$$|y^T(\wh{\mu}-\mu)| \leq \eps \sqrt{y^T(\Sigma+I)y} \;,$$
and
$$|y^T(\wh{\Sigma}-\Sigma)y| \leq \eps y^T (\Sigma+I) y \;.$$
\end{lemma}

The above lemma and its proof follow from a minor modification of an analogous lemma in 
~\cite{DKT15}. We include the proof here for the sake of completeness.
{We will use the following simple lemma:
\begin{lemma}[Lemma 21 from~\cite{DKT15}] \label{lem:more-copy-paste} For any vector $y \in \R^k,$ 
given sample access to an $(n,k)$-PMD $\p$ with mean $\mu$ and covariance
matrix $\Sigma$, there exists an algorithm which can produce estimates $\wh{\mu}$ and $\wh{\Sigma},$ 
such that with probability at least $19/20$:
$|y^T (\wh{\mu}-\mu)| \leq \eps \sqrt{y^T \Sigma y}$ and $|y^T(\wh{\Sigma}-\Sigma)y| \leq \eps y^T \Sigma y \sqrt{1+\frac{y^T y}{y^T \Sigma y}}.$
The sample and time complexity are $O(1/\eps^2 ).$
\end{lemma}
}
\begin{proof}[Proof of Lemma \ref{lem:copy-paste}]  
The proof will follow by applying Lemma \ref{lem:more-copy-paste}  
to $k^2$ carefully chosen vectors simultaneously using the union bound. 
Using the resulting guarantees, we show that the same estimates hold for
any direction, at a cost of rescaling $\eps$ by a factor of $k.$
Let $S$ be the set of $k^2$ vectors $\{v_i \},$ for $1 \leq i \leq k,$ 
and $\{ \frac{1}{\sqrt{\lambda_i+1}}v_i+\frac{1}{\sqrt{\lambda_j}}v_j \},$ 
for each $i \not=j,$ where the $v _i$'s are an orthonormal eigenbasis 
for $\Sigma$ with eigenvalues $\lambda_i.$ From Lemma~\ref{lem:more-copy-paste} 
and a union bound, with probability $9/10,$ for all $y \in S$, we have
$$|y^T(\wh{\mu}-\mu)| \leq (\eps/k) \sqrt{y^T \Sigma y} \;,$$
and
$$|y^T(\wh{\Sigma}-\Sigma)y| \leq (\eps/3k) (y^T \Sigma y) \sqrt{1+\frac{y^T y}{y^T \Sigma y}} \; .$$
We claim that the latter implies that:
$$|y^T(\wh{\Sigma}-\Sigma)y| \leq (\eps/3k) (y^T (\Sigma+I) y) \; .$$
Note that if $y^T \Sigma y=0,$ we must have $y^T \wh{\Sigma} y = 0,$ 
since then $y^T X$ is a constant for a PMD random variable $X.$ Otherwise,
$$(y^T \Sigma y) \sqrt{1+\frac{y^T y}{y^T \Sigma y}} = \sqrt{(y^T \Sigma y)(y^T \Sigma y + y^y y)} = \sqrt{(y^T \Sigma y)(y^T (\Sigma+I) y)} \leq (y^T (\Sigma+I) y) \;.$$
The claim about the accuracy of $\wh{\Sigma}$ now follows from Lemma~\ref{lem:more-copy-paste}.

We now prove that the mean estimator $\wh{\mu}$ is accurate. Consider an arbitrary vector $y$, which
can be decomposed into a linear composition of the eigenvectors $y = \sum_i \alpha_i  v_i.$

Then,
$$y^T (\wh{\mu}-\mu) = \sum_i \alpha_i v_i^T (\wh{\mu}-\mu) \leq (\eps/k) \sum_i |\alpha_i| \sqrt{\lambda_i + 1} \leq (\eps/k) \sqrt{k} \sqrt{k \sum_i \alpha_i^2 (\lambda_i+1)} \;,$$
but $\sum_i \alpha_i^2 (\lambda_i+1) =y^T (\Sigma+I) y,$ 
so we have $y^T (\wh{\mu}-\mu) \leq \eps \sqrt{y^T(\Sigma+I)y}$ as required.
\end{proof}

We are now ready to complete the proof of the desired lemma.

\medskip

\noindent {\bf Lemma~\ref{lem:sample-mean-and-covariance}.}
{\em
With probability $19/20,$ we have that $(\wh{\mu}-\mu)^T(\Sigma+I)^{-1}(\wh{\mu}-\mu) = O(1)$, $2(\Sigma+I) \geq \wh{\Sigma}+I \geq (\Sigma+I)/2.$
}

\smallskip

\begin{proof} We apply Lemma \ref{lem:copy-paste} with $\eps:=1/2.$ 
For all $y$, we have $|y^T(\wh{\Sigma}-\Sigma)y| \leq \eps y^T (\Sigma+I) y,$ that is
$$\frac{1}{2} y^T (\Sigma+I) y \leq y^T (\wh{\Sigma}+I) y \leq \frac{3}{2} y^T (\Sigma+I) y \; .$$
Thus, we have $\frac{1}{2}(\Sigma+I) \leq \wh{\Sigma}+I \leq \frac{3}{2} (\Sigma+I)$ as required.

Note that since $\Sigma+I$ is positive definite, it is non-singular. 
Setting $y=\frac{1}{(\wh{\mu}-\mu)^T(\Sigma+I)^{-1}(\wh{\mu}-\mu)} (\Sigma+I)^{-1}  (\wh{\mu}-\mu),$ 
we have $y^T \wh{\mu}-\mu)=1$ and 
$y^T (\Sigma+I) y = 1/(\wh{\mu}-\mu)^T(\Sigma+I)^{-1}(\wh{\mu}-\mu).$ 
So, Lemma \ref{lem:copy-paste} gives us:
$$|y^T(\wh{\mu}-\mu)| \leq \frac{1}{2} \sqrt{y^T(\Sigma+I)y}.$$
Substituting the above:
$$1 \leq \frac{1}{2} \sqrt{1/(\wh{\mu}-\mu)^T(\Sigma+I)^{-1}(\wh{\mu}-\mu)} \; .$$
Therefore, we have $(\wh{\mu}-\mu)^T(\Sigma+I)^{-1}(\wh{\mu}-\mu) \leq 1/4,$ as required.
\end{proof}

\end{document}